\pdfoutput=1
\documentclass[a4paper,twocolumn,11pt,accepted=2026-10-16]{quantumarticle}

\pdfoutput=1

\usepackage{amsthm, amsmath, amssymb, amsfonts}
\usepackage[numbers,sort&compress]{natbib}
\usepackage{dsfont}
\usepackage{mathtools}
\usepackage{physics}
\usepackage{xcolor}
\usepackage{graphicx}
\usepackage{adjustbox}
\usepackage{placeins}
\usepackage[T1]{fontenc}
\usepackage{lipsum}
\usepackage[utf8]{inputenc}
\usepackage{subcaption}
\usepackage{csquotes}
\usepackage[english]{babel}

\makeatletter
\AtBeginDocument{%
  \@ifpackageloaded{hyperref}{%
    \hypersetup{
      colorlinks=true,
      linkcolor=quantumviolet,
      citecolor=quantumviolet,
      urlcolor=quantumviolet
    }%
  }{}%
}
\makeatother

\usepackage{hyperref}

\makeatletter
\renewcommand{\@printtitletextwithappropriatefontsize}{%
  \@titleatfontsize{\fontsize{20}{21}\selectfont}%
}
\makeatother

\newtheorem{theorem}{Theorem}
\newtheorem{lemma}[theorem]{Lemma}
\newtheorem{corollary}[theorem]{Corollary}

\usepackage[svgnames]{xcolor}

\newcommand{\oneop}[1]{{#1}}

\newif\ifshowrevisions
\showrevisionsfalse  % Colors OFF
\definecolor{darkgreen}{RGB}{0,120,0}

\newcommand{\revise}[1]{%
    \ifshowrevisions\textcolor{blue}{#1}\else#1\fi
}
\newcommand{\revisered}[1]{%
    \ifshowrevisions\textcolor{red}{#1}\else#1\fi
}
\newcommand{\revisegreen}[1]{%
    \ifshowrevisions\textcolor{darkgreen}{#1}\else#1\fi
}

\begin{document}

\title{Certified bounds on optimization problems in quantum theory}

\author{Younes Naceur}
    \email{Younes.Naceur@icfo.eu}% Your name
    \affiliation{ICFO - Institut de Ciencies Fotoniques, The Barcelona Institute of Science and Technology, Av. Carl Friedrich Gauss 3, 08860 Castelldefels (Barcelona), Spain}
    \affiliation{LAAS-CNRS - Laboratoire d’Analyse et d’Architecture des Systèmes, Centre National de la Recherche Scientifique, Toulouse, France}

\author{Jie Wang}
    \affiliation{State Key Laboratory of Mathematical Sciences, Academy of Mathematics and Systems Science, Chinese Academy of Sciences, Beijing, China}

\author{Victor Magron}
    \affiliation{LAAS-CNRS - Laboratoire d’Analyse et d’Architecture des Systèmes, Centre National de la Recherche Scientifique, Toulouse, France}
    \affiliation{Institute of Mathematics, Toulouse, France}

\author{Antonio Acín}
    \affiliation{ICFO - Institut de Ciencies Fotoniques, The Barcelona Institute of Science and Technology, Av. Carl Friedrich Gauss 3, 08860 Castelldefels (Barcelona), Spain}
    \affiliation{ICREA - Institució Catalana de Recerca i Estudis Avançats, 08010 Barcelona, Spain}

%\date{\today} % Leave empty to omit a date

\begin{abstract}
Semidefinite relaxations of polynomial optimization have become a central tool for addressing the non-convex optimization problems over non-commutative operators that are ubiquitous in quantum information theory and, more in general, quantum physics. Yet, as these global relaxation methods \revisegreen{are solved numerically}, the bounds issued by the semidefinite solver can - and often do - exceed the global optimum, undermining their certifiability. To counter this issue, we introduce a framework for extracting \revisegreen{rigorous certified} bounds on non-commutative optimization problems from \revisegreen{inexact} numerical data, and apply it to several paradigmatic problems in quantum information theory. An extension to sparsity and symmetry-adapted semidefinite relaxations is also provided and compared to the general dense scheme. Our results establish rational post-processing as a practical route to reliable certification, pushing semidefinite optimization toward a certifiable standard for quantum information science.
\end{abstract}


\maketitle

\section{Introduction} \label{sec:outline}

Non-commutative polynomial optimization (NPO) investigates the task of bounding the spectrum of a polynomial in non-commuting operator-valued variables, where said underlying operators are subject to polynomial inequality and equality constraints. As this matches the natural setting of optimization problems in quantum theory, many applications of NPO have been studied therein. Such problem classes comprise, not exclusively, the computation of the maximal quantum violation of a given Bell expression \cite{bounding_quantum}, the ground-state problem of many-body quantum systems~\cite{groundstate}, as well as fundamental problems in causality \cite{araujo_causality}, and contextuality \cite{contextuality}.
Generically, polynomial optimization problems (in commuting or non-commuting variables) are hard to solve exactly, which is why often \textit{relaxations} of these problems are the only tools available. By virtue of relaxations, optimization is performed over a \textit{bigger}, often convex feasible set. Note that relaxation to a \textit{convex} problem is often crucial for its computational tractability, since for these problems a manifold of efficient numerical methods exist. Moreover, they are free from local optima, yielding rigorous lower and upper bounds to minimization and maximization problems, respectively.

The \textit{Navascués-Pironio-Acín} (NPA) hierarchy \cite{navascues_convergent,npa} constitutes a very general relaxation framework based on available convex computational methods that allows one to obtain a monotone series of outer bounds $\{\lambda_d \}_{d=d_{\min}}^{\infty}$ (upper bounds to maximization problems) for a given NPO. The bounds are obtained numerically by solving a series of \textit{Semidefinite Programs} (SDP) of increasing sizes, where under mild conditions this series converges to the global optimum of the problem, $\lambda_{\infty} = \lambda_{\max}$.

This hierarchy has been extensively studied and applied to several emblematic problems in quantum theory. In recent years, several efforts have been made to extend the efficiency, scalability, and applicability of non-commutative semidefinite relaxations. As in the commuting case, the exploitation of the problem's sparsity structure has been shown to significantly improve the scalability of the hierarchy \citep{klep_sparse, wang_sparse}. To complement sparsity exploitation, one can further reduce the size of the underlying optimization problem by exploiting the physical symmetries of the underlying problems \cite{riener_symmetry}. To obtain tighter bounds at a given relaxation order, the use of first-order optimality conditions \citep{fawzi_certified, optimality} has been explored. 
Various extensions of the hierarchy, increasing its applicability, include the trace optimization of non-commutative polynomials  \cite{burgdorf_tracial}, as well as optimization over polynomials in state expectation values \cite{klep_state_2023}. Recently, the framework was also extended to deal with NPO's subject to differential constraints \cite{araujo_differential}.

While the computational efficiency of semidefinite relaxations is one of their principal advantages, it also entails a significant limitation. In practice, the underlying semidefinite programs are solved using floating-point arithmetic, and the resulting feasible points satisfy the positivity and equality constraints of the problem only approximately.
Consequently, the outer bounds obtained from hierarchies such as NPA, which are often described in the literature as \textit{certifying} that a given objective value cannot exceed the numerical bound $\lambda_d$, do \textit{not} constitute formal certificates. As will be illustrated in later sections, the unavoidable rounding errors inherent to numerical solvers can invalidate the positivity or feasibility of the solution. This absence of numerical exactness constitutes an unresolved bottleneck for using semidefinite relaxations as proof systems in quantum theory.

Although obtaining exact, rather than approximate, bounds is an evident step toward mathematical rigor, there are additional, more practical arguments for their necessity: Many optimization protocols, —such as branch-and-bound \cite{branch_and_bound},  or constraint-generation schemes \cite{dynamic_ineq} — rely on iterative schemes, reusing previously computed SDP bounds as input for subsequent computations. These methods typically rely on the exactness of the input bounds; however, if the input bounds are only approximate, the inference chains often lose validity, undermining their claimed guaranties. Similarly, in witness-based scenarios, such as randomness certification or entanglement detection, a numerical upper bound that is even slightly too low can lead to wrong physical conclusions. Establishing exact, algebraically verifiable SDP bounds is therefore a prerequisite for turning semidefinite relaxations into rigorous tools for certification in quantum information science, bridging the gap between efficient numerical computation and mathematical proof. 

To address this issue, we present a three-step post-processing procedure that turns \revisegreen{numerical} SDP outputs into rigorous, \revisegreen{exact} certificates for non-commutative polynomial optimization. The method applies to constrained and unconstrained problems, and is extended to sparsity- and symmetry-adapted scenarios. Starting from a numerical Gram-matrix certificate, we (1) round the \revisegreen{floating-point} solution to \revisegreen{nearby} rational data, (2) project it (Frobenius-optimally) onto the affine subspace enforcing the certificate equalities, and (3) lift the resulting matrix back into the positive-semidefinite (PSD) cone to restore feasibility. \revisegreen{After the rounding step, the subsequent algebraic operations are
performed using exact rational arithmetic, thereby avoiding any further round-off error.} Additionally, when the projected matrix lies strictly inside the PSD cone, our method can be used to further tighten the numerical bound.

The manuscript is organized as follows: Section \ref{sec:semidefinite} introduces the theory of semidefinite relaxations for non-commutative polynomial optimization problems. In Section \ref{sec:rational_bounds} we present our algorithmic contribution, that is a rigorous method of how to obtain certifiable bounds from numerically inaccurate certificates. Section \ref{sec:sparse} extends the certification scheme to the sparsity- and symmetry-adapted version of the hierarchy.
Sections \ref{sec:belllineq} and \ref{sec:qmb} are dedicated to our numerical tests. 
In Section \ref{sec:belllineq}, we apply our method to the maximal violation of a family of bipartite Bell inequalities, as well as a variant of the tilted CHSH inequality, and investigate its performance. Section \ref{sec:qmb} deals with the application of our scheme to certify bounds on quantum many-body ground-state observables. In Section \ref{sec:conclusion} we conclude our work and outline directions for further research.
\section{Semidefinite relaxations for non-commutative optimization problems} \label{sec:semidefinite}
We start by presenting the polynomial optimization problems \revise{analyzed} in this work. Given a separable Hilbert space $\mathcal{H}$, let $\mathcal{B}(\mathcal{H})$ be the set of bounded operators acting on $\mathcal{H}$. 
We consider the problem of finding the  maximal eigenvalue $\lambda_{\max}$ of a polynomial $f(\underline{X})$ with rational coefficients in tuples of bounded, self-adjoint operator-valued variables $\underline{X} = \{X_1, ..., X_n\}$ over a feasible set \revise{$K$} given  by
\begin{equation} \small
\{\underline{X} \subset \mathcal{B}(\mathcal{H}) \, |\, h_i(\underline{X}) = 0, g_i(\underline{X}) \succeq 0;\, h_i\in H,\,g_i\in G\},
\end{equation}
which is characterized by a set of polynomial inequality and equality constraints $G = \{g_i(\underline{X} ) \}_{i=1,..,m_G}$ and $H = \{h_i(\underline{X} )\}_{i=1,..,m_H}$, respectively. The dimension of the variables $\underline{X}$ is unrestricted, allowing for feasible solutions acting on a Hilbert space $\mathcal{H}$ of arbitrary dimension — a key feature of this framework, enabling the treatment of device-independent scenarios in quantum information science.

Without loss of generality, only maximization problems are considered in this section. The generic optimization problem then reads as
\begin{align}
    &\revise{\sup_{\ket{\psi},\mathcal{H},\underline{X}\subset \mathcal{B}(\mathcal{H})} }
    \qquad \bra{\psi} f(\underline{X}) \ket{\psi}  \label{eq:originalprob} \\
    &\qquad\mathrm{s.t.}\qquad \quad \quad \, \, h_i(\underline{X}) = 0, \qquad i=1,...,m_H, \nonumber \\
    &\qquad \qquad \qquad \quad \, \, \, \, \,g_i(\underline{X}) \succeq 0, \qquad i= 1,...,m_G. \nonumber
\end{align}
The optimization can be interpreted as follows: We maximize over all possible (finite and infinite dimensional) Hilbert spaces $\mathcal{H}$, the normalized states therein, $\ket{\psi} \in \mathcal{H}$, as well as the tuples of bounded operators $\underline{X} \subset \mathcal{B}(\mathcal{H})$ acting upon it. This is equivalent to optimizing for the maximal eigenvalue of the polynomial $f(\underline{X})$ s.t. $\underline{X} \in K$. Generically, this problem is non-convex and computationally intractable \cite{laurent2008sums}, which is why studying a relaxed form of the problem is often necessary.

\subsection{Relaxations of polynomial optimization}

One starts by recasting the optimization as a problem of non-negativity
\begin{align}
    &\revisered{\inf_{\lambda \in \mathbb{R}} }
    \qquad \qquad \lambda \label{eq:nonnegativity}\\
    \qquad \, \, \, \,&\mathrm{s.t.}  \qquad \lambda\revise{\oneop{1}} - f(\underline{X}) \succeq 0 \nonumber \quad \revisered{\forall \underline{X} \in K } \nonumber
\end{align}
which is equivalent to \eqref{eq:originalprob} so far: \revise{one rephrases the maximization problem \eqref{eq:originalprob} as a minimization problem over valid upper bounds, where the optima of the two problems coincide. 
Above, we denoted the identity operator by $\oneop{1}$. 
In what follows, by a slight abuse of notation, the same symbol will also be used to denote the constant polynomial 1 and the identity matrix.} 
Now, as the task of characterizing the set of non-negative polynomials $\mathcal{P}_K\langle\underline{X}\rangle$ on a \revise{basic closed} semialgebraic set $K$ is again a computationally intractable task, a version of a central result by Helton and McCullough \cite{helton_positivstellensatz_2004} is used. To do so, we rephrase Problem \eqref{eq:nonnegativity} in the framework of operator theory.

Let $\mathbb{K} \in \{\mathbb{R}, \mathbb{C} \}$ denote the underlying field of real or complex numbers, respectively. We then consider all polynomial objects as members of the $*$-algebra $\mathbb{K}\langle \underline{X} \rangle$ spanned by $\underline{X}$, and equipped with the involution $*$, acting on $\mathbb{K}\langle \underline{X} \rangle$ by reversing the word order of monomials and conjugating their coefficients. 
Given $r \in \mathbb{K}\langle \underline{X} \rangle$, the  polynomial $r^* r$ is called a \textit{Hermitian square}. 
Sums of the form $\sum_i r_i^* r_i$ are called \textit{Sums-of-Hermitian-Squares} (SOHS), and \textit{Sums-of-Squares} (SOS) when $\mathbb{K} = \mathbb{R}$.

In addition, we define the quadratic module $\mathcal{K}$ associated to the feasible set $K$ as
\begin{align}
\mathcal{K} &= \{ p \in \mathbb{K}\langle \underline{X} \rangle \ | \, p = \sum_i r_i^* r_i + \sum_{i,j} u_{ij}^* g_i u_{ij} + \nonumber \\ &\sum_{i,j} v_{ij}^*h_iw_{ij}; \, \, r_i,u_{ij},v_{ij},w_{ij} \in \mathbb{K}\langle \underline{X} \rangle \}. 
\end{align}
A quadratic module is said to fulfill the \textit{Archimedean} property if and only if there exists some $C > 0$ s.t. $C - \sum_{i=1}^n X_i^2 \in \mathcal{K}$. This property guarantees boundedness of the feasible operator tuples, a crucial aspect for the asymptotic convergence of the hierarchy.

\begin{theorem}[Helton–McCullough Positivstellensatz \cite{helton_positivstellensatz_2004}]
\label{thm:heltonmc}
    Let $K$ be a basic \revise{closed} semialgebraic set with associated Archimedean quadratic module $\mathcal{K}$. If $p \in \mathbb{K}\langle \underline{X} \rangle$ is positive definite on $K$, then $p \in \mathcal{K}$. 
\end{theorem}

Using this result, we rephrase the non-negativity condition as
\begin{align}
    &\inf_{\lambda \in \mathbb{R}} 
    \qquad \lambda \label{eq:inftqm} \\
    &\qquad \, \, \, \,\mathrm{s.t.}  \qquad \lambda \revise{\oneop{1}} - f \in \mathcal{K}, \nonumber
\end{align}
\revise{where we passed from explicit operator evaluations $f(X)$ over explicit Hilbert spaces and operator tuples to abstract polynomials $f$ in an algebra, which is why the optimization variables $\mathcal{H}$ and $\underline{X}$ are omitted.} The next step consists of relaxing this problem by not considering all polynomials in $\mathcal{K}$, but only those up to a certain degree.  To do so, we restrict the summands of $p \in \mathcal{K}$ to have maximal degree $2d$, which defines $\mathcal{K}_{2d}$, the \textit{truncated} quadratic module of degree $2d$. This allows one to define the $d$-th level of the non-commutative SOHS hierarchy, which is dual to the NPA hierarchy:
\begin{align}
    &\lambda_d = \inf_{\lambda \in \mathbb{R}} 
    \qquad \lambda \label{eq:truncated}\\
    &\qquad \qquad \, \, \, \,\mathrm{s.t.}  \qquad \lambda \revise{\oneop{1}} - f \in \mathcal{K}_{2d}, \nonumber
\end{align}
where the minimal level for which the relaxation can be run is given by 
\begin{equation}
d_{\min} = \lceil\max\{ \mathrm{deg}(f), \mathrm{deg}(h_i), \mathrm{deg}(g_i) \}/2\rceil.
\end{equation}
Now, obviously, $\mathcal{K}_{2d} \subseteq \mathcal{K}_{2d+2} \subseteq...\subseteq \mathcal{P}_K\langle \underline{X}\rangle$, which implies $\lambda_{d_{\min}} \geq \lambda_{d_{\min}+1} \geq \cdots \geq \lambda_{\max}$, yielding a series of monotonically decreasing upper bounds $\lambda_d$ to the original problem. 

Most conveniently, each membership certificate of $\mathcal{K}_{2d}$ can be written as a matrix positivity condition. Such optimization problems take the form of an SDP, enabling the use of a well developed toolbox for convex optimization. We note here that, in contrast, checking membership of $\mathcal{K}$ is computationally intractable, as no efficient numerical methods are known for this task. 

 % mention sdp solver review

It has been shown in \cite{navascues_convergent}, see also~\cite{npa,moment}, that under the Archimedean condition this series of upper bounds $\{ \lambda_d\}_{d_{min}}^\infty$ converges asymptotically, i.e., $\lambda_{\infty} = \lambda_{\max}$. However, there are instances where the relaxations converge at a finite order, i.e., $\exists d \in \mathbb{N}: \lambda_d = \lambda_{\max}$, where the occurrence is tied to a so-called \textit{flatness} or \textit{rank loop} condition on the moment matrix. For a more detailed discussion of the theory of non-commutative semidefinite relaxations, the reader is directed to \cite{burgdorf_optimization_2016}.

%\textbf{IMO, HERE WE COULD TRY TO MAKE THE DISCUSSION A BIT LESS TECHNICAL. I WOULD INTRODUCE EQUIVALENCE CLASSES, A TERM MORE FAMILIAR FOR PHYSICIST, BEFORE.}

Now, many problems naturally occurring in quantum theory exhibit polynomial equality constraints.  One can make use of this fact to significantly reduce the size of the underlying SDP by working over the associated quotient space, e.g. comparing equivalence classes up to the constraining equality relations: 
Let $\mathcal{I} = \left\{ \sum_i v_i h_i w_i\, |\, v_i, w_i \in \mathbb{K}\langle \underline{X} \rangle; h_i \in H \right\}$ denote the polynomial \textit{two-sided ideal} generated by $H$. Given a monomial order, any polynomial can be reduced to its unique equivalence class modulo the generators of the ideal. The resulting polynomial is called the \textit{normal form} $\mathcal{N}(\cdot)$ of the original one, where the set of normal forms of monomials in $\mathbb{K}\langle \underline{X}\rangle$ forms a canonical basis of the quotient space $\mathbb{K}\langle \underline{X}\rangle / \mathcal{I}$.
It is then sufficient to optimize only over said equivalence classes in the degree-$d$ truncation of the quotient algebra, $\mathbb{K}\langle \underline{X} \rangle_d / \mathcal{I}$:
\begin{align}
    &\lambda_d = \inf_{\lambda \in \mathbb{R}, f_i, u_i \in \mathbb{K}\langle\underline{X} \rangle/\mathcal{I} } 
    \qquad \lambda \label{eq:relaxation} \\
    & \, \, \, \, \,\mathrm{s.t.}  \quad \lambda \revise{\oneop{1}} - f  = \sum_i f_i^*f_i  + \sum_{i,j} u_{ij}^* g_i u_{ij}\quad \mathrm{mod}\,\mathcal{I}, \nonumber \\
     \,  \quad& \qquad \mathrm{deg}(f_i) \leq d, \,\, \mathrm{deg}(u_{ij}) \leq \left\lfloor \frac{d -\mathrm{deg}(g_i)}{2} \right\rfloor := d_i.\nonumber
\end{align}
 As many problems native to quantum theory exclusively feature equality constraints, we are going to restrict the following illustration to the case of no inequalities, $G = \emptyset$, for now.

To make the SDP underlying \eqref{eq:relaxation} more apparent, define a monomial basis $v_d$ of $\mathbb{K}\langle \underline{X} \rangle_d / \mathcal{I}$, with cardinality denoted by $s_{n,d}^{\mathcal{I}}$.  The problem can then be phrased as an optimization over PSD \textit{Gram matrices} $G_0 \in \mathbb{K}^{s_{n,d}^{\mathcal{I}}\times s_{n,d}^{\mathcal{I}}}$:
\begin{align}
    &\lambda_d = \inf_{\lambda \in \mathbb{R},G_0 \in \mathbb{K}^{s_{n,d}^{\mathcal{I}} \times s_{n,d}^{\mathcal{I}}}} 
    \qquad \lambda \label{eq:gramcompletion} \\
    &\qquad \qquad \, \, \, \,\mathrm{s.t.}  \qquad \lambda \revise{\oneop{1}} - f = v_d^* G_0 v_d \quad \mathrm{mod}\,\mathcal{I},  \nonumber \\
    & \, \qquad \qquad \qquad \qquad G_0 \succeq 0, \nonumber
\end{align}
where the unconstrained case is recovered by working over the trivial ideal $\mathcal{I}=\emptyset$. The optimization can be interpreted as follows: Some of the monomial coefficients of $f$ predetermine entries of $G_0$, while others are undetermined. When running the numerical SDP solver, it attempts to find the best pair $(\lambda, G_0)$ in the intersection of the PSD cone with the affine subspace spanned by the coefficients of $\lambda - f(\underline{X})$. 

The size of PSD matrices in the hierarchy grows exponentially in $d$, where in the unconstrained case it reads as
\begin{equation}
    s_{n,d} = \mathrm{dim}\left(\mathbb{K}\langle \underline{X}\rangle_d\right)=\frac{n^{d+1 - 1}}{n -1},
\end{equation}
which in general leads to only lower relaxation levels being numerically addressable. For constrained problems, working in the quotient algebra of dimension $\mathrm{dim}\left(\mathbb{K}\langle \underline{X}\rangle_d\right/\mathcal{I}) =s_{n,d}^{\mathcal{I}}<s_{n,d}$ often reduces the matrix sizes drastically. As mentioned earlier, improving the scalability of the NPA hierarchy, and SDPs in general is an active field of research, where the interested reader is directed to \cite{magron2023sparse} for a survey on scalability improvements by sparsity exploitation. 

 It is also possible to study this problem on its dual side, optimizing over the set of \textit{quantum moments} rather than non-negative polynomials~\cite{navascues_convergent,npa}. In that setting, the \textit{moment matrix} replaces the Gram matrix as the central PSD variable. Although the dual viewpoint is often used in the literature, we restrict our discussion to the primal side, as it is the more natural setting for algebraic certification. Note also that the works~\cite{navascues_convergent,npa} follow the opposite convention when defining the primal and dual problem.

While the above framework can be used to address generic non-commutative optimization problems of the form \eqref{eq:originalprob}, polynomial optimization does not constitute the exclusive use of SDPs in quantum theory. SDPs naturally occur, not exhaustively, in applications such as entanglement detection through the Doherty-Parrilo-Spedalieri hierarchy \cite{dps}, the quantum marginal problem \cite{marginal}, tomography problems \cite{meas_tomography}, or quantum network correlations \cite{qinflation}. In many cases, the SDP structure of the problem arises from the PSD nature of the underlying quantum states or measurements.

%Notes: Write about finite convergence 
\section{Extracting \revisegreen{certified} bounds} \label{sec:rational_bounds}
% Mention reasons for failing numerical solvers 
As anticipated in Section \ref{sec:outline}, SDP solvers  aiming to certify non-negativity of $\lambda - f$ via the numerical tuple $(\lambda, G_0)$ might issue inaccurate solutions due to several reasons, such as inexact termination, failure of strict feasibility, or poor conditioning of the problem. Although in theory SDPs are exactly solvable by quantifier elimination \cite{sdp_exact}, this method has doubly exponential complexity in the number of variables \cite{quant_elim}. In contrast, if only approximate solutions are required, SDPs can be solved in polynomial time using convex optimization methods such as the ellipsoid or interior-point algorithms \cite{sdp_poly}. Since scalability is crucial for many problems in quantum theory, this work focuses on post-processing of the numerical data rather than aiming for exact solutions of the underlying SDPs.

Various works discuss the appropriate handling of obtaining \revisegreen{exact} Sum-of-Squares (SOS) decompositions from inexact data. 
For univariate polynomials with real or complex coefficients, there are dedicated algorithms approximating complex roots of perturbed positive polynomials \cite{magron2019algorithms,magron2022exact}, or solving norm equations together with using Euler's identity \cite{koprowski2023pourchet}. 
In the multivariate case, one can either rely on the so-called \textit{rounding and projection} \cite{peyrl2008computing} (see also \cite{kaltofen_exact_2012}) or perturbation-compensation algorithms \cite{magron2021exact}, implemented in the RealCertify package \cite{magron2018realcertify}. 
For constrained compact problems, one can apply rigorous interval arithmetic techniques to obtain a certified upper bound \cite{jfr14}. 
Similarly to that setting, tailored methods to obtain certified bounds to SDP relaxations to the optimal power-flow problems have been presented in \cite{certified_bounds_acopf}. \revise{Related computer-assisted proof techniques have also been used in quantum information to certify finite-dimensional SDP bounds for specific problem instances, for example in channel-discrimination problems, by projecting numerical SDP solutions onto affine matrix constraints and restoring positivity \cite{PhysRevLett.127.200504}. In contrast, our method aims to extract certified solutions to a wider range of problems, encompassing any problem following the general NPO formulation \eqref{eq:originalprob}.}
Other algorithms have been designed for specific problems with finitely many minimizers \cite{magron2023sum,baldi2024effective}. 
Such methods provide \revisegreen{exact,} rational SOS decompositions, together with a \revisegreen{certified }rational upper bound of the problem solution. 
If one is interested in certifying global optimality, one can rely  on the recently developed high-precision solver \cite{leijenhorst2024solving} that provides exact SOS decompositions, possibly with non-rational coefficients. On the other hand, if one primarily seeks highly accurate numerical bounds rather than exact certificates, one may use arbitrary-precision SDP solvers (e.g., SDPA-GMP \cite{sdpa-gmp}) at the expense of a larger computational cost. However, the output remains a numerical approximation and does not constitute a proof unless further validated by an adequate certification method.
The rounding and projection algorithm also has been extended to unconstrained non-commutative problems in \cite{rational_sohs} to obtain \revisegreen{exact} rational SOHS decompositions. 
Our main algorithmic contribution is based on this rounding and projection scheme, in order to handle constraints for maximal eigenvalue problems, as well as sparse/symmetric input data. 

Our method offers a complementary guarantee in situations where rigorous bounds to a given problem are of primary interest, independently of the convergence of the original problem. 

We now present our constructive procedure that converts \revisegreen{an approximate numerical} SDP solution of Problem \eqref{eq:gramcompletion} into a formally exact rational certificate.
Starting with the case of an unconstrained NPO $(G = H = \mathcal{I}=\emptyset)$ let $(\lambda_d, G_0)$ be the numerical tuple of upper bound $\lambda_d \in \mathbb{R}$ and Gram matrix $G_0 \in \mathbb{K}^{s_{n,d} \times s_{n,d}}$, as issued by the SDP solver of one's choice. The \revise{left hand side (LHS) $L_d$} of the numerical certificate at order $d$, suggesting $\lambda_d > \lambda_{\max}$ then reads
\begin{equation}
\label{eq:inexactcertificate}
    \revise{L_d} := \lambda_d \revise{\oneop{1}} - f \approx v_d^*G_0 v_d.
\end{equation}
The method to obtain \revisegreen{an exact} certificate from (\ref{eq:inexactcertificate}) consists of three algebraic steps: (1) rounding \revisegreen{of the floating point data to rationally represented data in $\mathbb{Q}$ with controlled denominator}, (2) projection, and (3) lifting,  each preserving or restoring feasibility within the truncated quadratic module $\mathcal{K}_{2d}$. \revisegreen{After the rounding step, the algebraic operations are performed using exact rational arithmetic, so that no additional round-off error is introduced}. A geometric picture of the procedure is displayed in Figure \ref{fig:round-project}. After the first two steps, rounding and projection, the solver’s tuple $(\lambda_d, G_0)$  has been mapped  onto the affine subspace $\mathcal{L}$ consistent with the coefficients of $\lambda_d \revise{\oneop{1}} - f$, thereby enforcing exact equality in Eq.~\eqref{eq:inexactcertificate}. We call the obtained tuple $(\mathcal{P}(\tilde{G}_0),\tilde{\lambda}_d)$ a \textit{pre-certificate}. In the case that $\mathcal{P}(\tilde{G}_0) \succeq 0$, we established a proper SOHS certificate, proving $\tilde{\lambda}_d$ as an exact upper bound. If the resulting matrix $\mathcal{P}(\tilde{G}_0)$ is not PSD, in step (3) we lift it back into the PSD cone by adding a rational correction justified by the constraint structure. This lifted matrix $\mathcal{P}_+(\tilde{G}_0)$ then certifies a \revise{\textit{weaker}} bound $\lambda_d^{\mathrm{rat}} > \tilde{\lambda}_d$, which holds exactly. That is, step (3) is only necessary if $\mathcal{P}(\tilde{G}_0)$ is not PSD. However, in some situations, one can still apply step (3) when $\mathcal{P}(\tilde{G}_0) \succeq 0$ to further tighten the bound $\tilde{\lambda}_d$, as it is shown later.

\begin{figure}[ht]
  \centering
  \includegraphics[width=0.7\columnwidth]{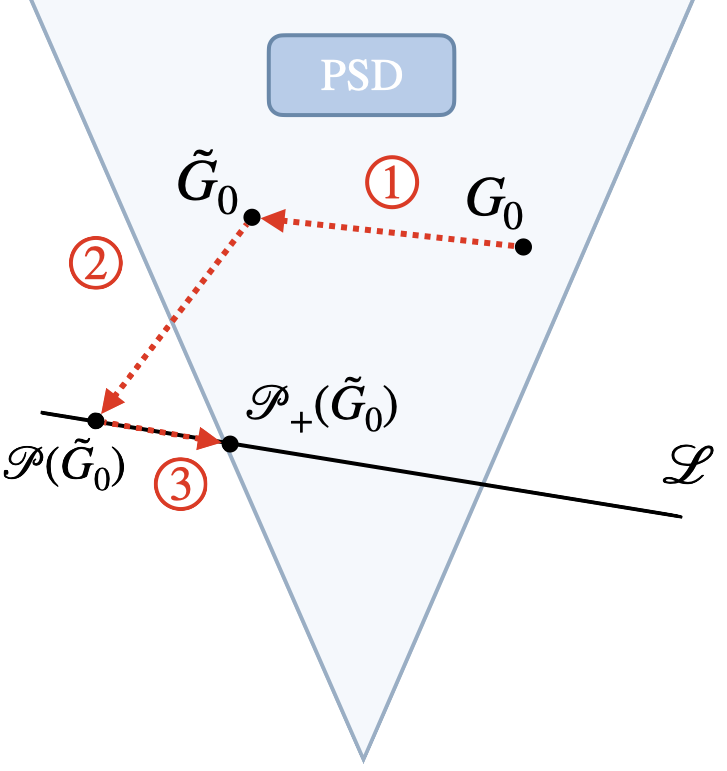}
  \caption{The geometric picture of our procedure. (1) $G_0$ \revisegreen{in floating-point representation} is rounded to a rational  $\tilde{G}_0$. (2) $\tilde{G}_0$ is projected onto $\mathcal{L}$ \revisegreen{in exact arithmetic}, possibly leaving the PSD cone. (3) $\mathcal{P}(\tilde{G}_0)$ is lifted back into the PSD cone, certifying $\lambda_d^{\mathrm{rat}}$ as an upper bound.}
  \label{fig:round-project}
\end{figure}

In the rest of the section, we provide a rigorous formulation of the three steps in the process, both for problems with and without equalities.

\subsection{Round and Project}

To obtain \revisegreen{a nearby rational} tuple $(\mathcal{P}(\tilde{G}_0),\tilde{\lambda}_d)$ \revisegreen{from their floating-point counterparts} in the unconstrained case, one can make use of the \textit{Round and Project} procedure from \cite{rational_sohs}: Fix a given rounding precision $\eta$ and start by rounding the \revisegreen{floating-point} data $(\lambda_d, G_0)$ entry-wise to \revisegreen{rational representations} $(\tilde{\lambda}_d, \tilde{G}_0)$, where established methods such as \textit{continued fractions} can be used. \revisegreen{This rounding step should be understood as a reconstruction of adequate rational representations of the floating-point numbers: every floating-point number already represents an exact
rational number, and one could in principle use this exact representation directly in all subsequent exact computations. We instead reconstruct a nearby rational number of controlled precision in order to suppress insignificant numerical noise and to avoid unnecessarily large numerators and denominators in the subsequent exact rational arithmetic operations.}

\revisegreen{Once the
reconstructed data $(\widetilde{\lambda}_d,\widetilde G_0)$ have been chosen, the subsequent algebraic projection and lifting operations are performed exactly over the rational coefficient field to not introduce any additional round-off error. Using exact arithmetic instead of floating-point arithmetic is hence crucial for the algebraic manipulation to preserve the exactness of the certificate in intermediate manipulation steps.} In the rest of the work, $\tilde{A}$ represents the entrywise rationalized form of the object $A$, where rationalization on polynomials is understood as acting on their coefficient vector in the monomial basis. We denote by $\mathbb{Q}$ the rational subfield of $\mathbb{K}$ and by $\mathbb{Q}\langle \underline{X}\rangle$ the associated rational polynomial algebra (real or complex).

Now, as mentioned earlier, since $f$ is assumed to have rational coefficients, it constrains the entries of $G_0$ corresponding to its support. To recover this relationship after rounding, one has to project the entries of $\tilde{G}_0$ back into the affine space $\mathcal{L}$ defined through the coefficients of $\revise{\tilde{L}_d}:=\tilde{\lambda}_d \revise{\oneop{1}} - f$ via \cite{rational_sohs}:
{\small
\begin{align}
\label{eq:projection}
    \mathcal{P}(\tilde{G}_0)_{\alpha, \beta} = (\tilde{G}_0)_{\alpha, \beta} - \frac{1}{n(\alpha, \beta)}\sum_{\gamma^*\delta \sim \alpha^* \beta} (\tilde{G}_0)_{\gamma, \delta} - \revise{(\tilde{L}_d)}_{\alpha^*\beta},
\end{align}}
which can be expressed compactly with a rational correction matrix $\Delta$ as $\mathcal{P}(\tilde{G}_0) = \tilde{G}_0 + \Delta$.
The tuple $(\alpha,\beta)$ refers to the indexing with respect to monomials $\alpha, \beta \in \mathbb{K}\langle \underline{X} \rangle_d$, and $p_{\alpha^*\beta}$ refers to the monomial coefficient of $\alpha^*\beta$ in $p$. Intuitively, this projection enforces exact algebraic consistency between the rounded Gram matrix and the target polynomial by averaging over equivalent monomial pairs.
Here, $n(\alpha,\beta)$ accounts for the number of such pairs:
\begin{equation}
    n(\alpha,\beta) = \{(\gamma,\delta) \,|\, \gamma^*\delta = \alpha^*\beta\}.
\end{equation}

When considering problems with equalilty constraints, this projection can be adapted by working over the quotient algebra $\mathbb{K}\langle \underline{X}\rangle /\mathcal{I}$ by enforcing coefficient equality of \textit{normal forms} $\mathcal{N}(\alpha^*\beta)$ of monomials $\alpha^*\beta$, i.e., their equivalence classes in the quotient algebra. The goal is again to find a correction matrix $\Delta$ such that the corrected certificate holds exactly up to normal form reductions
\begin{equation}
    \mathcal{N}(\revise{\tilde{L}_d}) = \mathcal{N}\left(v_d^*(\tilde{G}_0 + \Delta)v_d \right).
\end{equation}
As several entries $(\alpha, \beta)$ map to the same normal form, there are in general infinitely many ways of constructing $\Delta$. The most natural choice is the correction with the least Frobenius deviation from the original rationalized Gram matrix: As the numerical Gram matrices are PSD up to the solver precision, deviating the least possible from them will minimize the distance to the PSD cone after projection. The next Lemma provides the formal construction of this projection.

\begin{lemma}[Frobenius optimal Gram projection]
\label{thm:frobopt}
For a given optimization problem (\ref{eq:gramcompletion}), let $\alpha^*\beta$
be the monomials corresponding to the basis elements $(\alpha,\beta)$ of the Gram matrix $G_0$. By construction, they form a basis of $\mathbb{K}\langle \underline{X}\rangle_{2d}$.
Let $t$ index the basis of equivalence classes in $\mathbb{K}\langle \underline{X}\rangle_{2d}/\mathcal{I}$. The associated normal form operation can then be expressed as
\begin{align}
    \mathcal{N}: \mathbb{K}\langle \underline{X} \rangle_{2d} &\rightarrow \mathbb{K}\langle \underline{X}\rangle_{2d} / \mathcal{I} \\
    \alpha^*\beta &\mapsto \sum_t n_{\alpha, \beta}^t t \nonumber
\end{align}
where $n_{\alpha, \beta}^t$ is the coefficient with which the word $\alpha^*\beta$ reduces to $t$ under $\mathcal{N}$. Define the rational numerical certificate residual coefficients $r_t$ per reduced word $t$ as
\begin{align}
    r_t &= \mathcal{N}(\tilde{\revise{L}}_d)_t - \mathcal{N}\left(\sum_{\alpha,\beta}v_{d,\alpha}^*(\tilde{G}_0)_{\alpha, \beta}v_{d,\beta}\right)_t \nonumber \\
    &= \mathcal{N}(\tilde{\revise{L}}_d)_t - \sum_{\alpha,\beta}n_{\alpha,\beta}^t(\tilde{G}_0)_{\alpha, \beta}.
\end{align}
Let $\Xi$ be the matrix with entries
\begin{equation}
    (\Xi)_{t,s} := \sum_{\alpha, \beta}n_{\alpha,\beta}^tn_{\alpha,\beta}^s.
\end{equation}
If $\Xi$ is non-singular, the following projection $\mathcal{P}(\tilde{G})$ yields exact equality between the left hand side (LHS) and right hand side (RHS) of the SOHS certificate:
\begin{align}
\label{eq:froboptproj}
    \mathcal{P}(\tilde{G}_0)_{\alpha,\beta} & = (\tilde{G}_0)_{\alpha,\beta} + \Delta_{\alpha,\beta} \nonumber\\&= (\tilde{G}_0)_{\alpha,\beta} + \sum_t (\Xi^{-1}r)_t n_{\alpha,\beta}^t.
\end{align}
Furthermore, it is the optimal projection $\mathcal{P}$ in terms of Frobenius deviation 
\begin{equation}
    \|\mathcal{P}(\tilde{G}_0) - \tilde{G}_0 \|^2_F = \sum_{\alpha,\beta} |\left(\mathcal{P}(\tilde{G}_0) - \tilde{G}_0\right)_{\alpha,\beta} |^2
\end{equation}
from the rounded numerical Gram matrix $\tilde{G}_0$.
\end{lemma}

\begin{proof}[Proof of Lemma \ref{thm:frobopt}]
Let the assumptions in Lemma \ref{thm:frobopt} hold. We phrase the problem of finding the correction $\Delta$ with least Frobenius norm as a convex quadratic optimization problem:
\begin{align}
    &\inf_{\Delta \in \mathbb{Q}^{s_{n,d}^{\mathcal{I}}\times s_{n,d}^{\mathcal{I}}}} \frac{1}{2} \|\Delta\|_F^2 \\
    \mathrm{s.t.} \qquad  & \quad r_t = \sum_{\alpha,\beta} n_{\alpha, \beta}^t\Delta_{\alpha,\beta}, \quad \forall t, \nonumber
\end{align}
which can be solved using Lagrangian multipliers $\lambda_t$. Stationarity of the Lagrangian yields
\begin{equation}
\label{eq:proof2eq2}
    \Delta_{\alpha,\beta} = \sum_t \lambda_t n_{\alpha, \beta}^t.
\end{equation}
Plugging back into the constraints yields
\begin{equation}
\label{eq:proof2eq3}
    r_t = \Xi\lambda_t.
\end{equation}
If $\Xi$ is non-singular, using \eqref{eq:proof2eq2} and \eqref{eq:proof2eq3} one finds the global optimizer
\begin{equation}
    \Delta_{\alpha,\beta} = \sum_t (\Xi^{-1}r)_t n_{\alpha, \beta}^t.
\end{equation}
\end{proof}

Now, in many settings, such as Bell inequalities or ground states of quantum many body systems, the equality constraints defining the problem relate two monomials to each other. Hence, the associated normal forms map monomials onto monomials instead of polynomials. We call such normal forms and the associated ideals \textit{binomial.} In such cases, one finds a diagonal $\Xi \succ 0$ and the Frobenius optimal solution always exists:

\begin{corollary} \label{thm:binomial_proj}
    If $\mathcal{N}$ is binomial, the Frobenius ideal projection reduces to
    \begin{equation}
    \label{eq:binomialproj}
        \mathcal{P}(\tilde{G})_{\alpha,\beta} = \tilde{G}_{\alpha,\beta} + \frac{r_{\mathcal{N}(\alpha^*\beta)}}{n_\mathcal{I}(\alpha,\beta)}
    \end{equation}
    with 
    \begin{equation}
        n_{\mathcal{I}}(\alpha,\beta) = \#\{(\gamma,\delta): \mathcal{N}(\alpha^*\beta) = \mathcal{N}(\gamma^*\delta) \}.
    \end{equation}
\end{corollary}

\begin{proof}[Proof of Corollary \ref{thm:binomial_proj}]
Binomial normal forms can be expressed as 
\begin{equation}
    \mathcal{N}: \alpha^*\beta \mapsto \delta_{\mathcal{N}(\alpha^*\beta),t} t
\end{equation}
which allows us to rewrite the residuals $r_t$ as
\begin{equation}
    r_t = \mathcal{N}(\revise{\tilde{L}_d})_t - \sum_{\alpha,\beta:\mathcal{N}(\alpha^*\beta) = t}\tilde{G}_{\alpha,\beta}.
\end{equation}
Applying the method of Lagrangian multipliers yields
\begin{equation}
\label{eq:coreq2}
    \Delta_{\alpha,\beta} = \lambda_t \delta_{\mathcal{N}(\alpha^*\beta),t}.
\end{equation}
Plugging this into the constraints yields
\begin{equation}
    r_t = \sum_{\alpha,\beta:\mathcal{N}(\alpha^*\beta) = t} \lambda_t \Rightarrow \lambda_t = \frac{r_t}{n_{\mathcal{I}}(\alpha,\beta)},
\end{equation}
which, when used with \ref{eq:coreq2}, implies the claim.
\end{proof}

As (\ref{eq:projection}) is implied by (\ref{eq:binomialproj}), in what follows we are going to restrict our considerations to the quotient case without loss of generality.

The projection formula is also readily adapted to problems featuring inequality constraints. Recall that in that case, the numerical certificate reads
\begin{equation}
    \lambda \revise{\oneop{1}} - f =\sum f_i^*f_i  + \sum_{i,j}u_{ij}^*g_iu_{ij} \quad \mathrm{mod}\,\mathcal{I},
\end{equation}
which, when expressed through Gram matrices, features one matrix $G_0$ corresponding to the SOHS $\sum_i f_i^*f_i$, and $m_G$ matrices $\{G_i\}_{i=1}^{m_G}$ of sizes $s_{n,d_i}^{\mathcal{I}}$ corresponding to the so-called localizing terms $\sum_{i,j} u_{ij}^* g_i u_{ij}$. To deal with such constraints, one adapts the LHS from Equation \eqref{eq:inexactcertificate} to feature the inequality terms, and projects the rounded principal Gram matrix $\tilde{G}_0$ onto the rationalized LHS as before to match coefficients.

%\textbf{I'M A BIT LOST HERE, UNTIL THE NEW SUBSECTION. I FIND IT SOMETIMES A BIT REDUNDANT AND NOT SO CLEAR.}

To do so, start by constructing the rational LHS \revise{$\tilde{L}_d$} featuring the inequality terms:
\begin{equation}
\label{eq:ineqratcert}
    \revise{\tilde{L}_d} := \tilde{\lambda}_d \revise{\oneop{1}} - f - \sum_i \sum_j \tilde{u}_{ij}^* g_i \tilde{u}_{i,j}.
\end{equation}
To obtain the rational polynomials $\tilde{u}_{ij}$ in the decomposition, one starts by rounding the \revisegreen{floating-point} localizing Gram matrices $\{G_i\}_{i=1}^{m_G}$ to rationals $\{\tilde{G}_i\}_{i=1}^{m_G}$. \revisegreen{Afterwards, to obtain a nearby rational PSD Gram matrix, we numerically project $\widetilde G_i$ onto the PSD cone via an eigendecomposition and rationally reconstruct the resulting factors:}
\begin{equation} \small 
    \mathcal{P}(\tilde{G}_i) = \tilde{U}_i\tilde{U}_i^T; \quad \tilde{U}_i = \tilde{V}_i \mathrm{diag}\left (\max (\tilde{\sqrt{\lambda_{ij}}},0) \right)
\end{equation}
with $\tilde{V}$ being the matrix of rationalized eigenvectors of $\tilde{G_i}$ and $\{\lambda_{ij}\}_{j=1,..,s_{n,d_i}}$ the corresponding \revisegreen{numerical} eigenvalues. From these, the polynomials $\tilde{u}_{ij}$ are obtained as 
\begin{equation}
    \tilde{u}_{ij} = (\tilde{U}_i^T {{v}_d}_i)_j,
\end{equation}
where $ {{v}_d}_i$ denotes the monomial basis appropriate for the localizing matrix of $g_i$. Using Lemma \ref{thm:frobopt} to project $\tilde{G}_0$ onto \revise{$\tilde{L}_d$} from Equation \eqref{eq:ineqratcert} yields a pre-certificate $\revise{\tilde{L}_d} = v_d^* \mathcal{P}(\tilde{G}_0)v_d$, and with $\mathcal{P}(\tilde{G}_0) \succeq 0$ it would certify membership of $\tilde{\lambda}_d \revise{\oneop{1}} -f $ in $\mathcal{K}_{2d}$ and hence $\tilde{\lambda}_d$ as an upper bound to $f$. This procedure also applies to moment inequality constraints of the form $\langle m_i\rangle \geq 0$, where the associated multiplier polynomials $\tilde{u}_{ij}$ are scalar valued (degree $0$).

Performing the Round and Project procedure, in the unconstrained or constrained case, results in an exact pre-certificate
\begin{equation}
\label{eq:exact_cert}
    \revise{\tilde{L}_d} = v_d^*\mathcal{P}(\tilde{G_0})v_d \quad \mathrm{mod}\,\mathcal{I},
\end{equation}
where this expression constitutes an exact algebraic identity, ensuring that any remaining infeasibility stems solely from the non-positivity of $\mathcal{P}(\tilde G_0)$.

As anticipated, in the best case, the resulting projected Gram matrix $\mathcal{P}(\tilde{G_0})$ is already PSD. This can only be the case if the rounded bound $\tilde{\lambda}_d$ is a valid upper bound to the global optimum, $\tilde{\lambda}_d \geq \lambda_{\max}$, and hence leaves us with the certified bound $\tilde{\lambda}_d$.

\subsection{Lifting positivity}
\label{sec:liftingpositivity}
The remaining question is how to repair a pre-certificate \eqref{eq:exact_cert} whose Gram matrix was projected out of the PSD cone, rendering it vacuous. This particularly happens if $\tilde{\lambda}_d \ngeq \lambda_{\max}$. The next Theorem provides, under reasonably mild assumptions, an explicit, constraint-dependent lifting that guarantees a valid upper bound.

\begin{theorem}[Certified bounds to NPOs]
\label{thm:theorem2}
Let $\mathcal{P}(\tilde{G}_0)$ be the rounded and projected Gram matrix obtained from Lemma \ref{thm:frobopt} with  $\mu_{\mathrm{min}} \in \mathbb{Q}$ a lower bound to its spectrum, and let $s^{\mathcal{I}}_{n,d}$ denote the number of monomials in $n$ variables up to degree $d$ after normal form reduction. Assume that the underlying variables $\underline{X}$ are subject to one of the following polynomial constraints:
\begin{itemize}
    \item \label{itm:ball} Ball constraints: $1 \succeq \sum_i X_i^2$, 
    \item \label{itm:box} Box constraints: $1-X_i^2 \succeq 0\, \, \forall i$, 
    \item \label{itm:uni}Unipotency: $1-X_i^2 = 0\, \, \forall i$, 
    \item \label{itm:proj} Projection constraints: $X_i^2 = X_i\, \, \forall i$, 
\end{itemize}
then, the rational bound $\lambda^{\mathrm{rat}}_d \in \mathbb{Q}$ with
\revise{
\begin{equation}
    \lambda^{\mathrm{rat}}_d := \tilde{\lambda}_d + \delta_d; \quad \delta_d := - \min\{\mu_{\mathrm{min}} ,0\}\cdot s^{\mathcal{I}}_{n,d}
\end{equation}}
constitutes a certified upper bound to the optimization problem, i.e.,  $\lambda^{\mathrm{rat}}_d \geq \lambda_{\max}$. 
\end{theorem}
Importantly, the condition on the existence of constraints is rather weak: In the non-locality scenario, Bell type expressions can be defined in terms of projectors, while in the quantum many-body setting, the underlying local spin operators naturally fulfill unipotency.

The geometric intuition behind the proof of Theorem \ref{thm:theorem2} is that the non-PSD Gram matrix $\mathcal{P}(\tilde G_0)$ is lifted by adding a multiple of the identity to the certificate, leading to a \revise{weaker} bound $\lambda^{\mathrm{rat}}$. The allowed correction magnitude is guaranteed to be feasible within the quadratic module due to the following Lemma.

\begin{lemma}[Constant SOHS decomposition]
\label{thm:boundlowering}
    Let $\varepsilon >0$. If the variables $\underline{X}$ are subject to constraints as in Theorem \ref{thm:theorem2}, one can find the following \revisegreen{exact,} constant SOHS decomposition
    \begin{equation}
       \varepsilon s^{\mathcal{I}}_{n,d} \revise{\oneop{1}} = \varepsilon  v_d^* v_d + \sum_{k} r_{k}^* r_k + \sum_i q_{ki}^* c_i q_{ki} \quad \mathrm{mod}\,\mathcal{I} \label{eq:ballconstr}
    \end{equation}
    for some $r_k, q_{ki} \in \mathbb{Q}\langle \underline{X} \rangle_d$ and $v_d$ the monomial basis of
    $\mathbb{K}\langle \underline{X} \rangle_d/\mathcal{I}$.
\end{lemma}

Lemma~\ref{thm:boundlowering} can be proven by induction, as shown in Appendix \ref{sec:appendix}. Using this, we can prove Theorem \ref{thm:theorem2}:

\begin{proof}[Proof of Theorem \ref{thm:theorem2}]
Let the assumptions in Theorem \ref{thm:theorem2} hold. We are going to show that for $\lambda^{\mathrm{rat}}$ there exists some PSD gram matrix $\mathcal{P}_+(\tilde{G}_0)$ and polynomials $p_{ij}$ such that
\begin{align}
    \lambda^{\mathrm{rat}} \revise{\oneop{1}} - f &= v_d^* \left( \mathcal{P}_+(\tilde{G}_0) \right) v_d  + \sum_{i,j}p_{ij}^*g_ip_{ij}\qquad \mathrm{mod}\,\mathcal{I} \nonumber \\
    &\Leftrightarrow \lambda^{\mathrm{rat}} \revise{\oneop{1}} - f \in \mathcal{K}_{2d},
\end{align}
which implies the claim.

Consider the exact pre-certificate \eqref{eq:exact_cert}. If $\mu_{\mathrm{min}}<0$ then adding \eqref{eq:ballconstr} with $\varepsilon = -\mu_{\min}$ to both sides of the certificate lifts $\mathcal{P}(\tilde{G}_0)$ back into the PSD cone as follows:
\begin{align}
    &\tilde{\lambda}_d \revise{\oneop{1}} - \mu_{\min}s^{\mathcal{I}}_{n,d} \revise{\oneop{1}} - f - \sum_{i,j}\tilde{u}^*_{ij}g_i \tilde{u}_{ij} \\
    &= v_d^*\left(\mathcal{P}(\tilde{G}_0) -\mu_{\min}\oneop{1} \right)v_d + \sum_k r_k^*r_k + \sum_{i,k}q_{ki}^*c_iq_{ki}. \nonumber
\end{align}
Intuitively, this addition corresponds to lifting the spectrum of $\mathcal{P}(\tilde G_0)$ by $-\mu_{\min}>0$ while staying within the quadratic module, thereby restoring positivity without changing the algebraic identity of the certificate. 

%\footnote{This bound tightening step $(\varepsilon=-\mu_{\min}<0)$ is only valid when the constraints used in Lemma \ref{thm:boundlowering} are equality constraints. For inequality constraints, the localizing terms define a cone and cannot be scaled by a negative coefficient while remaining in $\mathcal K_{2d}$.}

Expressing the $r_k$ contributions with a PSD matrix $R_0$ of the same size, and joining terms corresponding to constraints $c_i$ yields
\begin{align}
        \tilde{\lambda}_d \revise{\oneop{1}} - \mu_{\min}s^{\mathcal{I}}_{n,d} \revise{\oneop{1}} - f &= v_d^*\left(\mathcal{P}(\tilde{G}_0) - \mu_{\min}\oneop{1} + R_0 \right)v_d \nonumber \\ &+ \sum_{i,j}p_{ij}^*g_i p_{ij},
\end{align}
which implies the claim with
\begin{equation}
   \mathcal{P}_+(\tilde{G}_0) =\left(\mathcal{P}(\tilde{G}_0) - \mu_{\min}\oneop{1} + R_0 \right)
\end{equation}
depending on the constraints involved.
\end{proof}
It should be noted that a numerical bound may have been correct even if the corresponding Gram matrix $G_0$ is projected outside the PSD cone. Consequently, our method allows one to (a) recover valid bounds from possibly incorrect ones, and (b) validate correct bounds, but it does not provide a necessary certificate of invalidity for numerical bounds. 

The above argument only holds if $\mu_{\min}$ actually constitutes a rigorous lower bound to the spectrum of $\mathcal{P}(\tilde{G}_0)$. Unfortunately, the Abel-Ruffini theorem implies no closed form solution for the eigenvalues of generic square matrices of dimension larger than four, as they are defined as the roots of the associated characteristic polynomial. \revisegreen{A rigorous lower bound on the smallest eigenvalue can instead be obtained using certified interval arithmetic. In our implementation, \texttt{Arblib.jl} \cite{arblib} performs arbitrary-precision floating-point ball arithmetic, yielding an interval guaranteed to contain a given eigenvalue. The certified lower endpoint of this enclosure constitutes a certified lower bound to the eigenspectrum. Since this endpoint is a finite binary floating-point number, we use its exact rational representation in the subsequent algebraic computations. As the floating-point bound is already certified and holds exactly, its exact rational representation can be used while maintaining validity throughout. We emphasize that the working precision of \texttt{Arblib} controls the tightness of the eigenvalue enclosure, but not the validity of the certification. Throughout this work, all eigenvalue bounds entering the certification procedure are understood in this rigorous sense. }

Theorem \ref{thm:theorem2} establishes that any numerical non-negativity certificate returned by an SDP solver can be transformed into an exact rational bound by a finite algebraic correction depending on the minimal eigenvalue defect and the problem size.  An important observation is that the bound correction term $\delta_d = -\min\{\mu_{\min},0\}\cdot s^{\mathcal{I}}_{n,d}$ is affected by two factors, (i) the distance from the projected $\mathcal{P}(\tilde{G}_0)$ to the PSD cone, and (ii) the size of the problem in terms of variables and relaxation order. While working with high precision solvers and rounding arithmetic can lead to smaller distances to the PSD cone, the term $s^{\mathcal{I}}_{n,d}$ ultimately dominates the bound correction, as it will be illustrated in later sections.

\subsection{Bound tightening of interior pre-certificates}

Now, in the case of the projected Gram matrix being positive-definite (PD), e.g. $\mathcal{P}(\tilde{G}_0)\succ0$, it is not necessary to perform the lifting step to restore feasibility. However, it is possible to leverage the positive eigenspectrum to obtain a \textit{tighter} bound $\lambda_d^{\mathrm{rat}}<\tilde{\lambda}_d$ by shifting the PSD matrix closer to the boundary of its cone. This can be done in two ways, (1) in the presence of unipotency constraints $X_i^2=1$, one can use a modification of Lemma \ref{thm:boundlowering} with $\varepsilon<0$, while in the other cases one can make use of a rank-1 PSD downgrade, as explained in the next Lemma.
\begin{theorem}[Certified tightening for interior pre-certificates]\label{thm:boundtightening}
    Assume that the Frobenius optimal projection \ref{thm:frobopt} leads to $\mathcal{P}(\tilde{G}_0)\succ 0$ lying in the interior of the PSD cone. Let $\mu_{\min}>0$ be a \revisegreen{certified} lower bound to the spectrum, and let 
    \begin{equation}
    \tau = \frac{1}{e_0^*\mathcal{P}(\tilde{G}_0)^{-1}e_0}
    \end{equation}
    where $e_0\in \mathbb{K}^{s_{n,d}^{\mathcal{I}}}$ denotes the coordinate vector selecting the constant monomial s.t. $e_0^*v_d = 1$. Then, a tighter bound $\lambda_d^{\mathrm{rat}}<\tilde{\lambda}_d$ can be obtained depending on the constraints involved: \\
    \item (i) If the variables are subject to unipotency constraints, $\{1-X_i^2 \}\subseteq H$, then $\lambda_d^{\mathrm{rat}} = \tilde{\lambda}_d - \mu_{\min}\cdot s_{n,d}^{\mathcal{I}}$ constitutes a rigorous upper bound to the optimization problem \\
    \item (ii) Independently of the constraints, $\lambda^{\mathrm{rat}}_d = \tilde{\lambda}_d - \tau$ constitutes a rigorous upper bound to the optimization problem, where $\tau\in[\mu_{\min},\mu_{\max}]$ with $\mu_{\max}$ an upper bound to the spectrum of $\mathcal{P}(\tilde{G}_0).$
\end{theorem}
The proof for this Theorem can be found in Appendix \ref{sec:tighteningproof}. Applying this Theorem in the case of $\mathcal{P}(\tilde{G}_0)\succ 0$ lets one further tighten the numerical rounded bounds while certifying them. Note that as the numerical optima of the SDP lie on the boundary of the PSD cone, in the generic case the eigenspectrum of $\mathcal{P}(\tilde{G}_0)$ is not positive as the Frobenius distance between $\mathcal{P}(\tilde{G}_0)$ and $\tilde{G}_0$ is minimized. However, in numerically delicate instances where PSD solvers terminate in the interior of the cone, the above scenario does occur, as it is observed and interpreted in Section \ref{sec:qmb}. Also note that in the case of unipotency constraints one can apply either of the two tightening strategies, while procedure (i) carries an extra factor $s_{n,d}^{\mathcal{I}}$ which, for large problems, will typically result in a tighter bound than procedure (ii).

Concludingly, we established a rigorous method to obtain mathematical non-negativity proof from the solver's data and presented an adapted method for interior solutions where the algorithm provides certified bounds tighter than their numerical counterpart. 

The same above reasoning of rounding, projecting, and lifting, can be extended to structured relaxations such as sparsity or symmetry-adapted hierarchies, as presented in the next section.
\section{\revisegreen{Exact} certificates under sparsity and symmetry}
\label{sec:sparse}
There exist problem for which it is necessary to go to higher orders of the hierarchy to get tight bounds to the solution of the \revise{optimization} problem. In those cases, the exponential growth of the PSD matrices with the relaxation order often makes higher intractable in practice, especially if a large number of variables is involved. This bottleneck of semidefinite relaxations is of particular importance, for instance, in quantum many-body physics, as one is often interested in the \revise{behavior} of systems of large sizes. One way to be able to address such problems, is to employ a sparse version of the hierarchy, which relies on a sparse analogue of Theorem \ref{thm:heltonmc}. Applying the sparse instead of the standard, dense hierarchy has led to improvements in scalability by orders of magnitude, if the problem admits certain structure \citep{wang_sparse, groundstate}. In other cases, the symmetry of the problem can be exploited to achieve a block-diagonal form of the Gram matrix and thus reducing computational complexity. In this section, we adapt the procedure described above to get rigorous bounds for systems that are sparse or symmetric. First, we will focus on the sparse case.

\subsection{Hierarchies for sparse systems}

We start by giving a brief introduction to correlative sparsity and the associated SOHS hierarchy. The idea is that one can exploit a sparse \textit{correlation pattern} of the variables, to then construct several, smaller Gram (or moment) matrices, indexed by \textit{cliques} of variables that correlate among each other. This is in stark contrast to the dense hierarchy, where a single PSD block is indexed by monomials between all variables in the problem. Employing the sparse hierarchy leads to more PSD conditions of smaller sizes instead of one larger one. As the dominant computational cost for SDP solvers arises from the largest block in the PSD conditions, employing methods to reduce its dimension is crucial to improve the efficiency of SDP hierarchies.

To start, build the correlation graph $C = \{V,E \}$ with vertices $V = \underline{X}$ given by the operators in the problem, and edges $E$ between them whenever they occur together multiplicatively in the support of the objective $f$, or additively in the constraints in $G$, and $H$.

The next step to obtain the hierarchy consists in preprocessing the correlation graph $C$: A \textit{chordal graph} is a graph in which every cycle of length $k \ge 4$ contains a chord, i.e., an edge joining two non-adjacent vertices of the cycle. An important fact is that every non-chordal graph can be extended to be chordal by adding edges to $E$, which is called a \textit{chordal extension}. Note that there is often more than one way to render a given correlation graph chordal.%, which eventually introduces a tradeoff between computational complexity (smaller cliques) and tightness of the relaxation (larger cliques). 

\begin{figure}[h]
  \centering
  \includegraphics[width=0.95\columnwidth]{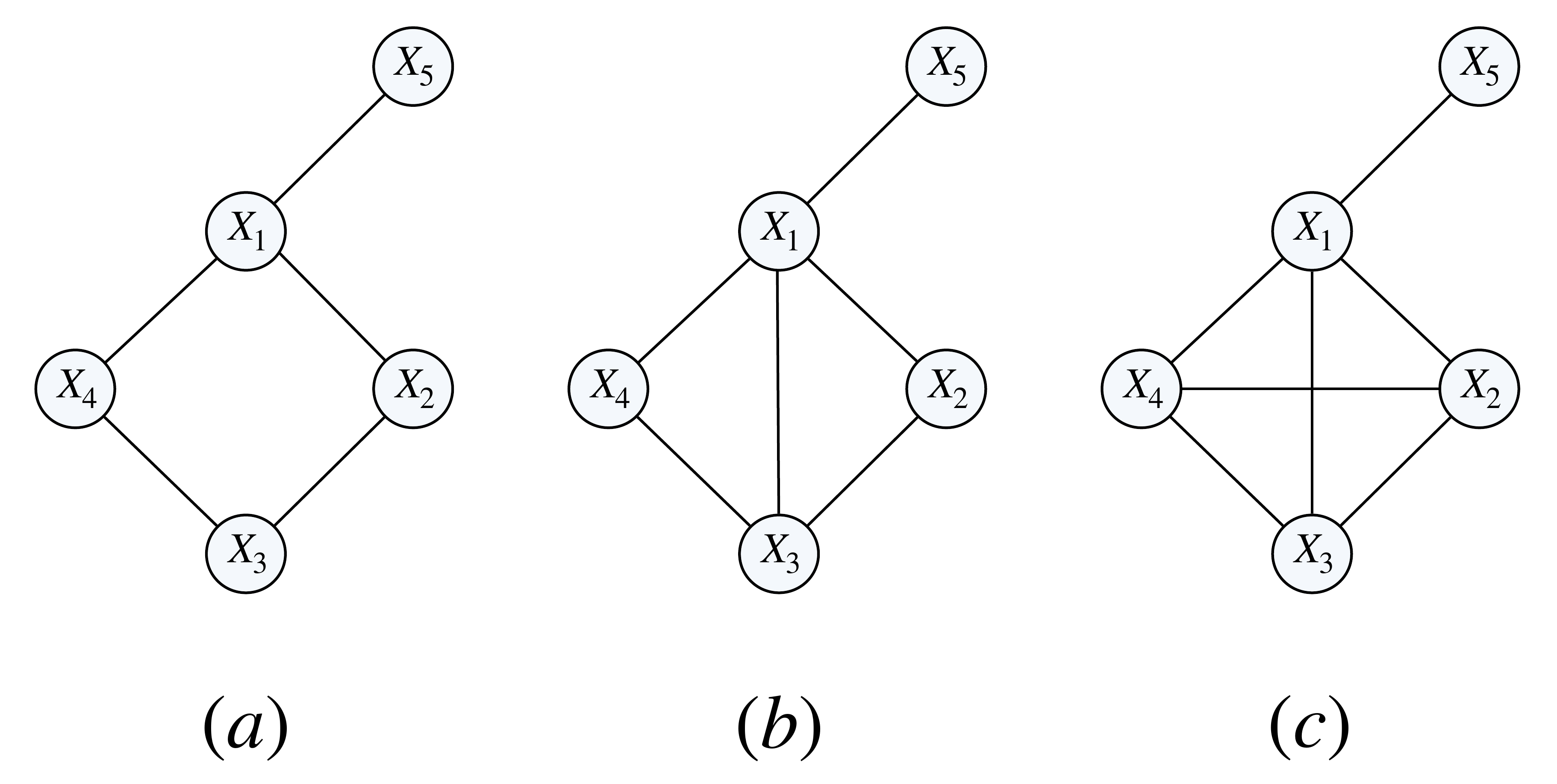}
  \caption{$(a)$ The correlation graph $C$ for $f(\underline{X}) = X_1X_5 + X_1 X_2 + X_2 X_3 + X_3 X_4 + X_4 X_1$. $(b)$ A minimal chordal extension of $C$. $(c)$ Another chordal extension of $C$.}
  \label{fig:chordal_ext}
\end{figure}

After performing the chordal extension, one identifies the so-called \textit{maximal cliques} in the graph: A clique is a subset $C_k \subset V$ such that every pair $X_i, X_j \in C_k$ satisfies $(i,j) \in E$, which is called maximal if no other clique $C_k'$ with $C_k \subsetneq C_k'$ exists. As an example, in Figure \ref{fig:chordal_ext} $(b)$, the maximal cliques can be identified as $C_1 = \{X_1, X_2, X_3 \}$, $C_2 = \{X_1, X_3, X_4\}$, and $C_3 = \{X_1, X_5\}$. When opting for the chordal extension $(c)$ instead, one can identify the two cliques $C_1 = \{ X_1, X_2, X_3, X_4\}$ and $C_2 = \{X_1, X_5\}$. We note here that, by construction of the correlation graph, each constraint features variables from a single clique. Hence, we define $J_k \subseteq [m_G]$ as the set of indices $i$ of inequality constraints contained in $C_k$. Furthermore, let $N_C$ denote the number of maximal cliques, $n_k = |C_k|$ the cardinality of the clique, and $\mathbb{K}\langle \underline{C_k} \rangle$ the $*$-algebra spanned by $C_k$.

Having constructed the chordal extension and identified the maximal cliques, one can define the $d$-th relaxation level of the sparse SOHS hierarchy:
\begin{align}
    &\lambda_d^{\mathrm{sp}} = \inf_{\lambda \in \mathbb{R},\, f_{k,\ell},\, u^{(k)}_{ij} \in \mathbb{K}\langle \underline{C_k} \rangle/\mathcal{I} } 
    \qquad \lambda  \label{eq:relaxation_sparse} \\
     & \quad \, \, \, \, \mathrm{s.t.}  \, \, \, \, \, \, \, \, \,\lambda \revise{\oneop{1}} - f  = \nonumber \\
     & \sum_{k=1}^{N_c} \left( \sum_{\ell} f_{k,\ell}^{*} f_{k,\ell} \;+\; \sum_{i \in J_k} \sum_{j} (u^{(k)}_{ij})^*\, g_i \, u^{(k)}_{ij} \right) \, \, \mathrm{mod}\,\mathcal{I} \nonumber \\
    & \,  \, \qquad \mathrm{deg}(f_{k,\ell}) \leq d, \,\, \mathrm{deg}\!\left(u^{(k)}_{ij}\right) \leq \left\lfloor \frac{d - \mathrm{deg}(g_i)}{2} \right\rfloor. \nonumber
\end{align}
Similarly to the dense hierarchy \ref{eq:relaxation}, the first condition characterizes membership in the associated \textit{sparse truncated quadratic module}, 
which according to the sparse version \cite{klep_sparse} of Theorem \ref{thm:heltonmc} certifies positivity on $K$. The main difference to the dense case is that the SOHS splits over contributions per clique, i.e., summands $f_{k,l}$ and $u_{i,j}^{(k)}$ stemming from the clique-wise algebras $\mathbb{K}\langle \underline{C_k }\rangle/\mathcal{I}$. Hence, when expressing the SOHS in terms of PSD Gram matrices, one is left with $N_C$ principal Gram matrices $G_0^{(k)}$ indexed by the variables in $C_k$, as well as $m_G$ localizing-type Gram matrices associated to the $u_{ij}^{(k)}$, equivalently indexed by variables from $C_k$. This manifests the anticipated reduction in the maximal PSD block size by employing the sparse hierarchy. 

As mentioned, for a given correlations graph, there are different ways of constructing chordal extensions, which eventually introduces a tradeoff between computational complexity (smaller cliques) and tightness of the relaxation (larger cliques). Depending on the chosen chordal extension of the correlation graph, one obtains PSD constraints on larger or smaller blocks, where the dense hierarchy \ref{eq:relaxation} is recovered in the case of a single clique $C_1 = \underline{X}$. This introduces a second dimension to the hierarchy, i.e., the degree of the chordal extension, where more complete graphs are associated to larger computational complexity while providing tighter bounds. Intuitively, allowing for SOHS summands from a bigger space might let one find more elaborate expressions to certify a tighter bound. It has been shown in \cite{klep_sparse} that in the case of a chordal extension, one recovers the dense asymptotic convergence results, $\lambda_{\infty}^{\mathrm{sp}} = \lambda_{\max}$. However, at each finite relaxation order, the dense bounds lower bound the sparse bounds, $\lambda_d \leq \lambda_d^{\mathrm{sp}}$.

\subsection{\revisegreen{Exact} certificates for sparse systems}

Now that we have introduced the numerical sparse SOHS certificates, we are going to establish a rigorous way on how to extract a certified upper bound from the corresponding numerical data. The rounding–projection–lifting strategy developed for the dense hierarchy
extends naturally to such relaxations which exploit correlative sparsity. Let the sparse numerical certificate be denoted by the tuple $(\lambda_d^{\mathrm{sp}},\{G_0^{(k)}\}_{k=1}^{N_C},\{G_i^{(k)}\}_{i=1}^{m_G})$.

The numerical certificate (possibly over a quotient space) then reads as
\begin{equation}
\label{eq:numcertsparse}
    \revise{L_d} 
    \approx  \sum_{k} \left(v_d^{(k)}\right)^*G_0^{(k)}v_d^{(k)} \quad \mathrm{mod}\,\mathcal{I},
\end{equation}
where $v_d^{(k)}$ denotes the monomial basis vector of $\mathbb{K}\langle \underline{C_k}\rangle_d / \mathcal{I}$, and the LHS $L_d$ includes the inequality localizing contributions $G_i^{(k)}$ exactly as in the dense case. The normal form $\mathcal{N}$ can then be expressed on the monomial basis of each clique
\begin{equation}
    \mathcal{N}: (\alpha,\beta)_k \mapsto \sum_t n_{\alpha\beta k}^t t,
\end{equation}
where reductions occur in the global quotient $\mathbb{K}\langle \underline{X} \rangle/\mathcal{I}$, so reduced words $t$ can receive contributions from multiple cliques. We search again for rational correction matrices $\Delta^{(k)}$ such that the certificate holds exactly for all reduced word coefficients $t$:
\begin{equation} \small
    \mathcal{N}(\revise{L_d})_t
    = \mathcal{N}\left( \sum_{k} \left(v_d^{(k)}\right)^*(G_0^{(k)} + \Delta^{(k)})v_d^{(k)} \right)_t
\end{equation}
For that, one again defines the residuals per reduced words, now involving Gram matrices associated to all cliques:
\begin{equation}
\label{eq:sparseresiduals}
    r_t  =     \mathcal{N} \left(\revise{L_d}\right)_t
    - \sum_{\alpha,\beta,k} n_{\alpha \beta k}^t \left(G_0^{(k)} \right)_{\alpha,\beta},
\end{equation}
where $(\alpha,\beta,k)$ runs over monomial pairs $(\alpha,\beta)$ belonging to clique $k$. Following exactly the same argument as in Theorem \ref{thm:frobopt} but summing over tuples involving the clique index $(\alpha,\beta,k)$ let us derive the Frobenius optimal projection in the sparse case.

\begin{lemma}[Frobenius optimal sparse Gram projection]
\label{thm:sparsefrobopt}
 Let $(\alpha^*\beta)_k \in \mathbb{K}\langle \underline{X} \rangle_{2d}$ be the word corresponding to indices $(\alpha,\beta)$ of Gram matrix $G_0^{(k)}$.
Let
\begin{equation}
    (\Xi)_{t,s} = \sum_{\alpha,\beta,k}n_{\alpha\beta k}^tn_{\alpha\beta k}^s.
\end{equation}
If $\Xi$ is invertible, the Frobenius optimal correction reads
\begin{equation}
    \Delta^{(k)}_{\alpha,\beta} = \sum_t (\Xi^{-1}r)_tn_{\alpha \beta k}^t,
\end{equation}
where for binomial $\mathcal{N}$ the Frobenius optimal correction reads
\begin{equation}
    \Delta^{(k)}_{\alpha,\beta} = \frac{r_{\mathcal{N}(\alpha^*\beta)}}{n_\mathcal{I}(\alpha,\beta,k)}
\end{equation}
with
\begin{equation} \small
    n_{\mathcal{I}}(\alpha,\beta,k) = \#\{(\gamma,\delta,l): \mathcal{N}((\alpha^*\beta)_k) = \mathcal{N}((\gamma^*\delta)_l) \}.
\end{equation}
\end{lemma}
The proof follows exactly the same argument as the proof of Lemma \ref{thm:frobopt} and is hence not presented here.

Intuitively, this result tells us how to update the entries of Gram matrices  associated to different cliques, such that the projected Gram matrices deviate the least from their numerical counterpart. Again, this is a natural choice because we expect the least deviation of the eigenvalues of the numerical Gram matrices, hence minimizing the bound loss associated to our scheme.

In the case of the projected Gram matrices leaving the PSD cone, we perform a similar lifting of $\lambda$ as in the dense case, where we again denote the rounded, projected, and lifted matrices with $\mathcal{P}_+(\tilde{G}_0^{(k)})$. 

\begin{theorem}[Sparse certified bounds]
\label{thm:sparsebounds}
    Assume an inaccurate numerical certificate $(\lambda_d^{\mathrm{sp}},\{G_0^{(k)}\}_{k=1}^{N_C},\{G_i^{(k)}\}_{i=1}^{m_G})$ with the Gram matrices $\{\mathcal{P}(G_0^{(k)})\}$ admitting lower spectral bounds $\mu^k_{\mathrm{min}}$ after rounding and projecting. Furthermore, assume that all of the variables are subject to one of the constraint sets $\{c_i\}$ from Theorem \ref{thm:theorem2}. Then, 
    \begin{equation} \small 
        \lambda^{\mathrm{sp,rat}}_d = \tilde{\lambda}_d^{\mathrm{sp}}+\delta_d^{\mathrm{sp}}:=\tilde{\lambda}_d^{\mathrm{sp}} - \sum_{k} \min\{\mu_{\mathrm{min}}^k,0\}\cdot s_{n_k,d}^{\mathcal{I}} 
    \end{equation}
    where $s^{\mathcal{I}}_{n_k,d} = \dim(\mathbb{K}\langle C_k\rangle_d/\mathcal{I})$, constitutes an exact upper bound to the maximization problem, i.e., $\lambda^{\mathrm{sp,rat}} \geq \lambda_{\max}$.
\end{theorem}

\begin{proof}[Proof of Theorem \ref{thm:sparsebounds}]
Let $\{C_k\}_{k=1}^{N_C}$ be the maximal cliques of a chordal extension of $C$. 
For each $k$, let $v_d^{(k)}$ be the monomial basis of $\mathbb{K}\langle C_k\rangle_d/\mathcal{I}$ and $J_k$ the indices of constraints supported on $C_k$.

Assume an exact pre-certificate
\begin{equation}
\label{eq:numprecert}
    \revise{L_d} 
    =  \sum_{k} \left(v_d^{(k)}\right)^*\mathcal{P}(\tilde{G}_0^{(k)})v_d^{(k)} \quad \mathrm{mod}\,\mathcal{I}.
\end{equation}
For blocks $\mathcal{P}(\tilde{G}_0^{(k)})\nsucceq0$, Lemma \ref{thm:boundlowering} with $\varepsilon = - \mu_{\mathrm{min}}^k$ on variables of $C_k$ gives
\begin{align}
-\mu_{\min}^k s_{n_k,d}^{\mathcal{I}} \revise{\oneop{1}}&=-\mu_{\min}^k(v_d^{(k)})^\ast v_d^{(k)}+\sum_i r^{(k)\ast}_ir_i^{(k)}+  \\ 
&\sum_{i,j} q_{ij}^{(k)\ast} c_i\,q_{ij}^{(k)} \quad \mathrm{mod}\,\mathcal{I}. \nonumber
\end{align}
Fixing $k$ and adding one of those expressions to \ref{eq:numprecert} yields a PSD update as in the dense case
\begin{equation}
\mathcal{P}_+(\tilde G_0^{(k)})= \mathcal{P}(\tilde G_0^{(k)})-\mu_{\min}^k \revise{\oneop{1}} +R_0^{(k)}\succeq 0,
\end{equation}
and raises $\tilde{\lambda}_d$ on the LHS by $-\mu_{\min}^k s_{n_k,d}^{\mathcal{I}}$.

Aggregating over all $k$ with $\mu_{\min}^k<0$ gives
\begin{align}
\lambda_d^{\mathrm{sp,rat}}-f=&\sum_{k=1}^{N_C}(v_d^{(k)})^\ast \mathcal{P}_+(\tilde G_0^{(k)})v_d^{(k)}+ \\&\sum_{k}\sum_{i\in J_k}\sum_j (p_{ij}^{(k)})^\ast g_i p_{ij}^{(k)} \in \mathcal{K}^{\mathrm{sp}}_{2d}, \nonumber
\end{align}
which implies the claim.
\end{proof}
We note that if $\exists k:\mathcal{P}(\tilde{G}_0^{(k)})\succ 0$, the bounds can also be tightened clique-wise as in Theorem \ref{thm:boundtightening}, which can partly compensate negative eigenvalues from other cliques.

\subsection{\revisegreen{Exact} certificates for symmetric systems}

The structure of the numerical certificate \eqref{eq:numcertsparse} is not exclusive to sparse relaxations, but also occurs in the case of block diagonal Gram matrices $G_0$. As block-diagonalizable Gram matrices correspond to associated symmetries, one can adapt Theorem \ref{thm:sparsebounds} to also be able to obtain certified bounds on symmetry-reduced problems of the corresponding form.

\begin{lemma}[Frobenius optimal Gram projection over polynomial bases]
\label{thm:froboptsym} Assume that due to the symmetry of a problem, the underlying Gram matrix $G_0$ can be block-diagonalized by a transformation $U$ into $K$ blocks $G_0^{(k)}$ of maximal size $N$. Depending on $U$, each block will be indexed by a \textit{non-monomial} basis, e.g., polynomials $\{b_i^{(k)} \}_{i=1}^{N}$, which can be expressed in the monomial basis of the full space $\{v_{\alpha} \}_{\alpha=1}^{K N}$as
\begin{equation}
    b_i^{(k)} = \sum_{\alpha}b_{i,\alpha}^{(k)}v_\alpha,
\end{equation}
where $b_{i,\alpha}^{(k)}$ are the matrix elements of $U$.
Assuming an ideal $\mathcal{I}$ generating a normal form expressed on the monomial basis as in \eqref{thm:sparsebounds}, the corresponding symmetry-adapted non-negativity certificate reads as
\begin{align}
    \mathcal{N}(f - \lambda \revise{\oneop{1}}) &= \mathcal{N}\left( \sum_{i,j,k} (b_i^{(k)})^*G_{ij}^{(k)}b_j^{(k)}\right) \\
    &= \sum_{i,j,k,\alpha,\beta} (b_{i,\alpha}^{(k)})^*G_{ij}^{(k)}b_{j,\beta}^{(k)} \,\mathcal{N}(v_{\alpha}^*v_{\beta}), \nonumber \\
    \mathcal{N}(f - \lambda \revise{\oneop{1}})_t &= \sum_{i,j,k,\alpha,\beta} (b_{i,\alpha}^{(k)})^*G_{ij}^{(k)}b_{j,\beta}^{(k)} \,n_{\alpha \beta k }^t := \sum_{i,j,k} \tilde{n}_{ijk}^t G_{ij}^{(k)}, \nonumber
\end{align}
which has the same form as in Theorem \ref{thm:sparsebounds}. Defining the matrix $\Xi$ accordingly, the Frobenius optimal projection for the Gram blocks indexed by polynomial bases is given by
\begin{equation}
    \Delta^{(k)}_{i,j} = \sum_t (\Xi^{-1}r)_t\tilde{n}_{ijk}^t.
\end{equation}
\end{lemma}
The proof follows the same arguments as the proof of Lemmas \ref{thm:frobopt} and \ref{thm:sparsefrobopt}.

Note that to establish a certified bound in the case of $\mathcal{P}(\tilde{G}_0^{(k)})\nsucceq0$ , one cannot simply apply Lemma \ref{thm:boundlowering} and lift the Gram blocks independently as in the sparse case, as they are not indexed by the monomial basis in which constraints $\{c_i\}$ are expressed. Instead, one has to lift the full block-diagonal ambient Gram matrix $\mathcal{P}(\tilde{G}_0):=\mathrm{diag}\left(\mathcal{P}(\tilde{G}_0^{(1)}),..,\mathcal{P}(\tilde{G}^{(K)}_0)\right)$, where eigenvalue computations can still be performed block-wise due to the unitary invariance of the matrices' spectra.

\begin{theorem}[Certified bounds  in symmetry-adapted bases] \label{thm:symbounds}
Assume an inaccurate numerical certificate $(\lambda_d,\{G_0^{(k)}\}_{k=1}^{K},\{G_i^{(k)}\}_{i=1}^{m_G})$ to an NPO where $G_0^{(k)}$ are blocks indexed by polynomials $\{b_i^{(k)}\}_{i=1}^{N}$ that stem from block-diagonalization of a Gram matrix in the monomial basis, and assume variables are subject to one of the constraint sets $\{c_i \}$ from Theorem \ref{thm:theorem2}. Let 
\begin{equation}
    \mathcal{P}(\tilde{G}_0^{(k)}) = \tilde{G}_0^{(k)} + \Delta^{(k)} \quad k\in[K]
\end{equation}
be the set of blocks after rounding and projecting according to Lemma \ref{thm:froboptsym} with \revisegreen{certified bounds to} minimal eigenvalues $\mu_{\mathrm{min}}^k$. Then, 
\begin{equation}
    \lambda_d^{\mathrm{rat}} = \tilde{\lambda}_d + \delta_d  = \tilde{\lambda}_d -  \min_k(\{\mu_{\mathrm{min}}^k,0 \} )\cdot s_{n,d}^{\mathcal{I}}
\end{equation}
constitutes a certified upper bound to the maximization problem, i.e., $\lambda_d^{\mathrm{rat}}\geq \lambda_{\max}$.
\end{theorem}

The proof follows from the invariance of the eigenspectrum under unitary transformations and applying Lemma \ref{thm:boundlowering} to $\mathcal{P}(\tilde{G}_0)$ in the monomial basis. Again, if all blocks fulfill $\mathcal{P}(\tilde{G}_0^{(k)})\succ 0$, the same argument as in Theorem \ref{thm:boundtightening} holds, which lets one further tighten the numerical symmetry-adapted bound.

One has to note here that even though the projection scheme in the sparse and symmetric case exhibit a similar structure, they also strongly differ in terms of computational complexity. In the sparse (and trivially, dense) case, one can show that the rationalization scheme carries \textit{less} computational complexity with respect to the matrix sizes than solving the underlying SDP. However, in the symmetry-adapted case, computing the Frobenius optimal projection may exceed the cost of SDP solving, depending on the block partitioning of the initial Gram matrix $G_0$. A discussion of the complexity in both sparsity and symmetry-adapted cases is given in Appendix \ref{sec:complexity}.

As a second remark we note that in the case of sparse cliques with sufficiently small overlap, one finds
\begin{equation}
    s^{\mathcal{I}}_{n,d} > \sum_k s_{n_{C_k},d}^{\mathcal{I}}
\end{equation}
in which case the loss in bounds through rationalization in the dense case, $\delta_d$, often exceeds the sparse loss, $\delta_d^{\mathrm{sp}}$. This fact can lead to the counterintuitive phenomenon that rationalized sparse bounds become tighter than the corresponding rationalized dense bounds, $\lambda_d^{\mathrm{sp,rat}} < \lambda_d^{\mathrm{rat}}$, as it is illustrated in the next section.

\section{Maximal quantum violation of Bell inequalities}
\label{sec:belllineq}

As a first benchmark for our certification method, as well to demonstrate its necessity, we apply it to the problem of determining the maximal quantum violation of Bell inequalities. All of the numerical experiments in this work were modeled in Julia using the libraries NCTSSOS \cite{wang_sparse} and QMBCertify \cite{groundstate}, and the resulting SDPs were solved with the commercial solver MOSEK \cite{mosek}\revisered{, using $10^{-8}$ tolerances for primal feasibility, dual feasibility, and the relative duality gap}. The solver choice is critical for numerical reliability, where MOSEK was selected for its state-of-the-art interior-point implementation and its widespread adoption in semidefinite optimization. \revisegreen{Certified} bounds on \revisered{minimal} eigenvalues \revisered{$\mu_{\min}$ of the \revisegreen{rounded and projected} Gram matrices (as explained in Section \ref{sec:liftingpositivity})} have been computed using \texttt{Arblib.jl} \cite{arblib}, \revisegreen{where exact rational representations of the certified enclosures returned by \texttt{Arblib} are used for subsequent exact arithmetic operations necessary for the algebraic lifting. }
The code to obtain certified bounds and reproduce our results is integrated in the Julia libraries NCTSSOS\footnote{\url{https://github.com/wangjie212/NCTSSOS/tree/master/examples}} and QMBCertify\footnote{\url{https://github.com/wangjie212/QMBCertify/tree/main/examples}} packages for Bell and many-body benchmarks respectively.

As an \revise{illustrative example}, consider the famous Clauser-Horne-Shimony-Holt (CHSH) inequality~\cite{CHSH}
%\begin{equation*}
    $\langle A_0 B_0 \rangle + \langle A_0 B_1 \rangle + \langle A_1 B_0 \rangle - \langle A_1 B_1 \rangle \leq 2. $
%\end{equation*}
Here, $\langle A_i B_j \rangle$ denote expectation values of products of dichotomic observables $A_i$ and $B_j$. Intuitively one can think of a game involving two parties, Alice and Bob, that perform either of the two measurements randomly, trying to maximize the correlations constituting the inequality. Such inequalities characterize the boundaries between \textit{local behaviors}, admitting a description in terms of local hidden variables, and \textit{quantum behaviors}, e.g., correlations achievable by measurements on a joint bipartite quantum system. If a given strategy exceeds the local bound, equal to 2 in the case of the CHSH inequality, one can conclude that no local hidden-variable model can reproduce the corresponding correlations, and hence that nonlocal (quantum or supraquantum) resources were used. A central question along these lines is \textit{how} different the emergent correlations can be for quantum states with respect to local strategies, which is captured in the \textit{maximal quantum violation} of a given Bell expression, known to be equal to $2\sqrt2$ in the CHSH case. 

Finding maximal violations of Bell expressions to bound the set of quantum correlations constitutes the first family of problems for which non-commutative SDP relaxations were exploited~\cite{bounding_quantum}. Such scenarios can be extended to feature an arbitrary number of parties, measurements per parties, and outcomes per measurement, depending on the setting of interest. In this section we focus on two-party Bell expressions featuring $m_A$ and $m_B$ dichotomic measurements on Alice's and Bob's side, respectively. In general, an optimization problem of this class reads
\begin{align}
\lambda_{\max} &=\sup_{\mathcal{H},\ket{\psi}\in\mathcal{H},A_i,B_j\in \mathcal{B}(\mathcal{H})} 
    \quad  \sum_{i=1}^{m_A} \sum_{j=1}^{m_B} c_{ij}\langle A_i  B_j \rangle \nonumber\\
    &+\sum_{i=1}^{m_A} a_{i}\langle A_i  \rangle +\sum_{j=1}^{m_B} b_{j}\langle B_j \rangle\label{eq:bellprob}  \\
    &\qquad\mathrm{s.t.}\qquad \quad \quad \, \, A_i^2 = 1, \qquad i= 1,...,m_A, \nonumber \\
    &\qquad \qquad \qquad \quad  \, \, \, \,B_j^2 = 1, \qquad j= 1,...,m_B, \nonumber \\
    &\qquad \qquad \qquad \quad \, \, \, \,[A_i, B_j] = 0, \nonumber
\end{align}
where the coefficients $\{c_{ij},a_i,b_j\}$ encrypt the type of Bell inequality, and commutation of $A_i$ and $B_j$ is enforced to model the bipartite Hilbert space structure. This problem exhibits the same structure as \eqref{eq:originalprob}, allowing us to make use of the NPA hierarchy and our certification scheme. 

The need for certification already becomes clear in the simplest case of the CHSH inequality: Running hierarchy level $d=2$ for CHSH yields the following (wrong) upper bound:
\begin{equation}
    \lambda_2 \,\approx2.828427122879\ngeq 2\sqrt{2} = \lambda_{\max} ,
\end{equation}
smaller than the known \revisered{maximal} quantum violation. However, the \revisegreen{certified} bound fulfills $\lambda_2^{\mathrm{rat}} > \lambda_{\max}$\revisered{, as illustrated by the following truncated decimal expansion of the exact rational representation:}
\revise{
\begin{equation}
    \lambda_2^{\mathrm{rat}} \revisered{\,\approx\,} \revise{2.828427125402 > 2 \sqrt{2} = \lambda_{\max}.}
\end{equation}
}
\revisered{We emphasize here that while the numerical bound $\lambda_2$ \textit{does} hold when truncated at the $8$th digit, corresponding to the SDP tolerance of the solver,  the main conceptual point of this work is to actually \textit{certify} the soundness of a given bound via an exact algebraic certificate, as it cannot be guaranteed by merely truncating numerical bounds at the expected numerical precision.}

\subsection{Certified bounds to maximum quantum violations of Bell inequalities \texorpdfstring{$A_2-A_{89}$}{}}

To further investigate this behavior, we ran our certification methods on a catalog of 88 inequivalent two-party, dichotomic Bell inequalities, called $A_2-A_{89}$, first introduced in \cite{a89}, for relaxation orders $d \in \{1,2,3 \}$. \revise{These problems constitute adequate benchmark problems for our method as \revisered{accurate numerical reference values have been established for $79/88$ inequalities in the benchmarking set \cite{a89_optimality}}, allowing us to judge the quality of both the numerical bounds as well as their} \revisegreen{certified} \revise{counterparts.} \revisered{Besides this, the catalog carries a manifold of different scenarios and hence Gram matrix sizes, allowing us to test our method on a wide spread of settings.} We start with the analysis of the dense certification scheme.

\revise{As a first question it is interesting to see how the} \revisegreen{certification} \revise{loss
\begin{equation}
    \delta_d :=\lambda_d^{\mathrm{rat}}-\lambda_d=-\min\{\mu_{\min},0\}\cdot s_{n,d}^{\mathcal{I}} 
\end{equation}
is correlated with the quality of the initial numerical certificate: we expect the loss in tightness to be worse if the initial numerical certificate had a large numerical error. To do so,  as a measure of inaccuracy of the numerical certificate at relaxation order $d$, we introduce the 1-norm $\mathcal{D}_d$ of the difference between the numerical left hand side $L_d$
\begin{equation}
    L_d = \lambda_d \revise{\oneop{1}} - f,
\end{equation}
and the numerical right hand side $R_d$ of the certificate
\begin{equation}
    R_d = v_d^* G_0 v_d,
\end{equation}
in terms of reduced basis words indexed by $t$:
}
\begin{equation}
    \mathcal{D}_d := \sum_t |L_d-R_d|_t,
\end{equation}
\revise{which lets one investigate how the} \revisegreen{certification} \revise{loss $\delta_d$ depends on the inaccuracy of the numerical certificate.}

Figure \ref{fig:bounds-1} shows the \revise{behavior} of the bound loss $\delta_d$ depending on the numerical certificate error $\mathcal{D}_d$. One can see empirically that large certificate errors $\mathcal{D}_d$ lead to higher bound losses, which is what we intuitively expect: as established in Section \ref{sec:rational_bounds}, the bound loss depends on two variables, the Gram matrix size $s_{n,d}^{\mathcal{I}}$, and the minimal eigenvalue of the projected matrix $\mu_{\min}$ \revise{, which are both correlated with the certificate error $\mathcal{D}_d$. We expect both $\mu_{\min}$ and $s_{n,d}^{\mathcal{I}}$ to independently affect the bound loss $\delta_d$, which is further examined in the following figures. }

\begin{figure}[!ht] 
  \centering
  \begin{subfigure}{\columnwidth}
    \centering
    \includegraphics[width=\columnwidth]{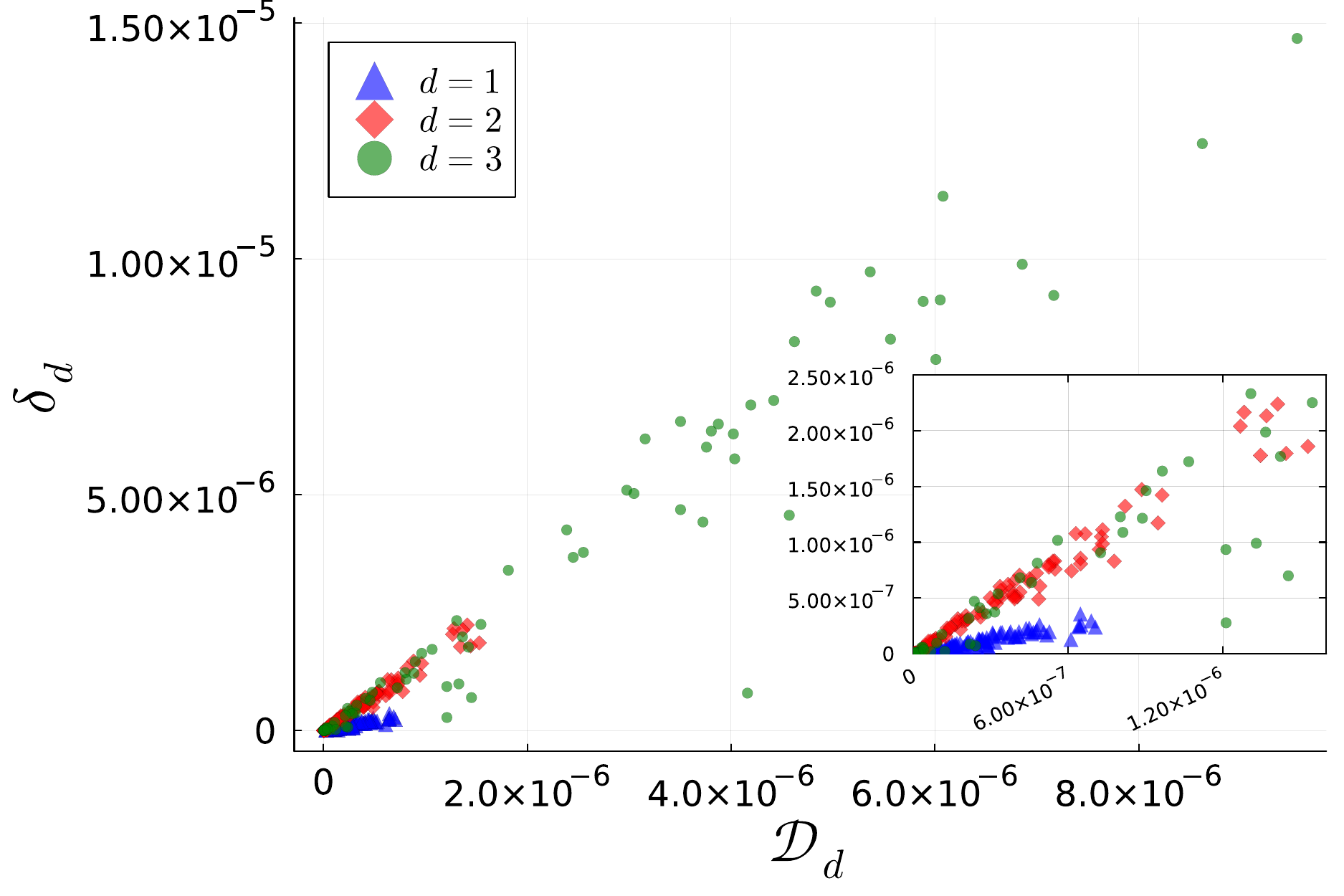}
    \caption{Numerical certificate error $\mathcal{D}_d$ vs. bound loss $\delta_d$. \revise{The worse the accuracy of the initial numerical certificate measured by its 1-norm error, the worse the resulting bound loss $\delta_d$.}}
    \label{fig:bounds-1}
  \end{subfigure}

  \medskip

    \begin{subfigure}{\columnwidth}
    \centering
    \includegraphics[width=\columnwidth]{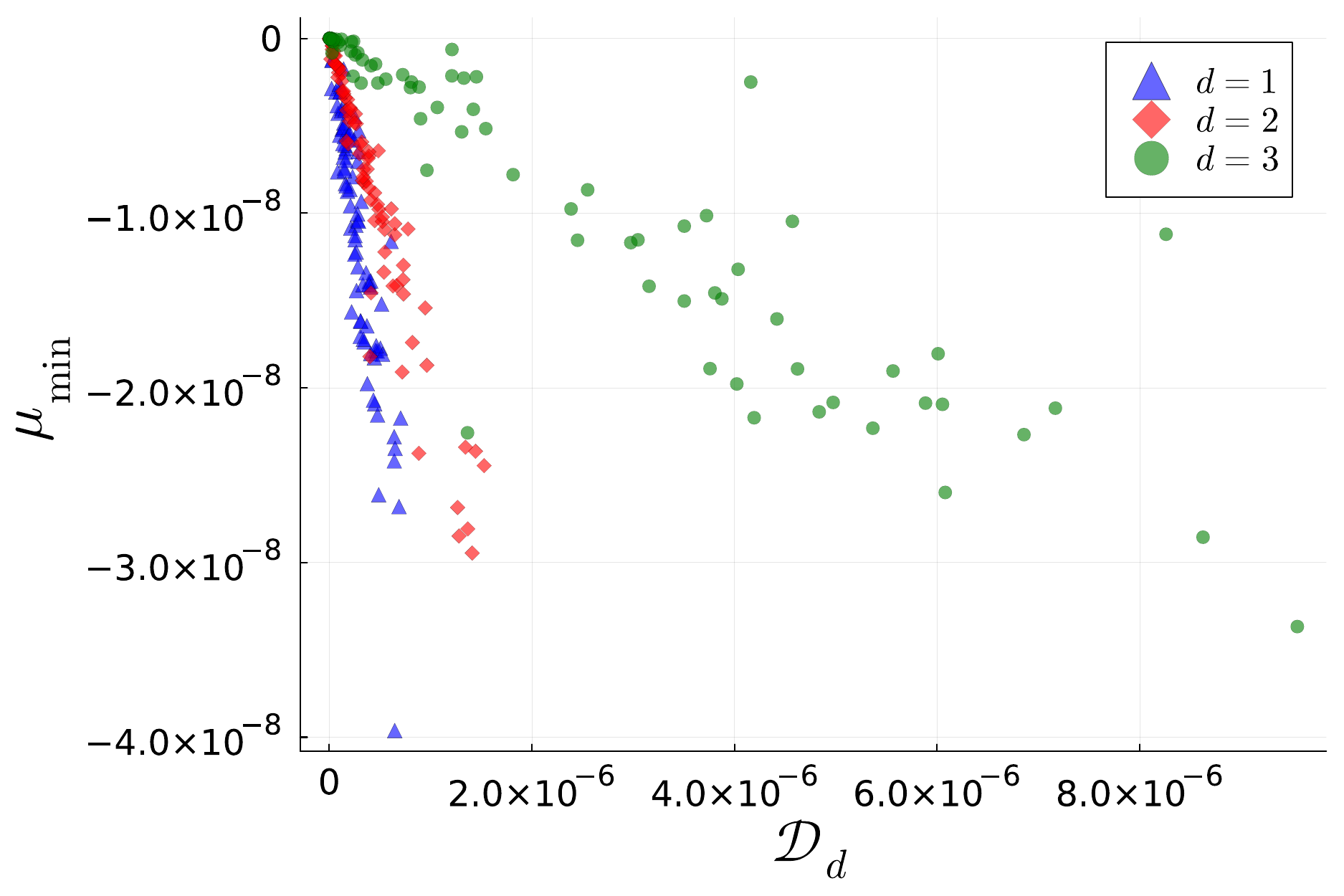}
    \caption{Numerical certificate error $\mathcal{D}_d$ vs. minimal eigenvalue after projection $\mu_{\min}$. \revise{ A large certificate error $\mathcal{D}_d$ leads to large Frobenius distances $\|\tilde{G}_0 -\mathcal{P}(\tilde{G}_0)\|_F$ trying to enforce coefficient equality in the certificate, leading to possibly large distances to the PSD cone after projection.}}
    \label{fig:bounds-b}
  \end{subfigure}

  \medskip

  \begin{subfigure}{\columnwidth}
    \centering
    \includegraphics[width=\columnwidth]{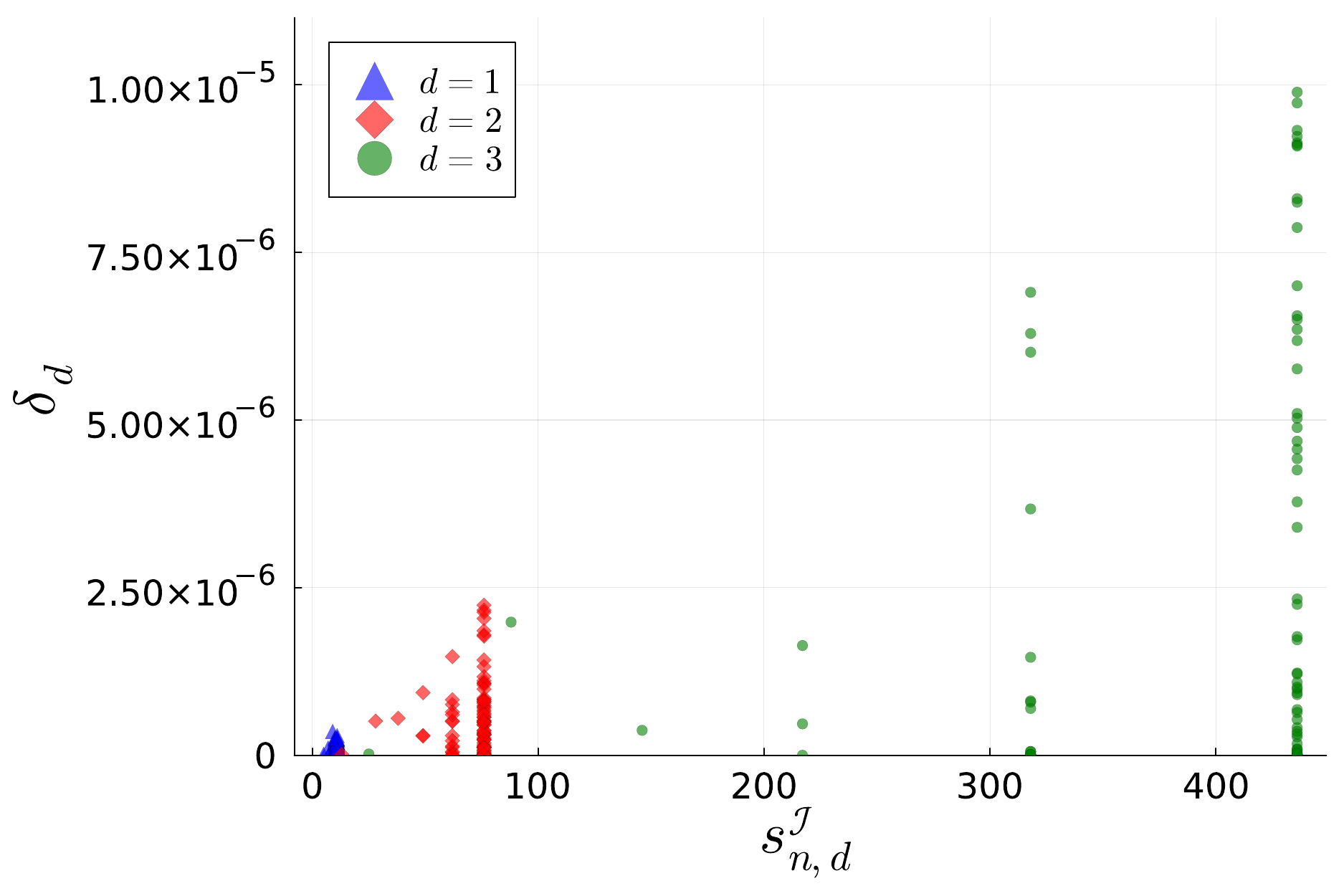}
    \caption{Gram-matrix size $s_{n,d}^{\mathcal{I}}$ vs. bound loss $\delta_d$. \revise{The larger the size of the Gram matrix $s_{n,d}^{\mathcal{I}}$, the larger the possible bound loss $\delta_d$. However, at fixed matrix size, $\delta_d$ is determined by the minimal eigenvalue $\mu_{\min}$}}
    \label{fig:bounds-c}
  \end{subfigure}

  \caption{Numerical results for the certification scheme on Bell inequalities A2-A89}
  \label{fig:bounds-all}
\end{figure}

Figure \ref{fig:bounds-b} shows the trend that larger numerical certificate errors produce larger distances of projected matrices $\mathcal{P}(\tilde{G}_0)$ to the PSD cone. This is expected, as large certificate errors lead to large Frobenius distances between $\mathcal{P}(\tilde{G}_0)$ and $G_0$, possibly projecting the matrices further out of the PSD cone to achieve coefficient equalities. \revise{Hence, we expect the minimal, negative, projected eigenvalue} \revisegreen{bound} \revise{ $\mu_{\min}$  to grow in magnitude as the numerical certificate error $\mathcal{D}_d$ worsens, yielding larger certification losses $\delta_d$.}

Figure \ref{fig:bounds-c} shows the \revise{behavior} of the bound loss $\delta_d$ with respect to the Gram matrix size $s_{n,d}^{\mathcal{I}}$. As the two are multiplicatively correlated, we see that larger $s_{n,d}^{\mathcal{I}}$ can result in larger $\delta_d$. However, one can also see that at a fixed size $s_{n,d}^{\mathcal{I}}$, the bound loss has a high variance, which is reflected by the variance in certificate errors $\mathcal{D}_d$ and hence minimal eigenvalues $\mu_{\min}$. This indicates that the quality of the numerical certificate is at least as important as the combinatorial growth for the tightness of the certified bounds.

\revise{
To summarize, our first benchmarks indicate that a large numerical certificate error $\mathcal{D}_d$ often leads to a large bound loss $\delta_d$, firstly due to the correlation of $\mathcal{D}_d$ and $\delta_d$ with the size $s_{n,d}^{\mathcal{I}}$ of the Gram matrix, and secondly through the larger Frobenius distance $\|\mathcal{P}(\tilde{G}_0)-\tilde{G}_0 \|_F$ induced through projection, enabling larger negative eigenvalues $\mu_{\min}$.}

In \cite{a89_optimality}, the authors established that relaxations for the majority of Bell inequalities $A_2-A_{89}$ \revisered{match constructive lower bounds up to the numerical accuracy of $10^{-8}$} at latest at order $d=3$, where the  \revisered{obtained reference maximal} violations \revisered{$\lambda_{\mathrm{ref}}$} are computed in \citep{a88_exact, a89_exact}. Splitting the inequalities into three sets, $\mathcal{B}_i,i\in \{1,2,3\}$ according to \revisered{the hierarchy level at which the lower and upper bounds numerically coincide} (sorting \revisered{the remaining inequalities} to $\mathcal{B}_3$), allows us to investigate the \revise{behavior} of the obtained numerical bounds, as well as the \revisegreen{certified} bounds, with respect to these \revisered{reference} maximal violations. As $|\mathcal{B}_1|=2$, we restrict the analysis to the cases of $\mathcal{B}_2$ and $\mathcal{B}_3$.

\begin{figure}[ht] 
  \centering
  \begin{subfigure}{\columnwidth}
    \centering
    \includegraphics[width=\columnwidth]{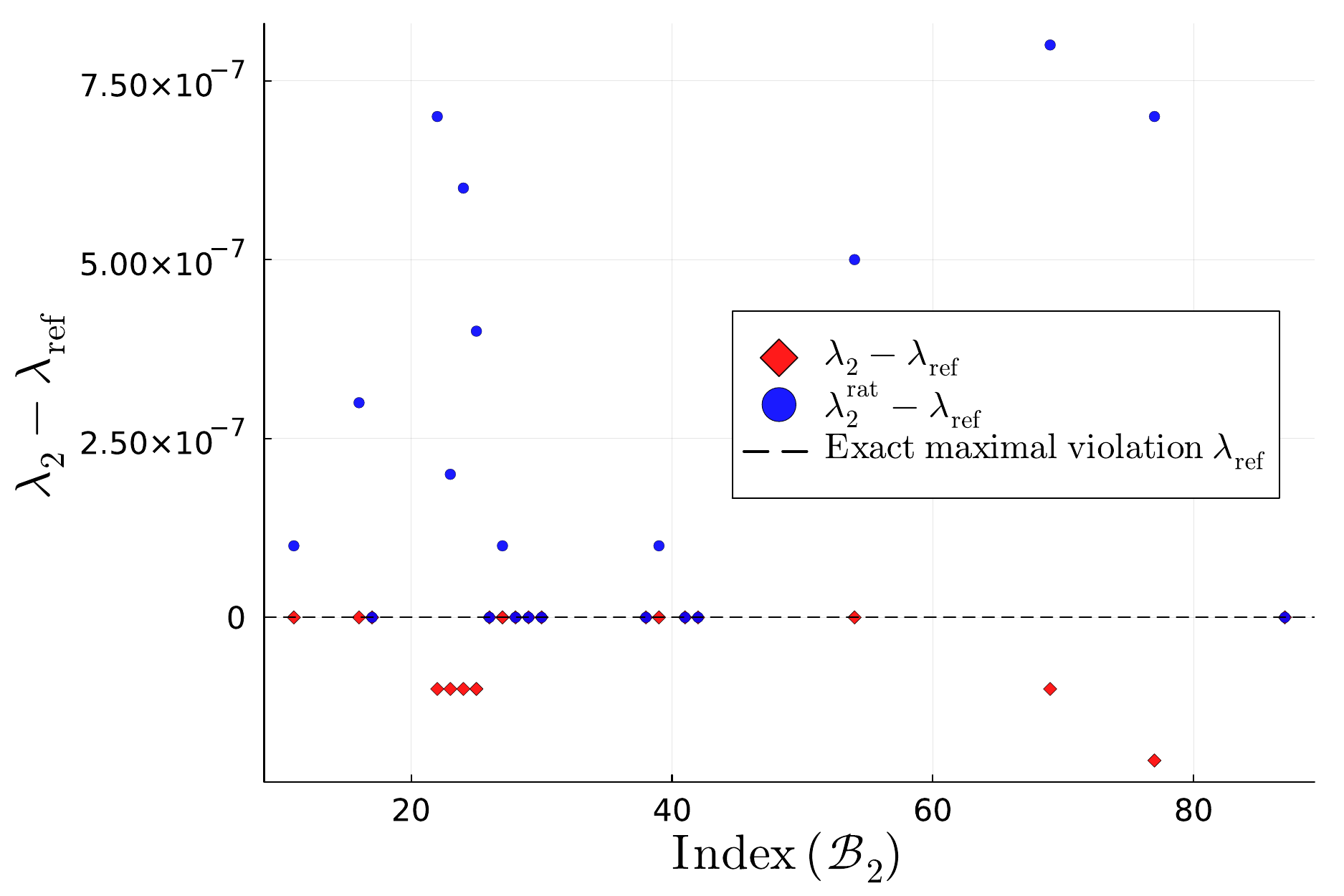}
    \caption{}
    \label{fig:bounddiff_b2}
  \end{subfigure}

  \medskip

    \begin{subfigure}{\columnwidth}
    \centering
    \includegraphics[width=\columnwidth]{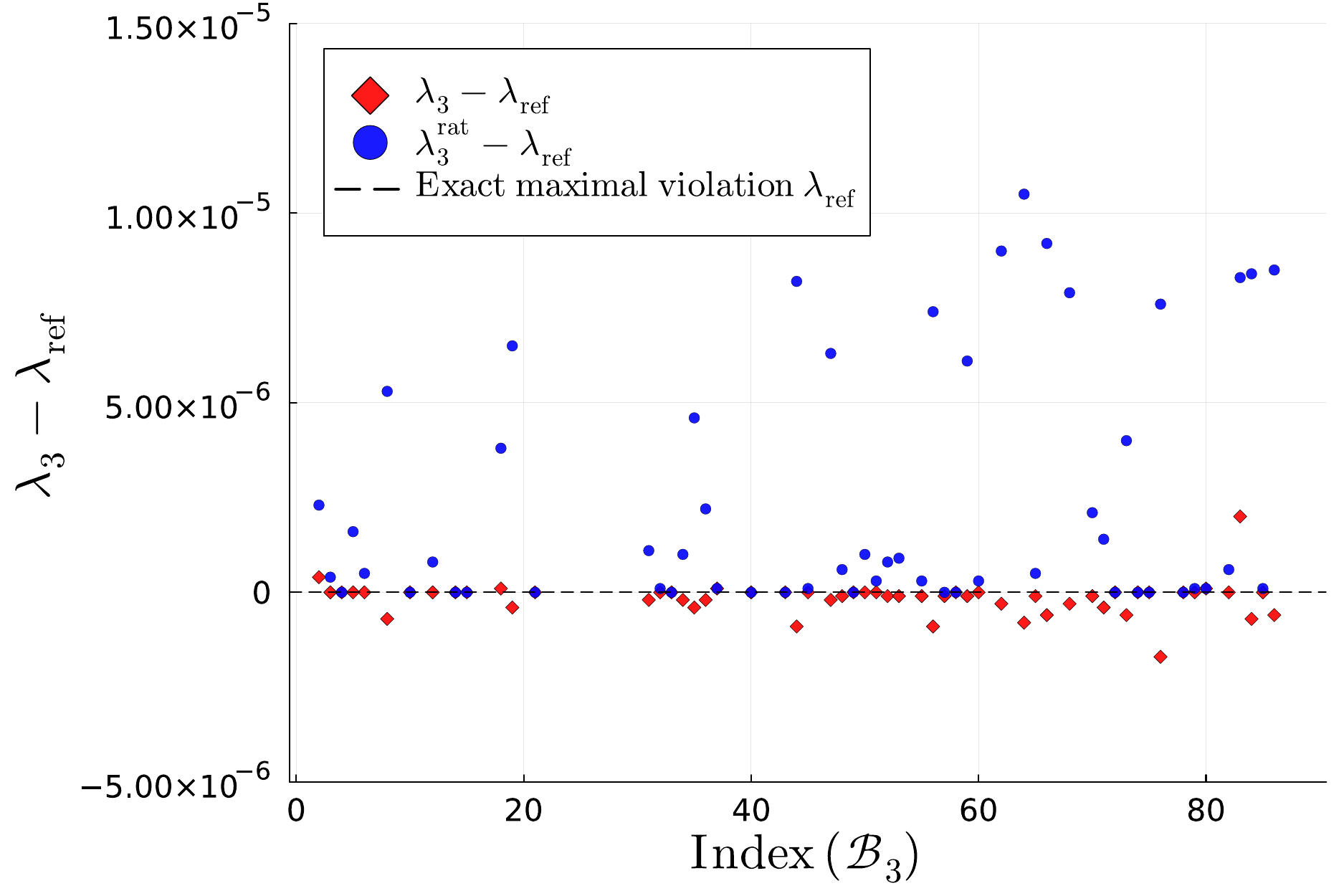}
    \caption{}
    \label{fig:bounddiff_b3}
  \end{subfigure}

  \caption{Numerical bounds $\lambda_d$ and rationalized bounds $\lambda_d^{\mathrm{rat}}$ relative to \revisered{reference}  violations $\lambda_{\mathrm{ref}}$ for $\mathcal{B}_2$ (a) and $\mathcal{B}_3$ (b). }
  \label{fig:boundsdiffs_exact}
\end{figure}

Figure \ref{fig:boundsdiffs_exact} shows how the numerical and rational bounds relate to the \revisered{numerical reference} maximal violations\footnote{Bell inequalities with indices $i\in \{7, 9,
 10,
 15,
 28,
 42,
 44,
 48,
 66 \}$ are not displayed because their numerical bounds hold by a significant margin.}. Note that a point in the positive halfspace denotes a bound which is consistent with the reference values, while points in the negative halfspace are necessarily invalid. We see that while a significant fraction of the numerical bounds $\lambda_d$ at both relaxation orders are invalid, all produced rational bounds are proper upper bounds to \revisered{$\lambda_{\mathrm{ref}}$}. 
In fact, we find that 6/20 numerical bounds in $\mathcal{B}_2$ are incorrect, while 26/66 bounds in $\mathcal{B}_3$ are incorrect. This demonstration highlights the necessity of certified bounds as reliable alternatives to their numerical counterparts. 
\subsection{Comparison between the dense and sparse scheme}

While the minimal projected eigenvalue $\mu_{\min}$ inherently depends on the certificate error $\mathcal{D}_d$, and hence on the numerical precision of the solver, one can mitigate the dependence on the Gram matrix size by making use of the sparse hierarchy and associated certification scheme. To do so, we applied the sparse scheme to the same set of Bell inequalities and compared the sparse \revisegreen{certification} bound loss $\delta_d^{\mathrm{sp}}$ to the dense one.

\begin{figure}[h!] 
  \centering
  \begin{subfigure}{\columnwidth}
    \centering
    \includegraphics[width=\columnwidth]{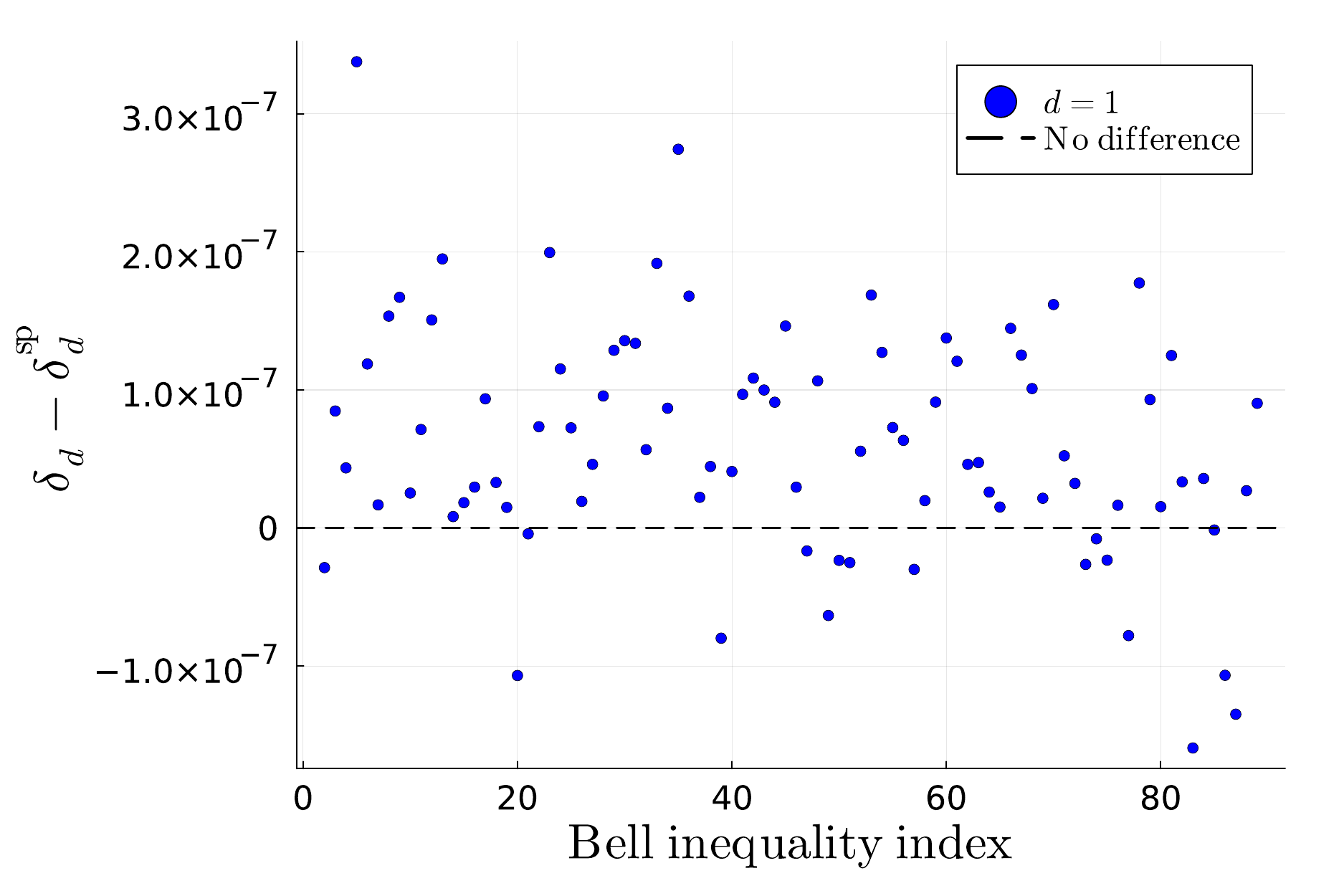}
    \caption{}
    \label{fig:dense_sparse_loss_r1}
  \end{subfigure}

  \medskip

  \begin{subfigure}{\columnwidth}
    \centering
    \includegraphics[width=\columnwidth]{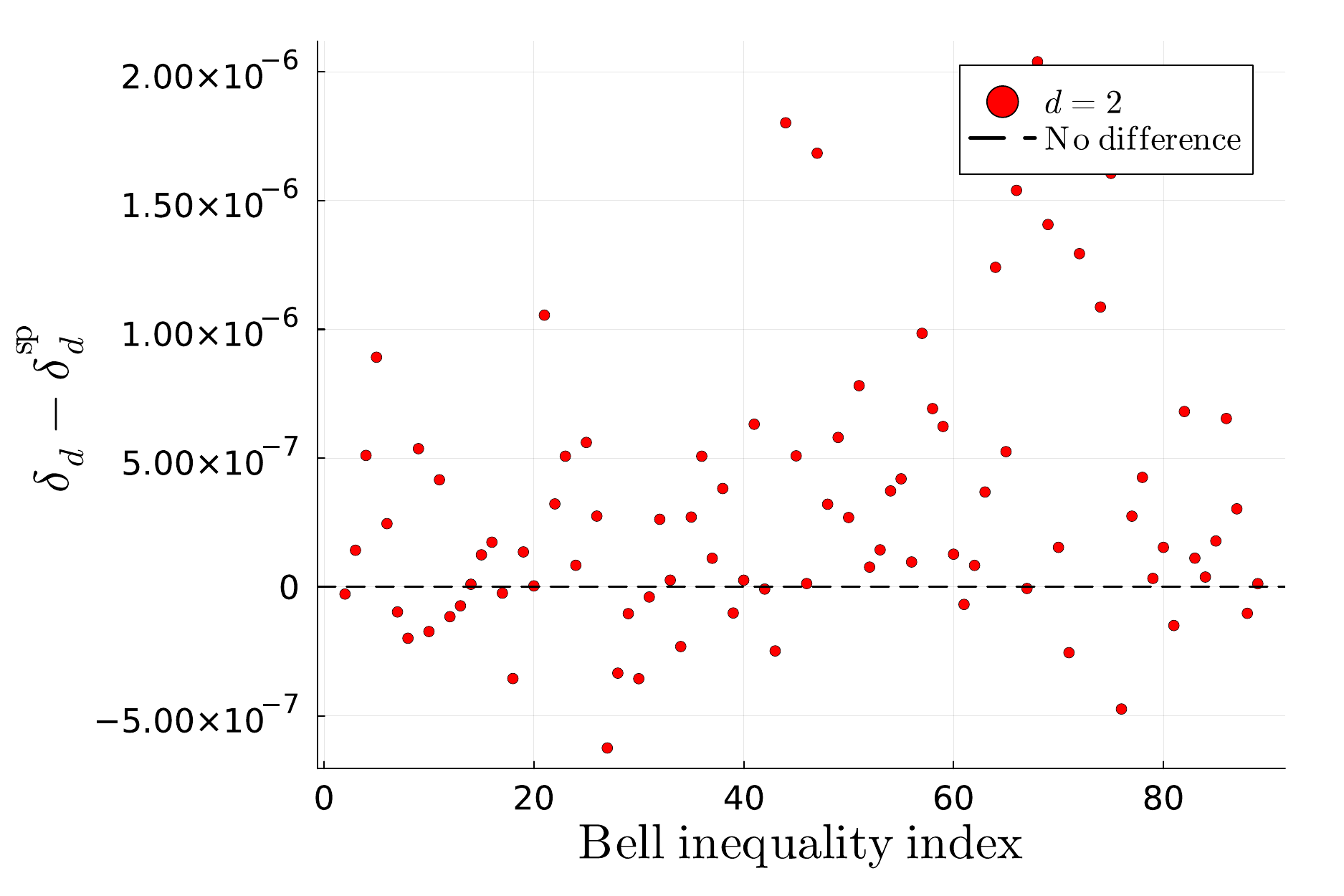}
    \caption{}
    \label{fig:dense_sparse_loss_r2}
  \end{subfigure}

  \medskip

  \begin{subfigure}{\columnwidth}
    \centering
    \includegraphics[width=\columnwidth]{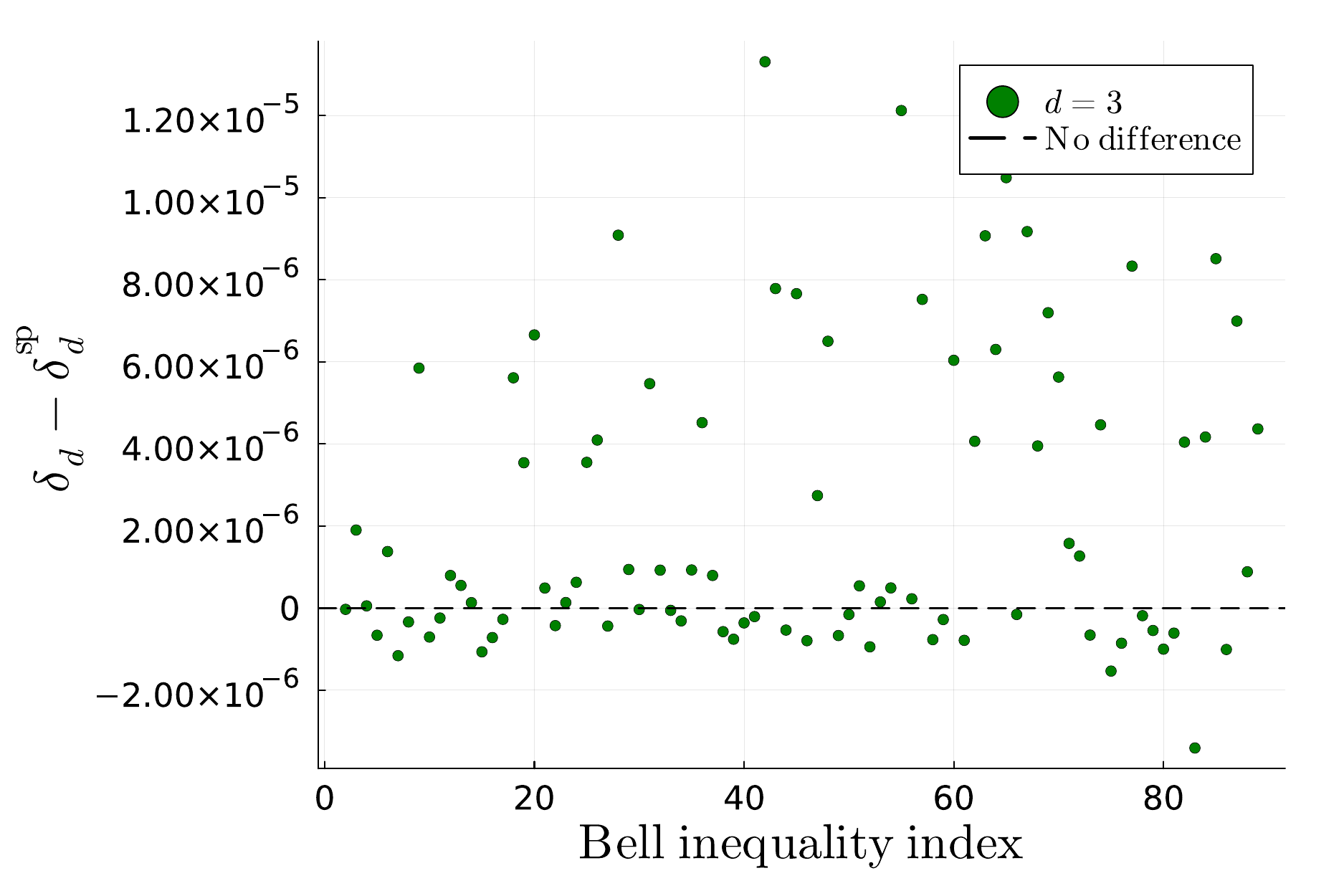}
    \caption{}
    \label{fig:dense_sparse_loss_r3}
  \end{subfigure}

  \caption{Difference in bound losses $\delta_d - \delta_d^{\mathrm{sp}}$ between dense and sparse certification schemes for Bell inequalities $A_2$–$A_{89}$ at relaxation orders $d=1$ (a), $d=2$ (b), and $d=3$ (c).}
  \label{fig:dense_sparse_loss_all}
\end{figure}

Figure \ref{fig:dense_sparse_loss_all} shows the difference between the dense and the sparse bound loss, $\delta_d - \delta_d^{\mathrm{sp}}$ for different relaxation orders $d \in \{1,2,3 \}$. Points in the positive half space hence correspond to experiments where $\delta_d > \delta_d^{\mathrm{sp}}$, i.e., where the sparse certification scheme lead to a smaller bound correction after projection.

As one can see, often times the dense loss $\delta_d$ exceeds the sparse loss $\delta_d^{\mathrm{sp}}$, which, when neglecting the distribution of minimal eigenvalues, can be explained by the often smaller total problem size in the sparse case with respect to the dense case. As discussed in Section \ref{sec:sparse}, by exploiting sparsity, some monomials, 
i.e., those involving variables from different cliques, do not appear in any of the sparse gram matrices, while they do in the dense description. On the other hand, a large variable overlap of cliques could cause the sum of sparse block sizes ($\sum_k s_{n_k,d}^{\mathcal{I}}$) to become larger than the dense block size ($s_{n,d}^{\mathcal{I}}$). However, if the overlap is small enough, the total size will be smaller than in the dense case, which is reflected in the figures. 

This \revise{behavior} can lead to the counterintuitive phenomenon that the certified sparse bounds become \textit{tighter} than their rational dense counterparts, in stark contrast to results from the numerical relaxations. This happens exactly when the excess in $\delta_d$ compensates for the numerical difference $\lambda_d < \lambda_d^{\mathrm{sp}}$. In our experiments we find that at $d=1$, 70/89 rational sparse bounds were tighter than the dense counterparts. At $d=3$ we find that 6/88 sparse bounds were tighter, while at $d=2$ all dense bounds were tighter. This dependency on the relaxation order can be explained by two competing mechanisms, assuming cliques with sufficiently small overlap: On the one hand, the numerical bound difference $\lambda_d - \lambda_d^{\mathrm{sp}}$ is expected to increase with increasing $d$, as the sparse hierarchy fails to capture the increasing amount of clique overlap monomials. On the other hand, the dense Gram matrix size grows more rapidly with $d$, leading to a larger multiplicative factor when computing the certified bounds $\lambda_d^{\mathrm{rat}}$.

\subsection{Tilted CHSH inequality}
As a last benchmark in the non-locality scenario we investigated the performance of our scheme to estimate the maximal quantum violation of a variant of the so called \textit{tilted} CHSH inequality~\cite{tilted_bell}, introduced in \cite{tilted_selftesting}:
\begin{align}
\mathcal{B}_t(\theta_A,\theta_B) &= \sum_{x,y} (-1)^{xy}\langle A_x B_y\rangle + \theta_A\langle A_0\rangle + \theta_B \langle B_0\rangle \nonumber \\ &\le 2 + \theta_A + \theta_B.
\end{align}
%The \textit{tilting} parameters $\theta_A,\theta_B$ are physically motivated and model the detector efficiencies on Alice's and Bob's side respectively, where one recovers the standard CHSH scenario under the assumption of perfectly efficient detectors, $\theta_A = \theta_B = 0$. 

The study of said inequality poses a computationally relevant problem for our framework: On the one hand, this tilted CHSH inequality is self-testing, and the authors provide an analytical solution for its maximal quantum violation. On the other hand, the authors show that this tilted CHSH inequality is \textit{intractable} using the NPA hierarchy as $\theta_A + \theta_B \rightarrow 2$, as increasingly high relaxation levels are needed for convergence. In fact, the authors showed, using the arbitrary-precision solver SDPA-GMP \cite{sdpa-gmp}, that for parameters $\theta_A = \theta_B = 0.999$, even NPA level $d=10$ does not achieve the analytical optimal violation. 
This provides an interesting application for the certification scheme, as it allows us (a) to investigate its performance at high relaxation orders \revise{without finite convergence} and (b) compare computed bounds to the exact maximum quantum violations.

\begin{figure}[ht]
  \centering
  \includegraphics[width=\columnwidth]{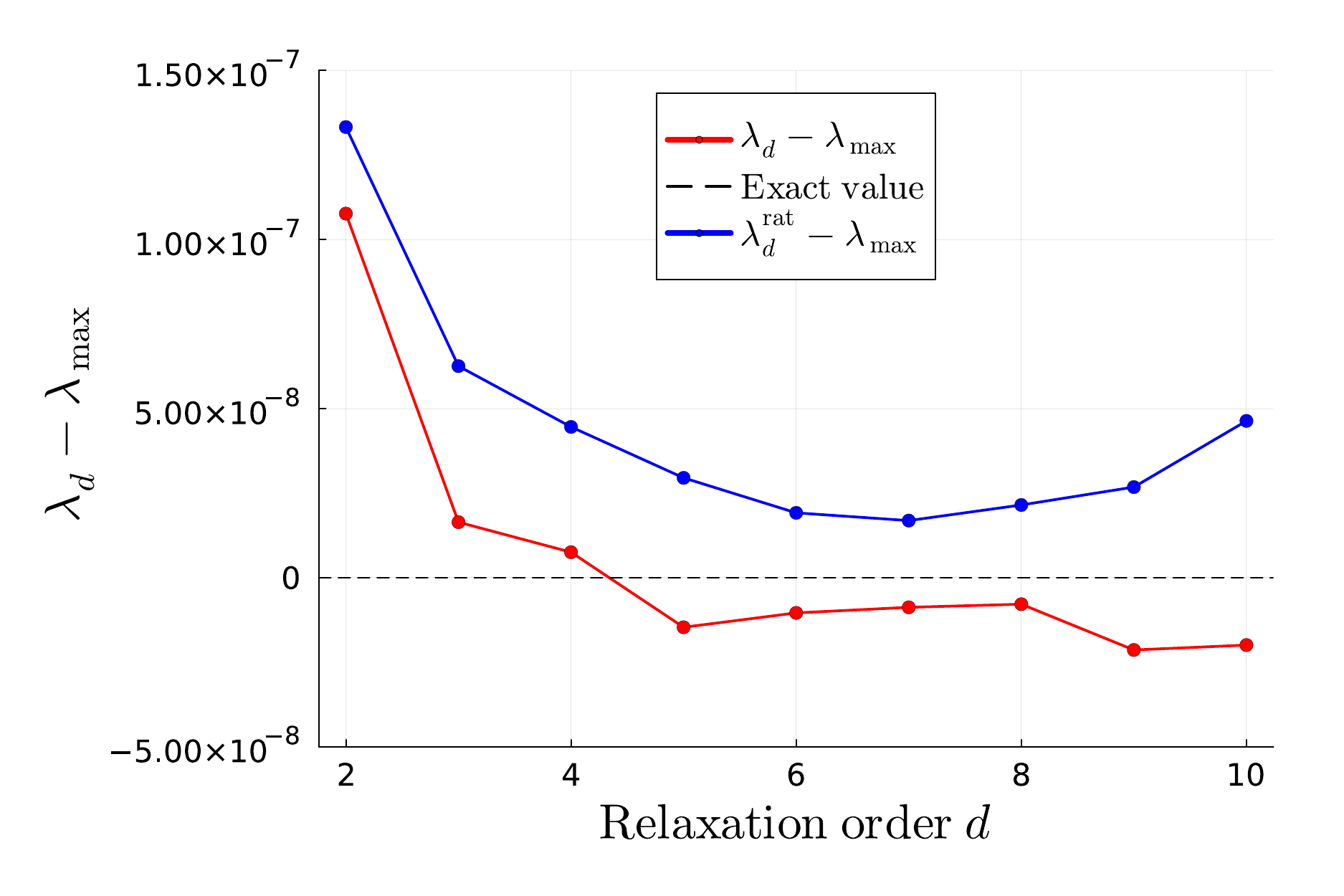}
  \caption{Numerical bounds $\lambda_d$ and certified bounds $\lambda_d^{\mathrm{rat}}$ with respect to the analytical maximal violation $\lambda_{\max}$ of the tilted CHSH inequality with parameters $\theta_A = \theta_B = 0.999$}
  \label{fig:tilted_CHSH}
\end{figure}
As one can see in Figure \ref{fig:tilted_CHSH}, even though high-precision calculations showed that the hierarchy does not converge at the displayed relaxation orders, the numerical upper bounds $\lambda_d$ computed with MOSEK are below the analytical maximum for relaxation orders $d>4$, resulting in wrong bounds, which furthermore do not monotonically decrease. On the other hand, our certification scheme produces valid bounds at all relaxation orders, with the tightest bound achieved at order $d=7$. In this scenario where numerical bounds are stagnating, the certified bounds start to worsen due to the rapid growth of $s_{n,d}^{\mathcal{I}}$ with respect to $d$.

%\textbf{THERE IS SOMETHING HERE THAT DESERVES AN EXPLANATION. HOW COME AUTHORS OF \cite{tilted_selftesting} SHOWED THAT FOR PARAMETERS $\theta_A = \theta_B = 0.999$, EVEN NPA LEVEL $d=10$ DOES NOT ACHIEVE THE ANALYTICAL OPTIMAL VIOLATION? WERE THEY USING A MORE PRECISE SOLVER AND DIDN'T THEY SEE THE NUMERICAL PROBLEM?}
\section{Ground-state observables in quantum many-body systems}
\label{sec:qmb}
A primary objective of quantum many-body physics is to understand the properties of the state of minimal energy, or ground state of a physical system. A conceptually simple method to do so is \textit{exact diagonalization}, where one represents the Hamiltonian describing a system as a matrix in a chosen basis and computes its full eigenspectrum numerically \cite{exact_diagonalization}. While this method provides numerically exact solutions, the Hilbert space dimension typically grows exponentially with the particle number $N$ (for instance, $\dim\mathcal{H} = 2^N$ for $N$ spin-$1/2$ particles), rendering exact diagonalization infeasible rapidly. To mitigate this, one often restricts the variational ground-state search to physically motivated manifolds \cite{variatonal_benchmarks}. Notably, tensor network methods~\cite{tensor_rmp}, e.g. Matrix-Product States or Projected Entangled Pair States, have been shown to significantly scale up the feasible problem sizes. However, a variational Ansatz, in general, does not provide any guarantee for how well the Ansatz space captures the actual ground state of the problem, and can just provide upper bounds to the true ground-state energy
\begin{equation}
    E_0^{var} \geq E_0.
\end{equation}
For any other observable, variational methods just offer an estimation, with no guarantee of how close it is to the true value or whether it bounds it from above or below. The NPA hierarchy constitutes a counterpart to these variational Ansätze by virtue of relaxation, hence allowing one to (numerically) certify an interval for $E_0$:
\begin{equation}
    E_0^{var} \geq E_0 \geq E_0^{\mathrm{SDP}}.
\end{equation}

Now, in \cite{groundstate} the authors make use of such inequality constraints on the ground-state energy to obtain bounds on ground-state observables beyond the energy (e.g. magnetization and local correlators). This constitutes a particularly interesting application for our method, as it constitutes one of the aforementioned iterative algorithms, reusing SDP bounds of former runs in consecutive programs. If wrong numerical lower bounds are issued, $E_0^{\mathrm{SDP}} \nleq E_0$, imposing such a wrong inequality in the computation of other ground state observables might cause an accumulation of errors. Our certification method precisely eliminates this failure mode by turning the SDP output into a rigorously valid lower bound.

We start by applying our certification scheme to the computation of ground state energies of the translationally invariant $J_1-J_2$ Heisenberg chain with nearest- and second-nearest \revise{neighbor} interactions and periodic boundary conditions with Hamiltonian
\begin{equation}
H = \frac14 
\sum_{i=1}^{N} \sum_{a \in \{x,y,z\}}
\left( \sigma_i^{a}\,\sigma_{i+1}^{a} + J_{2}\,\sigma_i^{a}\,\sigma_{i+2}^{a} \right),
\end{equation}
where $\sigma_i^a$ denotes the Pauli operator of type $a$ acting on site $i$. We adapt the convention $J_1 = 1$ from the underlying work.

In \cite{groundstate,wang2026scalable}, the authors exploit sparsity and various symmetries to significantly improve the scalability of the hierarchy, allowing them to study large system sizes. Exploiting the sign symmetry of the model lets one block-diagonalize the dense Gram matrix into four smaller blocks $G^{(j)},j\in \{1,2,3,4\}$. Due to the model's symmetry under permutation of the Pauli operators $\{\sigma_i^x,\sigma_i^y,\sigma_i^z\}$ , one can identify three of the blocks with each other, leaving one with two independent PSD blocks $G^{(j)},j\in \{1,2\}$ which can be further block diagonalized to blocks $G_i^{(j)},i\in\{1,\ldots,N/2+1\}$ due to the translational symmetry of the model. The associated transformation is the discrete Fourier transform  $P\in\mathbb{C}^{N\times N}$:
\begin{equation}
    P_{i,j}=\frac{1}{\sqrt{N}}\mathrm{e}^{-2\pi\mathbf{i}(i-1)(j-1)/N},\quad i,j=1,\ldots,N.
\end{equation} 
For details on these symmetry exploitations the reader is directed to \cite{wang2026scalable}. 

After exploiting the symmetries, the numerical SOHS certificate for lower bounds $\lambda_d$ to said ground state problems can be obtained as
\begin{align}
    &\frac{H}{N}-\lambda_d\simeq \left[1, \left(\oneop{1}_{s_1}\otimes P_1\right)\vec{b}_1\right]^*G^{(1)}_{1}\left[1, \left(\oneop{1}_{s_1}\otimes P_1\right)\vec{b}_1\right]\nonumber \\
    &+\sum_{i=2}^{\frac{N}{2}+1}\left(\left(\oneop{1}_{s_1}\otimes P_i\right)\vec{b}_1\right)^*G^{(1)}_{i}\left((\oneop{1}_{s_1}\otimes P_i\vec{b}_1\right)\\
    &+\sum_{i=1}^{\frac{N}{2}+1}\left(\left(\oneop{1}_{s_2}\otimes P_i\right)\vec{b}_2\right)^*G^{(2)}_{i}\left(\left(\oneop{1}_{s_2}\otimes P_i\right)\vec{b}_2\right), \nonumber
\end{align}
where $\vec{b}_j$ with $|\vec{b_j}| = Ns_j$ are monomial bases of the two sign symmetry sectors, $\oneop{1}_{s_j}$ denotes the identity matrix of size $s_j$, and $P_i$ the $i$-th row of $P$.
Note that one can interpret the discrete Fourier transformation as acting on the bases or on the Gram matrix blocks. Hence, one can either perform the projection on the polynomially indexed blocks $G_i^{(j)}$ according to Lemma \ref{thm:froboptsym}, or inflate them back to the monomial basis and perform the projection according to Lemma \ref{thm:sparsefrobopt}. As certification in the monomial basis becomes significantly more expensive for large $N$ (see Appendix \ref{sec:complexity}), all computations have been performed in the symmetry-adapted basis. In this scenario, the larger blocks $G^{(j)}$ are indexed by independent monomial bases $\vec{b_j}$, while their subblocks $G_i^{(j)}$ have monomial overlap in their indexing polynomials $(\oneop{1}_{s_j} \otimes P_i)\vec{b}_j$. This hence represents an instance featuring both sparsity and symmetry, where an intermediate lifting between Theorems \ref{thm:sparsebounds} and \ref{thm:symbounds} can be applied:
\begin{equation} \label{eq:qmblifting}
    \lambda_d^{\mathrm{rat}} = \tilde{\lambda}_d + \min_{i}\{\mu_{\min}^{i1}\}\cdot N s_1 + \min_{i} \{ \mu_{\min}^{i2} \}\cdot Ns_2,
\end{equation}
where $\{ \mu_{\min}^{ij} \}$ denotes the minimal eigenvalue of block $G_i^{(j)}$. We denote the arithmetic mean of the minima involved by $\overline{\mu_{\min}} = \sum_j\min_{i} \{\mu_{\min}^{ij}\}/2$.

Figure \ref{fig:j1j2_gs_energy} shows the results for applying our certification scheme to a $J_1-J_2$ Heisenberg chain up to length $N=50$. One can see from Figure \ref{fig:j1j2_ratbounds} that the certified bounds are consistent with the upper bounds computed by DMRG, recovering the results from \cite{groundstate}. Figure \ref{fig:j1j2_boundshifts} shows the bound shift by going from numerical to rational bounds depending on $N$. What is remarkable here is that in this scenario the bound shift \textit{decreases} with an increasing problem size. This can be explained by the negative minimal $\overline{\mu_{\mathrm{min}}}$ of the projected blocks $\mathcal{P}(\tilde{G}_i^{(j)})$ decreasing in magnitude with increasing $N$ (see Figure \ref{fig:j1j2_mineigs}), which indicates that the numerical certificates carry a smaller error. This is most likely due to the fact that at larger $N$ more matrix elements $(G_i^{(j)})_{\alpha,\beta}$ are assigned to the same equivalence class under constraints and symmetries, possibly leading to a more effective cancellation of errors. The second factor preventing a rapid increase in $\delta_d$ is the total size of all Gram matrices $s_{n,d}^{\mathcal{I}} = N(s_1 + s_2)$, growing linearly with $N$ (Figure \ref{fig:j1j2_gramsizes}) due to the translational invariance of the model. In conclusion, this first experiment suggests that the numerical results of the certification method seem to improve when symmetries are exploited. However, they also come with an increased computational cost, as it is shown in Appendix \ref{sec:complexity}.

\begin{figure}[hp] 
  \centering
  \begin{subfigure}{\columnwidth}
    \centering
    \includegraphics[width=\columnwidth]{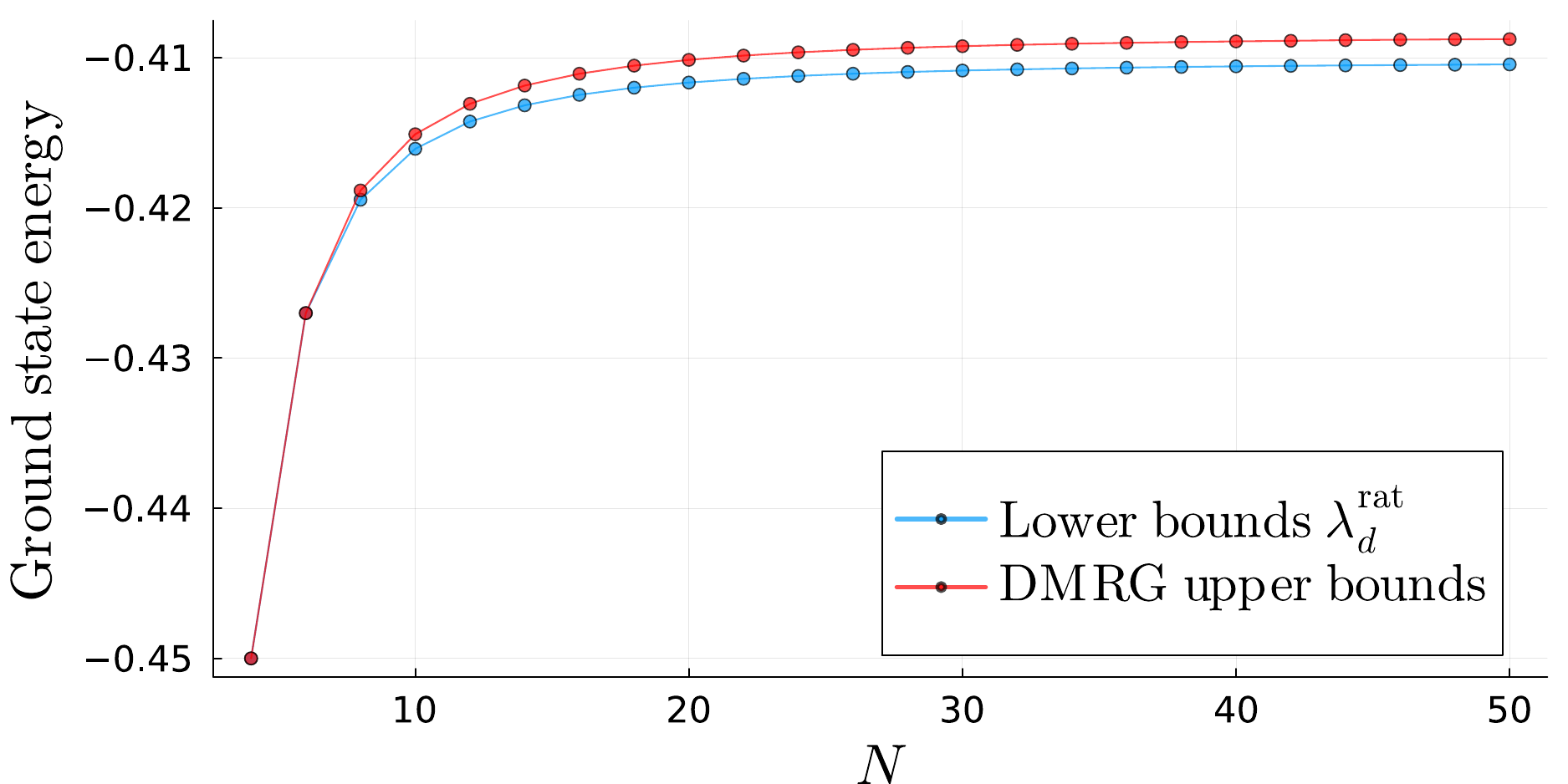}
    \caption{Certified lower bounds $\lambda_d^{\mathrm{rat}}$ and DMRG upper bounds to $E_0$}
    \label{fig:j1j2_ratbounds}
  \end{subfigure}

  \medskip

  \begin{subfigure}{\columnwidth}
    \centering
    \includegraphics[width=\columnwidth]{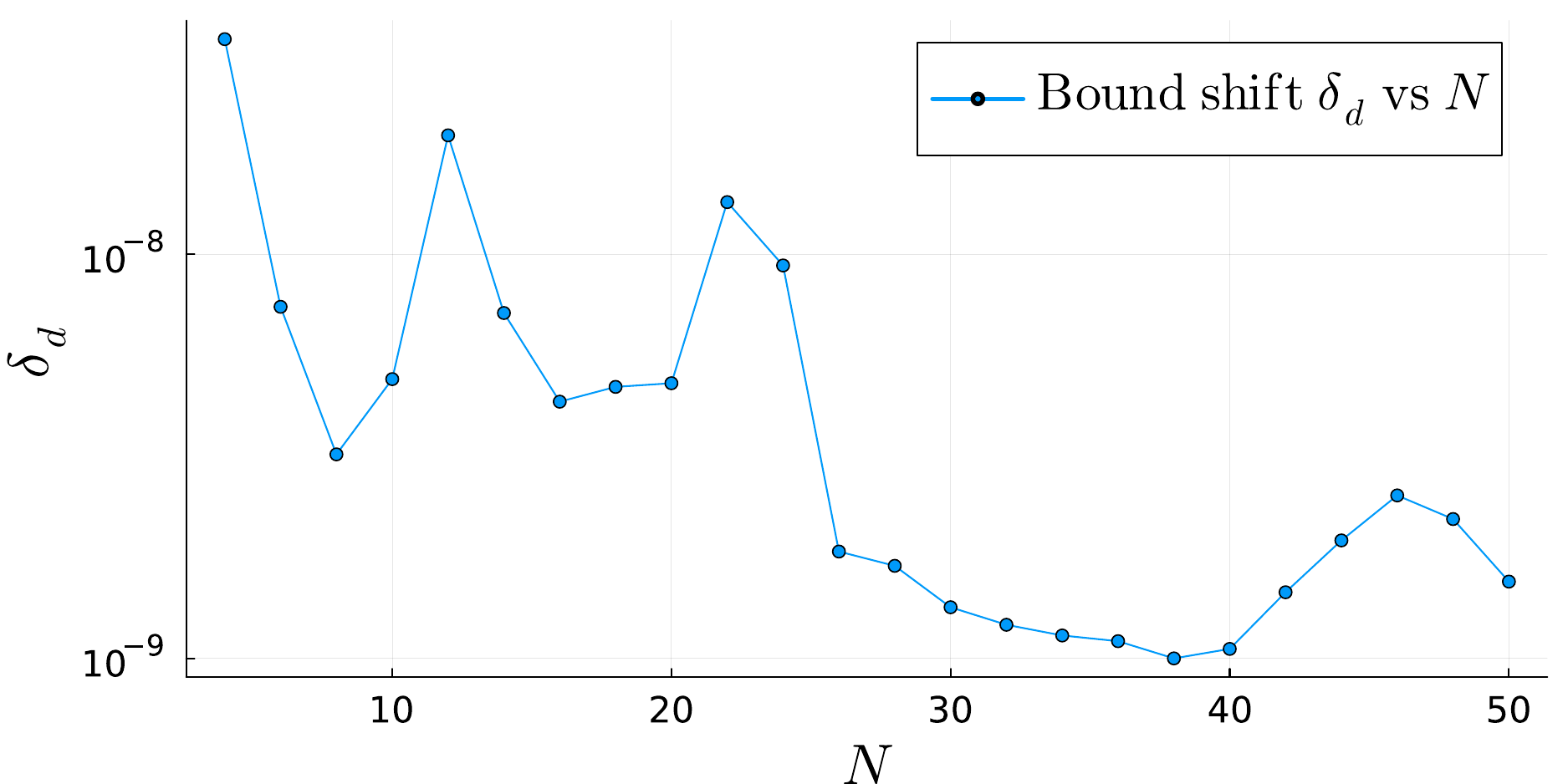}
    \caption{Rationalization bound shift $\delta_d = \lambda_d - \lambda_d^{\mathrm{rat}}$}
    \label{fig:j1j2_boundshifts}
  \end{subfigure}

  \medskip

  \begin{subfigure}{\columnwidth}
    \centering
    \includegraphics[width=\columnwidth]{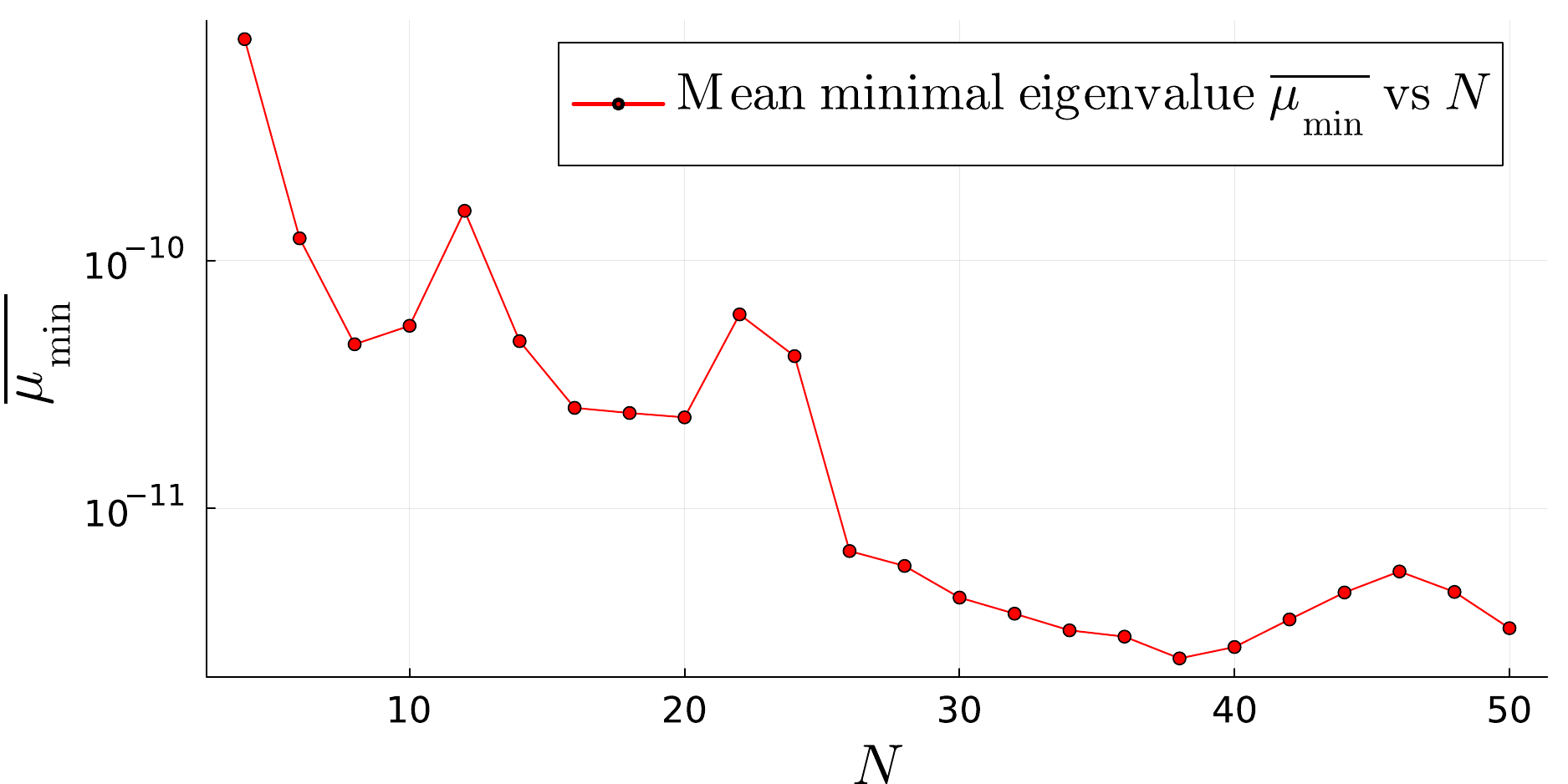}
    \caption{Mean minimal eigenvalue $\overline{\mu_{\mathrm{min}}}$ of the projected Gram blocks $\mathcal{P}(\tilde{G}_i^{(j)})$}
    \label{fig:j1j2_mineigs}
  \end{subfigure}

    \medskip

  \begin{subfigure}{\columnwidth}
    \centering
    \includegraphics[width=\columnwidth]{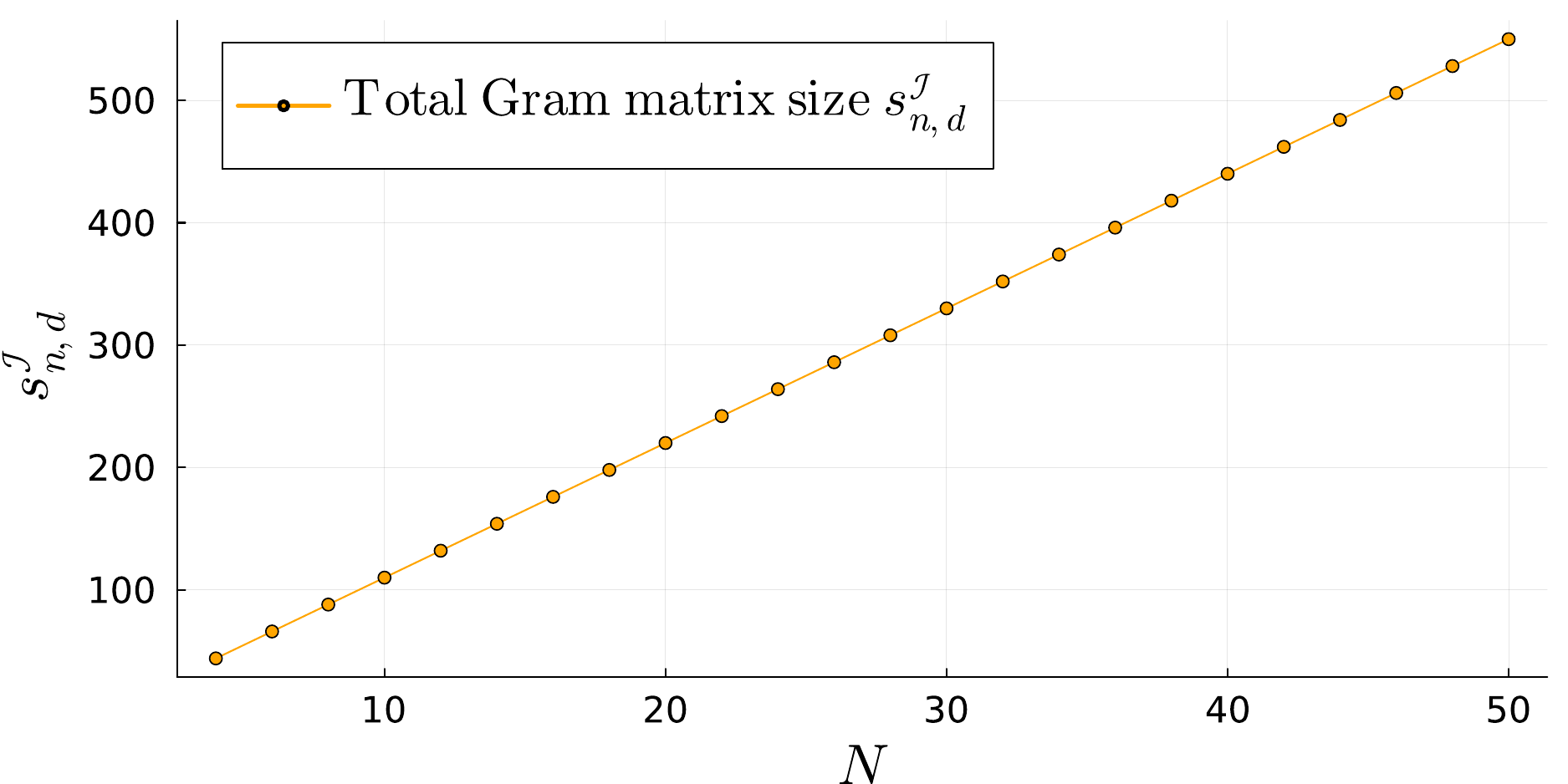}
    \caption{Total size of Gram matrices $\mathcal{P}(\tilde{G}_i^{(j)})$}
    \label{fig:j1j2_gramsizes}
  \end{subfigure}

  \caption{Numerical results for obtaining certified bounds to the $J_1-J_2$ Heisenberg chain ground state energy depending on the problem size $N$. Computations were performed at an intermediate relaxation order between $d=2$ and $d=3$. }
  \label{fig:j1j2_gs_energy}
\end{figure}

One can extend the rationalization scheme to observables beyond the ground state energy. The numerical SOHS certificate then involves scalar Gram multipliers $\kappa_1, \kappa_2>0$ that correspond to the inequalities in $E_0$. The corresponding numerical certificate for lower (or upper) bounding an observable $O$ under previously established energy constraints reads

\begin{align}
    \frac{\pm O}{N}\mp&\lambda_d\simeq \left[1, \left(\oneop{1}_{s_1}\otimes P_1\right)\vec{b}_1\right]^*G^{(1)}_{1}\left[1, \left(\oneop{1}_{s_1}\otimes P_1\right)\vec{b}_1\right] \nonumber \\
    &+\sum_{i=2}^{\frac{N}{2}+1}\left(\left(\oneop{1}_{s_1}\otimes P_i\right)\vec{b}_1\right)^*G^{(1)}_{i}\left((\oneop{1}_{s_1}\otimes P_i\vec{b}_1\right)\nonumber \\
    &+\sum_{i=1}^{\frac{N}{2}+1}\left(\left(\oneop{1}_{s_2}\otimes P_i\right)\vec{b}_2\right)^*G^{(2)}_{i}\left(\left(\oneop{1}_{s_2}\otimes P(i)\right)\vec{b}_2\right) \nonumber \\&+\kappa_1 (H - E_0^{\mathrm{SDP}}) +  \kappa_2 (E_0^{var} - H).
\end{align}

To obtain a rational certificate from the numerical one, one first has to establish rational, certified bounds on $E_0$ for the multiplier terms. For the SDP bound we use the previously established certified lower bounds on $H$. As establishing certified variational bounds is not the scope of this work, we use $E_0^{var,rat} = \widetilde{(E_0^{var} + \delta E_0)}$ as a rational upper bound, where $\delta E_0$ is chosen larger than the DMRG precision, and the rounding precision $\eta$ of lower magnitude than $\delta E_0$ to ensure a valid bound.
Next, one proceeds to construct a rational LHS featuring the rational inequality terms and associated rationalized multipliers as described in Section \ref{sec:rational_bounds}. As the aforementioned symmetry considerations also hold in this scenario, one can then proceed to perform the projection according to Lemma \ref{thm:froboptsym} and obtain certified bounds as in (\ref{eq:qmblifting}).

\begin{figure}[ht]
  \centering
  \includegraphics[width=\columnwidth]{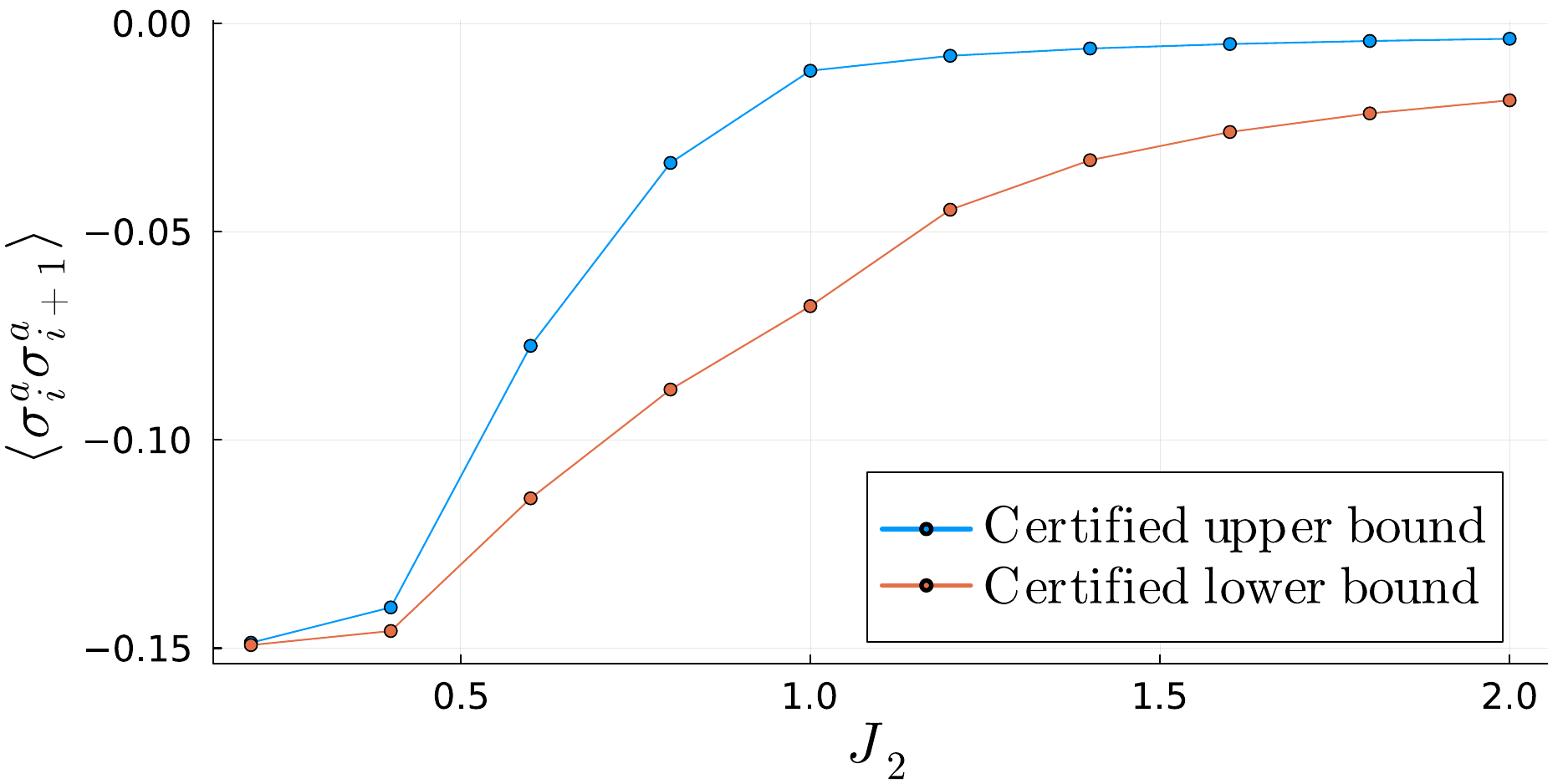}
  \caption{Certified bounds to  \revise{nearest-neighbor} correlators 
  $\langle \sigma^a_i \sigma^a_{i+1} \rangle$ in a $J_1-J_2$ Heisenberg chain 
  of length $N=12$. Computations were performed at an intermediate relaxation 
  order between $d=3$ and $d=4$.}
  \label{fig:antiferro_bounds}
\end{figure}

\begin{figure}[ht]
  \centering
  \includegraphics[width=\columnwidth]{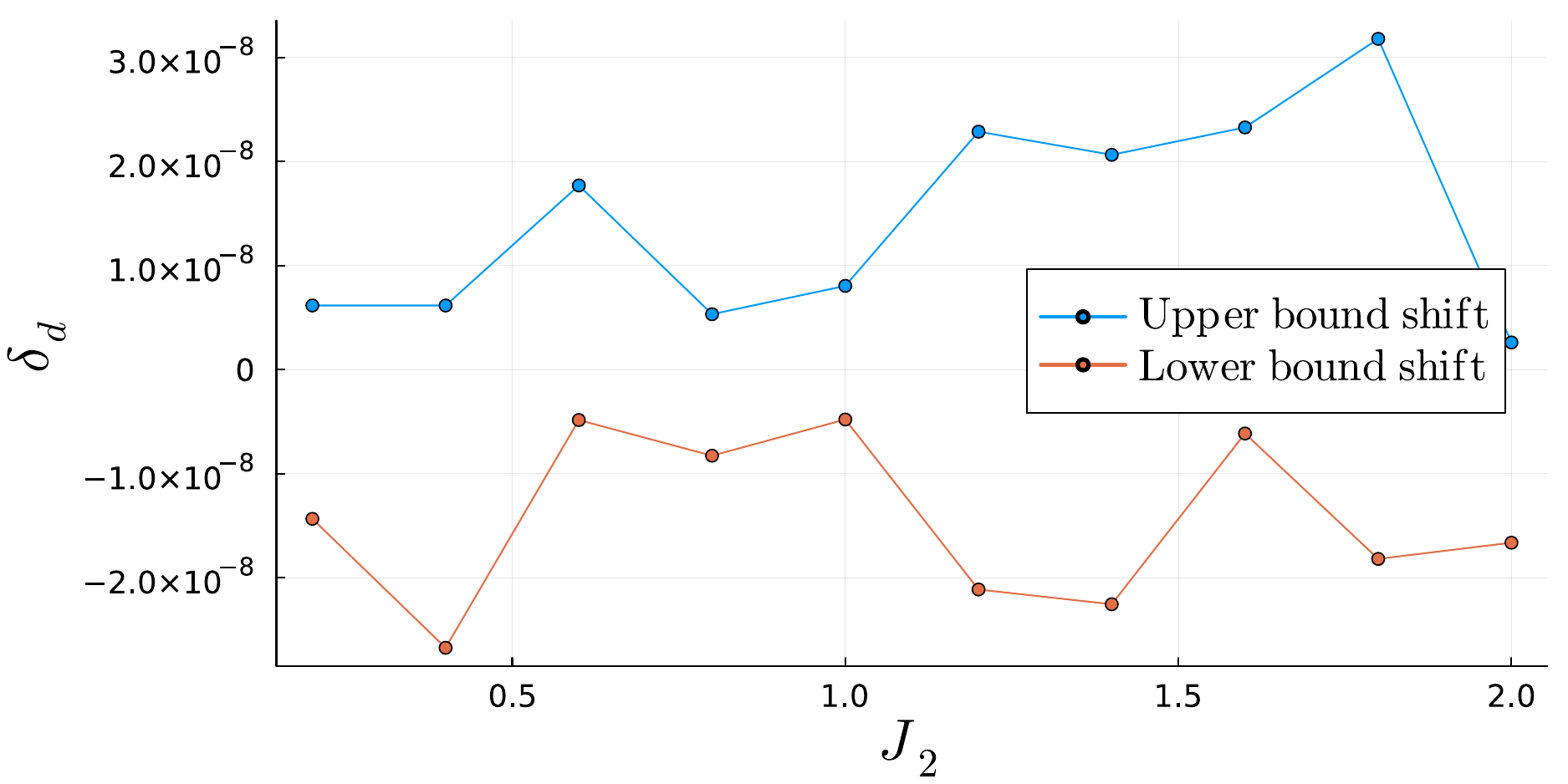}
  \caption{Bound shift $\delta_d$ for certified bounds on $\langle \sigma_i^a \sigma_{i+1}^a\rangle$ from Figure \ref{fig:antiferro_bounds}}
  \label{fig:antiferro_shifts}
\end{figure}

We start by investigating \revise{nearest-neighbor} correlations $\langle \sigma^a_i \sigma^a_{i+1} \rangle$ in the $J_1-J_2$ Heisenberg chain of length $N=12$ dependent on the second-nearest \revise{neighbor} couplings $J_2$. The results are shown in Figure \ref{fig:antiferro_bounds}, where we are able to certify antiferromagnetic correlations $\langle \sigma^a_i \sigma^a_{i+1} \rangle < 0$ for all $J_2$, recovering the results from \cite{groundstate}. Figure \ref{fig:antiferro_shifts} shows the corresponding bound shifts $\delta_d$, where we find no obvious correlations between the shifts $\delta_d$ and the coupling $J_2$. As the problem size $s_{n,d}^{\mathcal{I}}$ is fixed for all $J_2$, the differences in $\delta_d$ stem solely from differences in minimal eigenvalues $\mu^{ij}_{\min}$ of the projected Gram blocks. This indicates that the numerical problem is equally well conditioned for all $J_2$, in stark contrast to the following application.

\begin{figure}[ht]
  \centering
  \includegraphics[width=\columnwidth]{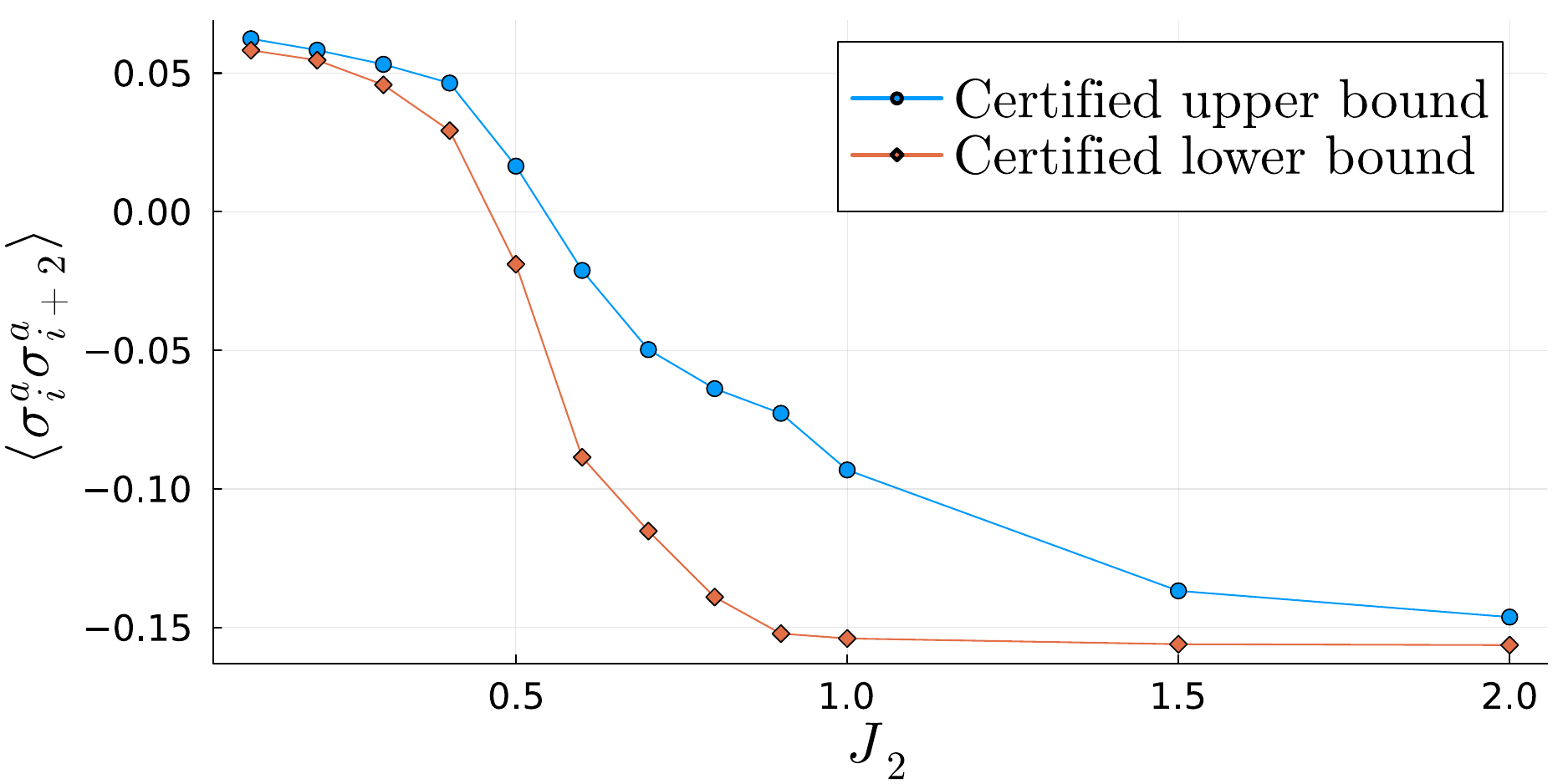}
  \caption{Certified bounds to second \revise{neighbor} spin correlations $\langle \sigma_i^a\sigma_{i+2}^a\rangle$ in the $J_1-J_2$ Heisenberg chain of length $N=12$ for different values of $J_2$. An intermediate relaxation order between $d=3$ and $d=4$ was used for computations.}
  \label{fig:phasetrans_bounds}
\end{figure}

\begin{figure}[ht]
  \centering
  \includegraphics[width=\columnwidth]{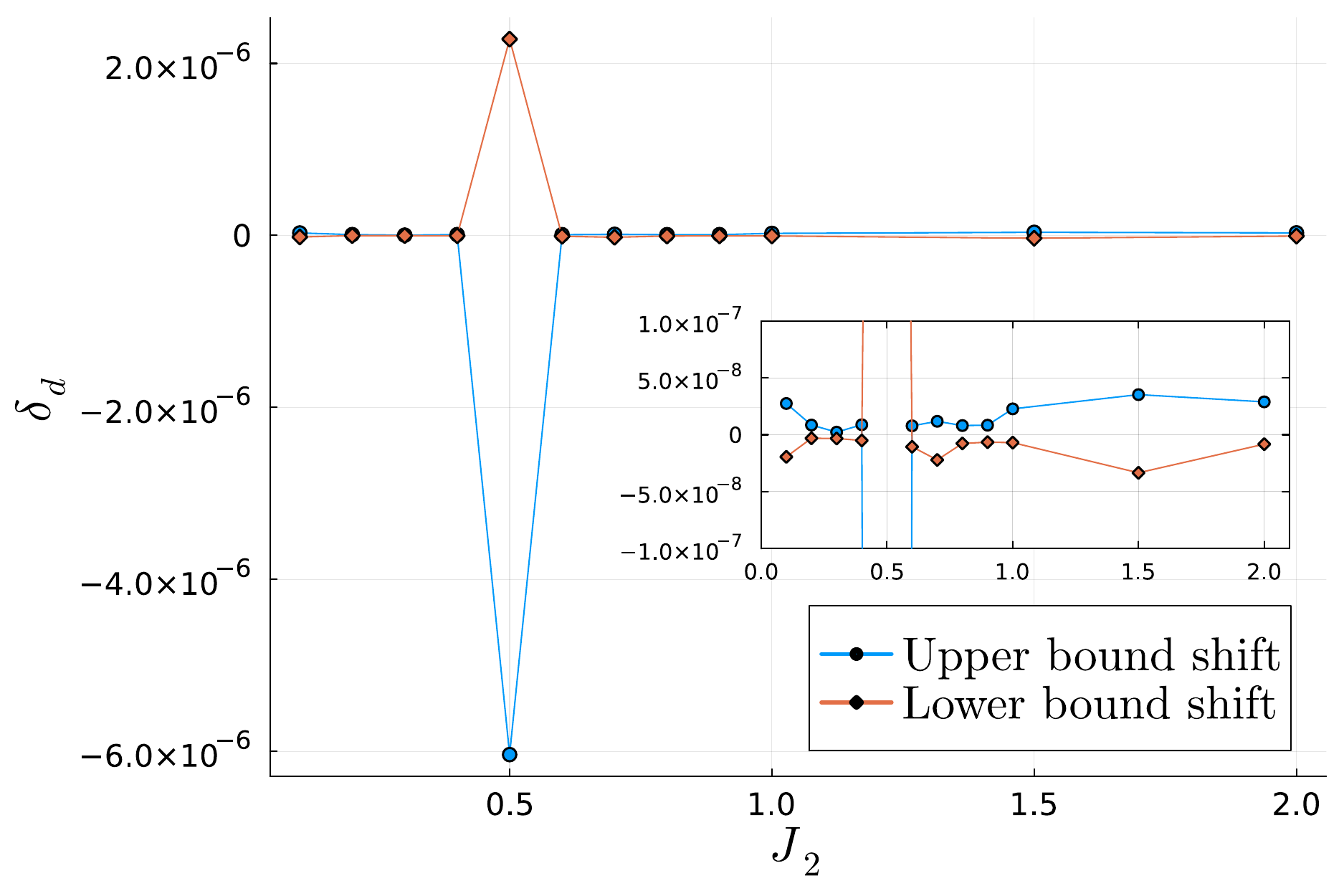}
  \caption{Bound shifts $\delta_d$ for certified bounds on $\langle \sigma_i^a\sigma_{i+2}^a\rangle$ from Figure \ref{fig:phasetrans_bounds}.}
  \label{fig:phasetrans_shifts}
\end{figure}

Figure \ref{fig:phasetrans_bounds} shows certified bounds to second \revise{neighbor} spin correlations $\langle \sigma_i^a\sigma_{i+2}^a\rangle$ in the $J_1-J_2$ Heisenberg chain of length $N=12$. We are able to certify ferromagnetic correlations $\langle \sigma_i^a\sigma_{i+2}^a\rangle>0$ for couplings $J_2\leq J^+_2<0.5$, and antiferromagnetic correlations $\langle\sigma_i^a\sigma_{i+2}^a\rangle<0$ for $J_2 \geq J^-_2>0.5$. Note that in the analysis of~\cite{groundstate}, the transition point was exactly at $J_2 =0.5$. At this value, the so called \textit{Majumdar-Ghosh point}, two exactly degenerate ground states $\ket{\psi_1}$ and $\ket{\psi_2}$ emerge \cite{degenerate_groundstate}, leading to a distinct \revise{behavior} in the observed bound shift, as displayed in Figure \ref{fig:phasetrans_shifts}: Remarkably, while for $J_2 \neq 0.5$ projected Gram matrices $\mathcal{P}(\tilde{G}_i^{(j)})$ with negative minimal eigenvalues $\mu^{ij}_{\min}<0$ emerge, for $J_2=0.5$ the numerical, as well as the projected Gram blocks, lie on the \textit{interior} of the PSD cone, carrying positive minimal eigenvalues $\mu^{ij}_{\min}>0$ of relatively large magnitude. This is in accordance with the observation of the SDP solver not being able to provide an optimal solution at $J_2=0.5$, but issuing a feasible point: As the objective function is linear and the feasible set is the intersection of an affine space $\mathcal{L}$ and the PSD cone, optimal numerical solutions have to lie on the boundary of the PSD cone and approximately fulfill $\mu^{ij}_{\min} \approx 0$. In the case of $\mu^{ij}_{\min}>0$, we are able to tighten the bounds to our problem by applying Theorem \ref{thm:boundtightening} (i) to shift the projected blocks from the interior of the PSD cone to its boundary.

This apparent ill-conditioning can be explained by the two-fold degeneracy of the ground state at $J_2=0.5$: As $\ket{\psi_1}$ and $\ket{\psi_2}$ share exactly the same ground state energy $E_0$, any convex combination will do so too, leading to an increase in dimensionality of the optimal subspace. Hence, the energy constraints employed to compute bounds on $\langle\sigma_i^a\sigma_{i+2}^a\rangle$ do \textit{not} restrict the search space in the same manner as for problems with non-degenerate ground states, as it is the case for the other investigated values of $J_2$ for a finite chain.

% Maybe to mention here: This phenomena of a stalling SDP solver gets less pronounced at larger N. In the corresponding references, it's shown that the region around J2=0.5 has a spectral gap that scales with 1/N, meaning that for larger N the degeneracy becomes more and more exact in the region, while for smaller N the degeneracy emerges "rapidly" when J2->0.5. In the thermodynamic limit N->inf the ground state is degenerate for all J2>J2_c=0.241. 

%\textbf{PLEASE DOUBLE CHECK THAT THE SYSTEM HAS DEGENERACY AT $J_2=0.5$. I SORT OF REMEMBER THAT PEOPLE TOLD ME THAT THE SYSTEM IS CRITICAL AT $J_2=1$ AND OFTEN CRITICALITY IS ASSOCIATED WITH ENERGY LEVEL CROSSING, THAT IS FIRST EXCITED AND GROUND-STATE ENERGIES BECOME EQUAL. I EVEN SORT OF REMEMBER THAT THE SYSTEM IS SOLVABLE AT $J_2=0.5$, WHERE THE GROUND STATE IS PRODUCT. I LIKE YOUR EXPLANATIONS, BUT PLEASE DOUBLE CHECK.}

This phenomenon suggests that, while in general our certification procedure produces \revise{weakened} bounds, it is possible to construct tighter rational bounds $\lambda_d^{\mathrm{rat}}$ with respect to their numerical counterpart $\lambda_d$ if SDP solutions are issued in the interior of the PSD cone. Conclusively, this suggests a second possible application of our method to obtain certified \textit{and} tighter bounds to problems which are intractable to SDP solvers.
\section{Conclusion and Outlook} \label{sec:conclusion}
In this work we presented a constructive method to obtain certified outer bounds to non-commutative optimization problems by postprocessing numerical data obtained from the NPA hierarchy. We do so by extracting \revisegreen{an exact} SOHS certificate from an SDP solver's Gram data, which consists of (1) rounding the data to rationals (2) performing a Frobenius optimal projection back onto the affine constraint subspace, and (3) a constraint dependent lifting of the projected Gram matrices into the PSD cone. We further extended this procedure to sparsity and symmetry-adapted scenarios, and \revise{analyzed} the method's complexity relative to the underlying SDP. We also provided a tightening procedure which for PD projected Gram matrices leads to tighter, certified bounds.

We then applied our method to obtain certified bounds on (a) the maximal violation of a \revise{catalog} of bipartite Bell inequalities, and (b) quantum many-body ground-state observables, where we found moderate losses in tightness that reflect a justified cost for obtaining rigorous guarantees. In the sparse case we observed that sparse rational bounds can become tighter than their dense counterparts, while in the symmetry-adapted scenario we identified a computational bottleneck that we suspect to be common to any SDP-based computer-assisted proof in arbitrary symmetry-adapted bases. In addition, we also find that in regimes of inexact SDP termination, the rational bounds can become tighter than their numerical, uncertified counterparts, which is of independent interest when dealing with numerical instabilities. In summary, we provided analytical constructions and numerical evidence that promote numerical semidefinite relaxations toward computer-assisted proofs in quantum theory.

Our method is readily extendable to modified non-commutative hierarchies such as trace- or state-optimization, as they differ from the presented hierarchy solely up to moment equality constraints. Following the same argument, it is also possible to extend this method to differentially constrained non-commutative polynomial optimization problems, allowing one to certify non-polynomial quantities in quantum theory. 
\revise{We also emphasize that if our post-processing framework proves applicable to non-commutative optimization problems arising in other application domains, this would further enhance its significance. }

There are two main venues for further investigation, namely (i) optimality of the bound-lifting step in the dense case, and (ii) scalability in the symmetry-adapted setting: Since our bound-lifting strategy scales with the size of the Gram blocks, which grows rapidly in the dense hierarchy, we observe increasing bound losses at higher relaxation orders. A refined bound-lifting procedure featuring a more involved spectral analysis of the Gram matrices, or a proof of optimality of the presented method hence deserve further investigation. Secondly, as we showed in Appendix \ref{sec:complexity}, when Gram blocks are indexed by polynomial bases, even establishing the numerical non-negativity certificate might exceed the cost of solving the associated SDP. Understanding how structure exploitation of the unitary basis change might reduce the certificate complexity is hence another open question. In addition, it would be interesting to systematically quantify the dependence of the certification loss on solver accuracy: Using higher-precision solvers should yield more accurate numerical certificates and hence smaller bound differences $\delta_d$, at the expense of computational efficiency. \revise{Regarding the last point, it would be also interesting to obtain \emph{lower} bounds on the relaxed maxima, providing the counterpart to our certified bounds and yielding certified enclosures for the initial numerical upper bound, allowing one to further judge its quality, e.g., as demonstrated in \cite{PhysRevLett.127.200504}. To do so, one has to apply the projection on the dual (pseudo-moment) side of the certificate, where the exact projection onto the pseudo-moment cone poses the main technical challenge.}

Beyond these directions, it would also be worthwhile to further study the certification scheme for numerically delicate instances in which the solver returns solutions strictly inside the PSD cone, and to systematically leverage such cases to obtain tighter bounds.
\section*{Acknowledgements} \label{sec:acknowledgements}
We thank Nando Leijenhorst, Jonas Britz, Erik Altelarrea Ferre, and Andreas Leitherer for useful discussions. This work has been supported by European Union’s HORIZON–MSCA-2023-DN-JD programme under the Horizon Europe (HORIZON) Marie Skłodowska-Curie Actions, grant agreement 101120296 (TENORS), the European Union Quantera project Veriqtas, the European Union Quantera project COMPUTE, the Government of Spain (Severo Ochoa CEX2019-000910-S, FUNQIP and NextGeneration EU PRTR-C17.I1) and European
 Union (QSNP, 101114043 and PASQuanS2.1, 101113690), Fundació Cellex, Fundació Mir-Puig, Generalitat de Catalunya (CERCA program), the ERC AdG CERQUTE, the AXA Chair in Quantum Information Science, the National Key R\&D Program of China under grant No.~2023YFA1009401, the Natural Science Foundation of China under grant No.~12571333, and the International Partnership Program of Chinese Academy of Sciences under grant No.~167GJHZ2023001FN. This work was granted access to the HPC resources of CALMIP supercomputing center under the allocation 2016-23035.

 \section*{Author contribution}

 The authors acknowledge the use of Large Language Models for mathematical writing and code development.

\appendix
\section{Proof of Lemma \ref{thm:boundlowering}} \label{sec:appendix}
\begin{proof}[Proof of Lemma \ref{thm:boundlowering}] We deal with the four different constraint families separately. \\
\textbf{Box and unipotency constraints.}
Assume $\{c_i\}_{i=1,..,n} =\{1-X_i^2\}_{i=1,..,n} \subseteq G$, e.g. Box constraints. We show that for $\varepsilon >0$ there exists the following rational decomposition
\begin{equation}
\label{eq:a1}
\varepsilon s_{n,d} = \varepsilon v_d^*v_d + \sum_{i,k} q_{ki}^\ast(1-X_i^2)q_{ki}.
\end{equation}
with polynomials $q_{ki}\in\mathbb{Q}_{d-1}\langle\underline{X}\rangle$. To do so, we first show that
\begin{equation}
\label{eq:a2}
1 = w^* w \;+\; \sum_{i,k} q_{ki}^*(1-X_i^2)q_{ki}.
\end{equation}
for all monomials $w$ with $|w|\leq d$. Proceed by induction on $d$. The variables $X_j$ yield the base case for words of length $|w| =1$:
\begin{equation}
\label{eq:a3}
1 = X_j^2 + (1-X_j^2).
\end{equation}
Assume the induction hypothesis \eqref{eq:a2} holds for $w$ with $|w|\leq d-1$. Multiplying by $X_j$ from  both sides gives
\begin{equation}
\label{eq:a4}
X_j^2 = u^* u + \sum_{i,k} X_j q_{ki}^*(1-X_i^2)q_{ki} X_j.
\end{equation}
with $u = X_jw$. Plugging this expression into \eqref{eq:a3} yields
\begin{equation*}
1 = u^* u + \sum_{i,k} X_j q_{ki}^\ast(1-X_i^2)q_{ki} X_j + (1-X_j^2)
\end{equation*}
which is of the desired form. Summing over all $u$ with $|u|\le d$ and multiplying by $\varepsilon$ proves the claim \eqref{eq:a1}. The full argument also holds for unipotency constraints.

\medskip
\textbf{Ball constraints.}
Assume $c_1 = (1-\sum_i X_i^2)\in G$. We want to show the existence of
\begin{equation}
\label{eq:a5}
\varepsilon s_{n,d} \revise{\oneop{1}} = \varepsilon v_d^*v_d + \sum_k r_k^\ast r_k + \sum_k q_k^\ast \Bigl(1-\sum_i X_i^2\Bigr) q_k.
\end{equation}
for $\varepsilon > 0$,  $r_k\in\mathbb{Q}_d\langle\underline{X}\rangle$ and $q_k\in\mathbb{Q}_{d-1}\langle\underline{X}\rangle$. We again proceed by induction on $|w|$, where we want to show that for all $w$ with $|w| \leq d$
\begin{equation}
\label{eq:a7}
1 = w^\ast w \;+\; \sum_k r_k^\ast r_k \;+\; \sum_k q_k^\ast \Bigl(1-\sum_i X_i^2\Bigr) q_k.
\end{equation}
 The base case is again given by the variables $X_j$:
\begin{equation}
\label{eq:a6}
1 = X_j^2 + \sum_{i\neq j} X_i^2 + \Bigl(1-\sum_i X_i^2\Bigr).
\end{equation}
Assuming \eqref{eq:a7} holds for $|w|\le d-1$, multiplying the expression from both sides by $X_j$ gives
\begin{equation}
\small
X_j^2 = u^\ast u \;+\; \sum_k X_j r_k^\ast r_k X_j + \sum_k X_j q_k^\ast \Bigl(1-\sum_i X_i^2\Bigr) q_k X_j.
\end{equation}
Substitute into \eqref{eq:a6}
\begin{align}
&1 = u^\ast u + \sum_{i\neq j} X_i^2 + \sum_k X_j r_k^\ast r_k X_j \\ &+ \sum_k X_j q_k^\ast \Bigl(1-\sum_i X_i^2\Bigr) q_k X_j 
   + \Bigl(1-\sum_i X_i^2\Bigr), \nonumber
\end{align}
which again proves the claim \eqref{eq:a5} by summing over all $|u|\le d$ and scaling by $\varepsilon$.

\medskip
\textbf{Projection constraints.}
Assume $ \{c_i\}_{i=1}^n =\{X_i - X_i^2\}_{i=1}^n\subseteq H$. We want to show
\begin{equation}
\label{eq:a8}
\varepsilon s_{n,d} \revise{\oneop{1}} = \varepsilon v_d^*v_d + \sum_k r_k^\ast r_k + \sum_{i,k} q_{ki}^\ast (X_i - X_i^2) q_{ki}.
\end{equation}
Which, as in the other cases, is ensured if the following holds for all $w$ with $|w| \leq d$:
\begin{equation}
\label{eq:a9}
1 = w^\ast w + \sum_k r_k^\ast r_k + \sum_{i,k} q_{ki}^\ast (X_i - X_i^2) q_{ki},
\end{equation}
which again is proven by induction. The following holds for all $j$
\begin{equation}
\label{eq:a10}
1 = X_j^2 + (X_j-1)^2 + 2(X_j - X_j^2).
\end{equation}
Assuming \eqref{eq:a9} holds for $|w|=d-1$, multiplying it by $X_j^2$ gives
\begin{equation} \small
X_j^2 = u^\ast u \;+\; \sum_k X_j r_k^\ast r_k X_j + \sum_{i,k} X_j q_{ki}^\ast (X_i - X_i^2) q_{ki} X_j.
\end{equation}
Plugging into \eqref{eq:a10} yields
\begin{align}
1 &= u^\ast u + \sum_k X_j r_k^\ast r_k X_j + \sum_{i,k} X_j q_{ki}^\ast (X_i - X_i^2) q_{ki} X_j \nonumber \\
  &\quad + (X_j-1)^2 \;+\; 2(X_j - X_j^2)
\end{align}
which is of the target SOHS form, yielding \eqref{eq:a8} by summing over all words and multiplying by $\varepsilon$. We note that the same argument also holds for $\{c_i\}_{i=1}^n\subseteq G$. 

\medskip
In all three cases, we obtain the desired SOHS decomposition for $f s_{n,d}$. The summing argument holds when working over a quotient algebra by replacing $s_{n,d}\rightarrow s_{n,d}^{\mathcal{I}}$, as well as when restricted to specific variable subsets in the sparse case, as in Theorem \ref{thm:sparsebounds}. 
\end{proof}
\section{Proof of Theorem \ref{thm:boundtightening}}\label{sec:tighteningproof}
\begin{proof}[Proof of Theorem \ref{thm:boundtightening}]
We start by proving the tightening scheme (i) from the Theorem. 

Let $\left(\tilde{\lambda}_d,\mathcal{P}(\tilde{G}_0)\right)$ be an exact pre-certificate obtained after rounding and projection
\begin{equation}
\label{eq:exactprecertap}
    \tilde{\lambda}_d \revise{\oneop{1}} - f \;=\; v_d^* \mathcal{P}(\tilde{G}_0)\, v_d \qquad \mathrm{mod}\,\mathcal I,
\end{equation}
and let $\mathcal{P}(\tilde{G}_0)\succ 0$ admit a positive lower bound $\mu_{\min}>0$ to its eigenspectrum. Recall the following constant SOHS decomposition from Equation \ref{eq:a1} in Appendix \ref{sec:appendix}:
\begin{equation}
    s_{n,d} \revise{\oneop{1}} = v_d^*v_d + \sum_{i,k}q_{ki}^*(1-X_i^2)q_{ki}
\end{equation}
Multiplying by any $\varepsilon\in \mathbb{R}$ and working modulo the underlying ideal then yields the following relation
\begin{equation} \label{eq:unipotentdec}
    \varepsilon s_{n,d}^{\mathcal{I}} \revise{\oneop{1}} = \varepsilon v_d^*v_d  \mod\mathcal{I}
\end{equation}
where this holds particularly for $\varepsilon <0$ and hence for $\varepsilon = -\mu_{\min} < 0$. This relation follows from the vector space property of the ideal. Adding Equations \ref{eq:exactprecertap} and \ref{eq:unipotentdec} then yields the following certificate in the quotient space
\begin{align}
    &\left(\tilde{\lambda}_d \revise{\oneop{1}} - \mu_{\min}s_{n,d}^{\mathcal{I}}\right) - f \;= \\ &v_d^*\left( \mathcal{P}(\tilde{G}_0)\, -\mu_{\min}\oneop{1}\right)v_d \quad \mathrm{mod}\,\mathcal I, \nonumber
\end{align}
which certifies $\lambda_d^{\mathrm{rat}} = \left(\tilde{\lambda}_d - \mu_{\min}s_{n,d}^{\mathcal{I}}\right) < \tilde{\lambda}_d$ as an upper bound to the optimization problem.

We will now prove the second tightening scheme (ii) of the Theorem. Again, assume an exact precertificate as in Equation \ref{eq:exactprecertap}, where no further assumptions on the constraints of the problem are taken. Our goal is to find the maximum $\tau>0$ such that the following rank-1 downdate is still PSD:
\begin{equation}
\label{eq:rkdowndate2}
    \mathcal{P}(\tilde{G}_0)\, -\tau e_0 e_0^* \succeq 0
\end{equation}
considering such downgraded matrix then leads to the following SOHS decomposition
\begin{equation} \small \label{eq:rkdowndate1}
    \left(\tilde{\lambda}_d - \tau\right) \revise{\oneop{1}} - f = v_d^*\left( \mathcal{P}(\tilde{G}_0)\, -\tau e_0 e_0^* \right) v_d \quad \mathrm{mod}\,\mathcal I, .
\end{equation}
Start by considering the following block hermitian matrix
\begin{equation}
    S := \begin{pmatrix} \mathcal{P}(\tilde{G}_0) & e_0 \\ e_0^* & \tau^{-1}\end{pmatrix}
\end{equation}
Now, the Schur complement of $\tau^{-1}$ lets one characterize the PSD condition \ref{eq:rkdowndate2}:
\begin{equation}
    \left( \mathcal{P}(\tilde{G}_0) - \tau e_0 e_0^* \right) \succeq 0 \Leftrightarrow S \succeq 0
\end{equation}
On the other hand, the Schur complement of $\mathcal{P}(\tilde{G}_0)$ yields
\begin{equation}
    \tau^{-1} - e_0^*\mathcal{P}(\tilde{G}_0)^{-1}e_0 \geq 0 \Leftrightarrow S\succeq 0.
\end{equation}
Putting both arguments together yields
\begin{equation}
    \left( \mathcal{P}(\tilde{G}_0) - \tau e_0 e_0^* \right) \succeq 0 \Leftrightarrow \tau \leq \frac{1}{e_0^*\mathcal{P}(\tilde{G}_0)^{-1}e_0}
\end{equation}
Hence, for $\tau = \frac{1}{e_0^*\mathcal{P}(\tilde{G}_0)^{-1}e_0}$, Equation \ref{eq:rkdowndate1} is a valid SOHS certificate, certifying $\lambda_d^{\mathrm{rat}} = \tilde{\lambda}_d - \tau$ as a valid upper bound to the optimization problem. One finds $\tau \in [\mu_{\min},\mu_{\max}]$ due to the relation of the eigenspectra of $\mathcal{P}(\tilde{G}_0)$ and its inverse. We also note that as $\mathcal{P}(\tilde{G}_0)\in\mathbb{Q}^{s_{n,d}^{\mathcal{I}}\times s_{n,d}^{\mathcal{I}}}$ we have $\tau \in \mathbb{Q}$, which is why no further rationalization is needed if the exact inverse of $\mathcal{P}(\tilde{G}_0)$ is computed.
\end{proof}
\section{Complexity considerations}
\label{sec:complexity}
In this section we address the computational complexity of the projection schemes in Lemmas \ref{thm:sparsefrobopt} and \ref{thm:froboptsym} (dense, sparse, and symmetry-adapted), as well as the complexity to compute the associated bound shift $\delta_d$.

First we show that the projection scheme \eqref{thm:sparsefrobopt}, dense or sparse, over a binomial ideal is computationally less complex than solving the underlying SDP. We give the argument in the sparse case as it trivially includes the dense case. We adopt the standard interior-point complexity model, in which 
an SDP with $\revise{S}$ blocks of maximal size $N$ requires $\mathcal{O}(\revise{S}N^3)$ arithmetic operations per iteration, and $\mathcal{O}(\sqrt{\revise{S}N})$ iterations \cite{sdp_complexity}. As the scaling with respect to the number of constraints involved is not made explicit, this constitutes an agnostic lower bound on the overall complexity.

%\textbf{PLEASE CHANGE THE VERB AND IMPROVE THE EXPLANANTION, AS AGREEED DURING THE CALL} 

According to Theorem \ref{thm:sparsebounds}, the Frobenius optimal projection over a quotient space generated by a binomial ideal is given by
\begin{equation}
    \Delta^{(k)}_{\alpha,\beta} = \frac{r_{\mathcal{N}(\alpha^*\beta)}}{n_\mathcal{I}(\alpha,\beta,k)}
\end{equation}
Now, for $\revise{S}$ cliques and associated Gram blocks with largest blocksize $N$, this update has asymptotic complexity $\mathcal{O}(\revise{S}N^2)$. To compute the maps $n_{\mathcal{I}}(\alpha,\beta,k)$, one has to iterate over all blocks $\revise{S}$ of maximal size $N^2$, also yielding complexity $\mathcal{O}(\revise{S}N^2)$.
Now, once the maps $n_{\mathcal{I}}(\alpha,\beta,k)$ are computed, the residuals $r_{\mathcal{N}(\alpha^*\beta)}$ are also obtained in $\mathcal{O}(\revise{S}N^2)$ according to Equation \eqref{eq:sparseresiduals}, leading to a total complexity of $\mathcal{O}(\revise{S}N^2)$, which is dominated by the SDP cost per iteration of $\mathcal{O}(\revise{S}N^3)$.

We are now going to show that in the case of Gram matrices indexed by non-monomial bases, as it is the case when a Gram matrix is block diagonalized due to symmetry, the computational complexity of the rationalization scheme can become dominant depending on the problem parameters.

Assume a given Gram matrix can be block diagonalized due to symmetries of the problem, as it is the case in \cite{groundstate}. Assume that after block-diagonalization we are again left with $\revise{S}$ blocks of maximal size $N$. Each Gram block $k$ is indexed by a polynomial basis
\begin{equation}
    b_i^{(k)} = \sum_{\alpha\in \revise{S}N} b_{\alpha,i}^{(k)}v_\alpha
\end{equation}
where $(v_{\alpha})_{\alpha=1}^{\revise{S}N}$ indexes the Gram matrix $G_0$, e.g., constitutes the monomial indexing basis before block diagonalization. The naive projection would be to transform the large block diagonal matrix back to the monomial basis, losing the block diagonal structure and ending up with complexity $\mathcal{O}(\revise{S}^2N^2)$ according to the previous dense argument. In that case, the algorithm is quadratic in $\revise{S}$ while the SDP solving cost is linear in $\revise{S}$. Thus, in the case of many small blocks, performing the projection in the monomial basis might become more costly than solving the SDP.

Secondly, one can perform the projection directly in the polynomial basis as in Lemma \ref{thm:froboptsym}.
Defining all variables accordingly, the numerical certificate reads
\begin{align}
    \mathcal{N}(f - \lambda \revise{\oneop{1}} )_t \approx \sum_{i,j,k} \tilde{n}_{ijk}^t G_{ij}^{(k)}
\end{align}
Now, in case of the underlying ideal $\mathcal{I}$ being binomial, the map $n_{ijk}^t$ reads
\begin{equation}
    \tilde{n}_{ijk}^t = \sum_{\alpha,\beta}^{\revise{S}N} (b_{i,\alpha}^{(k)})^*b_{j,\beta}^{(k)} \delta_{\mathcal{N}(\alpha^*\beta),t}
\end{equation}
In the most general case, each indexing polynomial $b_i^{(k)}$ can have the whole monomial basis of size $\revise{S}N$ as its support. This leads to a complexity of building $\tilde{n}_{ijk}^t$ in a polynomial basis of $\mathcal{O}(\revise{S}^2N^2)$, which for large $\revise{S}$ dominates the linear cost in $\revise{S}$ per SDP iteration.

The complexity  argument above also holds for the establishment of numerical SOHS certificates involving symmetries from SDP data and is hence independent of our rationalization scheme. This implies that any SOHS-based proof for non-negativity in symmetry-adapted bases might become computationally more complex than solving the underlying SDP, which constitutes a general bottleneck for SDP assisted proofs for non-negativity in symmetry-adapted scenarios.

Note that besides the asymptotic argument, there is an associated cost per arithmetic operation performed in rational numbers. Depending on the rounding accuracy, the rationalized Gram matrix entries will have numerators and denominators of large bit size. Depending on the size of the rational numbers, rational arithmetic can become orders of magnitude more expensive than the associated operations using floating point numbers. This introduces a degree of freedom, where one can trade rounding accuracy and hence quality of the result for computational complexity.

In addition to performing the projection, to obtain a certified bound $\lambda_d^{\mathrm{rat}}$, one needs to compute the minimal eigenvalues of the projected Gram matrices, where in the sparse case this results in complexity $\mathcal{O}(\revise{S}N^3)$, which is dominated by the overall SDP complexity. 
In the symmetry-adapted scenario, while the projection complexity is independent of the choice of basis, the eigenvalue computation  complexity does heavily depend on the choice: Performing the eigenvalue computation in the inflated monomial basis results in a complexity of $\mathcal{O}(\revise{S}^3N^3)$ while computing the eigenvalues in the block-diagonal basis has complexity $\mathcal{O}(\revise{S}N^3)$. Hence, for large $\revise{S}$ the eigenvalue computation in the monomial basis might become more costly than solving the underlying SDP, while the computation in the symmetry-adapted basis is always cheaper. Hence, from a computational point of view, it is preferable to perform the certification scheme in the symmetry-adapted basis than translating back to the monomial basis. 

These asymptotic considerations are consistent with our numerical observations: In the Bell example, projection and eigenvalue costs remain negligible relative to the SDP solving, while in the many-body scenario the combination of high rounding precision and a large number of small blocks lead to the projection eventually becoming the bottleneck of the computation.

%\nocite{*}
\bibliographystyle{quantum}
\bibliography{ref}

\begin{thebibliography}{99}

\bibitem{bounding_quantum}
Miguel Navascu\'es, Stefano Pironio, and Antonio Ac\'{\i}n.
\newblock ``Bounding the set of quantum correlations''.
\newblock \href{https://dx.doi.org/10.1103/PhysRevLett.98.010401}{Phys. Rev. Lett. {\bf 98}, 010401}~(2007).

\bibitem{groundstate}
Jie Wang, Jacopo Surace, Ir{\'e}n{\'e}e Fr{\'e}rot, Beno{\^\i}t Legat, Marc-Olivier Renou, Victor Magron, and Antonio Ac{\'\i}n.
\newblock ``Certifying ground-state properties of many-body systems''.
\newblock \href{https://dx.doi.org/10.1103/physrevx.14.031006}{Physical Review X {\bf 14}, 031006}~(2024).

\bibitem{araujo_causality}
Mateus Ara\'ujo, Cyril Branciard, Fabio Costa, Adrien Feix, Christina Giarmatzi, and {\v C}aslav Brukner.
\newblock ``Witnessing causal nonseparability''.
\newblock \href{https://dx.doi.org/10.1088/1367-2630/17/10/102001}{New Journal of Physics {\bf 17}, 102001}~(2015).

\bibitem{contextuality}
Armin Tavakoli, Emmanuel~Zambrini Cruzeiro, Roope Uola, and Alastair~A. Abbott.
\newblock ``Bounding and {Simulating} {Contextual} {Correlations} in {Quantum} {Theory}''.
\newblock \href{https://dx.doi.org/10.1103/PRXQuantum.2.020334}{PRX Quantum {\bf 2}, 020334}~(2021).

\bibitem{navascues_convergent}
Miguel Navascu{\'{e}}s, Stefano Pironio, and Antonio Ac\'{\i}n.
\newblock ``A convergent hierarchy of semidefinite programs characterizing the set of quantum correlations''.
\newblock \href{https://dx.doi.org/10.1088/1367-2630/10/7/073013}{New Journal of Physics {\bf 10}, 073013}~(2008).

\bibitem{npa}
Stefano Pironio, Miguel Navascu{\'e}s, and Antonio Ac\'{\i}n.
\newblock ``Convergent relaxations of polynomial optimization problems with noncommuting variables''.
\newblock \href{https://dx.doi.org/10.1137/090760155}{SIAM Journal on Optimization {\bf 20}, 2157--2180}~(2010).

\bibitem{klep_sparse}
Igor Klep, Victor Magron, and Janez Povh.
\newblock ``Sparse noncommutative polynomial optimization''.
\newblock \href{https://dx.doi.org/10.1007/s10107-020-01610-1}{Mathematical Programming {\bf 193}, 789--829}~(2022).

\bibitem{wang_sparse}
Jie Wang and Victor Magron.
\newblock ``Exploiting term sparsity in noncommutative polynomial optimization''.
\newblock \href{https://dx.doi.org/10.1007/s10589-021-00301-7}{Computational Optimization and Applications {\bf 80}, 483--521}~(2021).

\bibitem{riener_symmetry}
Cordian Riener, Thorsten Theobald, Lina~Jansson Andr{\'{e}}n, and Jean~B. Lasserre.
\newblock ``Exploiting symmetries in {SDP}-relaxations for polynomial optimization''.
\newblock \href{https://dx.doi.org/10.1287/moor.1120.0558}{Mathematics of Operations Research {\bf 38}, 122--141}~(2013).

\bibitem{fawzi_certified}
Hamza Fawzi, Omar Fawzi, and Samuel~O. Scalet.
\newblock ``Certified algorithms for equilibrium states of local quantum {Hamiltonians}''.
\newblock \href{https://dx.doi.org/10.1038/s41467-024-51592-3}{Nature Communications {\bf 15}, 7394}~(2024).

\bibitem{optimality}
Mateus Ara{\'u}jo, Igor Klep, Andrew J.~P. Garner, Tam{\'a}s V{\'e}rtesi, and Miguel Navascu{\'e}s.
\newblock ``First-order optimality conditions for non-commutative optimization problems''.
\newblock \href{https://dx.doi.org/10.1007/s10208-026-09761-x}{Foundations of Computational Mathematics}~(2026).
\newblock  \href{http://arxiv.org/abs/2311.18707}{arXiv:2311.18707}.

\bibitem{burgdorf_tracial}
Sabine Burgdorf, Kristijan Cafuta, Igor Klep, and Janez Povh.
\newblock ``The tracial moment problem and trace-optimization of polynomials''.
\newblock \href{https://dx.doi.org/10.1007/s10107-011-0505-8}{Mathematical Programming {\bf 137}, 557--578}~(2013).

\bibitem{klep_state_2023}
Igor Klep, Victor Magron, Jurij Vol{\v{c}}i{\v{c}}, and Jie Wang.
\newblock ``State polynomials: positivity, optimization and nonlinear {Bell} inequalities''.
\newblock \href{https://dx.doi.org/10.1007/s10107-023-02024-5}{Mathematical Programming {\bf 207}, 645--691}~(2024).

\bibitem{araujo_differential}
Mateus Ara\'ujo, Andrew J.~P. Garner, and Miguel Navascues.
\newblock ``Non-commutative optimization problems with differential constraints''.
\newblock \href{https://dx.doi.org/10.22331/q-2026-07-22-2169}{Quantum {\bf 10}, 2169}~(2026).
\newblock  \href{http://arxiv.org/abs/2408.02572}{arXiv:2408.02572}.

\bibitem{branch_and_bound}
Flavio Baccari, Christian Gogolin, Peter Wittek, and Antonio Ac\'{\i}n.
\newblock ``Verifying the output of quantum optimizers with ground-state energy lower bounds''.
\newblock \href{https://dx.doi.org/10.1103/PhysRevResearch.2.043163}{Phys. Rev. Res. {\bf 2}, 043163}~(2020).

\bibitem{dynamic_ineq}
Bissan Ghaddar, Juan~C. Vera, and Miguel~F. Anjos.
\newblock ``A dynamic inequality generation scheme for polynomial programming''.
\newblock \href{https://dx.doi.org/10.1007/s10107-015-0870-9}{Mathematical Programming {\bf 156}, 21--57}~(2016).

\bibitem{laurent2008sums}
Monique Laurent.
\newblock ``Sums of squares, moment matrices and optimization over polynomials''.
\newblock In Mihai Putinar and Seth Sullivant, editors, Emerging applications of algebraic geometry.
\newblock \href{https://dx.doi.org/10.1007/978-0-387-09686-5_7}{Volume 149 of The IMA Volumes in Mathematics and its Applications, pages 157--270}.
\newblock Springer, New York, NY~(2009).

\bibitem{helton_positivstellensatz_2004}
J.~W. Helton and Scott McCullough.
\newblock ``A {Positivstellensatz} for non-commutative polynomials''.
\newblock \href{https://dx.doi.org/10.1090/s0002-9947-04-03433-6}{Transactions of the American Mathematical Society {\bf 356}, 3721--3737}~(2004).

\bibitem{moment}
Andrew~C. Doherty, Yeong-Cherng Liang, Ben Toner, and Stephanie Wehner.
\newblock ``The quantum moment problem and bounds on entangled multi-prover games''.
\newblock In Proceedings of the 2008 IEEE 23rd Annual Conference on Computational Complexity.
\newblock \href{https://dx.doi.org/10.1109/CCC.2008.26}{Pages 199--210}.
\newblock CCC '08USA~(2008). IEEE Computer Society.

\bibitem{burgdorf_optimization_2016}
Sabine Burgdorf, Igor Klep, and Janez Povh.
\newblock ``Optimization of {Polynomials} in {Non}-{Commuting} {Variables}''.
\newblock \href{https://dx.doi.org/10.1007/978-3-319-33338-0}{{SpringerBriefs} in {Mathematics}}. Springer International Publishing. Cham~(2016).

\bibitem{magron2023sparse}
Victor Magron and Jie Wang.
\newblock ``Sparse polynomial optimization: theory and practice''.
\newblock \href{https://dx.doi.org/10.1142/q0382}{World Scientific}. ~(2023).

\bibitem{dps}
Andrew~C. Doherty, Pablo~A. Parrilo, and Federico~M. Spedalieri.
\newblock ``Complete family of separability criteria''.
\newblock \href{https://dx.doi.org/10.1103/PhysRevA.69.022308}{Phys. Rev. A {\bf 69}, 022308}~(2004).

\bibitem{marginal}
Xiao-Dong Yu, Timo Simnacher, Nikolai Wyderka, H.~Chau Nguyen, and Otfried G\"uhne.
\newblock ``A complete hierarchy for the pure state marginal problem in quantum mechanics''.
\newblock \href{https://dx.doi.org/10.1038/s41467-020-20799-5}{Nature Communications {\bf 12}, 1012}~(2021).

\bibitem{meas_tomography}
Marco Cattaneo, Matteo A.~C. Rossi, Keijo Korhonen, Elsi-Mari Borrelli, Guillermo Garc\'{\i}a-P\'erez, Zolt\'an Zimbor\'as, and Daniel Cavalcanti.
\newblock ``Self-consistent quantum measurement tomography based on semidefinite programming''.
\newblock \href{https://dx.doi.org/10.1103/PhysRevResearch.5.033154}{Phys. Rev. Res. {\bf 5}, 033154}~(2023).

\bibitem{qinflation}
Alejandro Pozas-Kerstjens, Rafael Rabelo, \L{}ukasz Rudnicki, Rafael Chaves, Daniel Cavalcanti, Miguel Navascu\'es, and Antonio Ac\'{\i}n.
\newblock ``Bounding the sets of classical and quantum correlations in networks''.
\newblock \href{https://dx.doi.org/10.1103/PhysRevLett.123.140503}{Phys. Rev. Lett. {\bf 123}, 140503}~(2019).

\bibitem{sdp_exact}
Hirokazu Anai.
\newblock ``On solving semidefinite programming by quantifier elimination''.
\newblock RIMS K{\^o}ky{\^u}roku {\bf 1038}, 154--162~(1998).
\newblock  url:~\url{https://www.kurims.kyoto-u.ac.jp/~kyodo/kokyuroku/contents/pdf/1038-21.pdf}.

\bibitem{quant_elim}
Christopher~W. Brown and James~H. Davenport.
\newblock ``The complexity of quantifier elimination and cylindrical algebraic decomposition''.
\newblock In Proceedings of the 2007 international symposium on {Symbolic} and algebraic computation.
\newblock \href{https://dx.doi.org/10.1145/1277548.1277557}{Pages 54--60}.
\newblock Waterloo Ontario Canada~(2007). ACM.

\bibitem{sdp_poly}
Yurii Nesterov and Arkadii Nemirovskii.
\newblock ``Interior-point polynomial algorithms in convex programming''.
\newblock \href{https://dx.doi.org/10.1137/1.9781611970791}{Number~13 in SIAM Studies in Applied and Numerical Mathematics}. SIAM. Philadelphia, Pa~(1994).

\bibitem{magron2019algorithms}
Victor Magron, Mohab Safey El~Din, and Markus Schweighofer.
\newblock ``Algorithms for weighted sum of squares decomposition of non-negative univariate polynomials''.
\newblock \href{https://dx.doi.org/10.1016/j.jsc.2018.06.005}{Journal of Symbolic Computation {\bf 93}, 200--220}~(2019).

\bibitem{magron2022exact}
Victor Magron, Mohab Safey El~Din, Markus Schweighofer, and Trung~Hieu Vu.
\newblock ``Exact {SOHS} decompositions of trigonometric univariate polynomials with {Gaussian} coefficients''.
\newblock In Proceedings of the 2022 International Symposium on Symbolic and Algebraic Computation.
\newblock \href{https://dx.doi.org/10.1145/3476446.3535480}{Pages 325--332}.
\newblock ~(2022).

\bibitem{koprowski2023pourchet}
Przemys{\l}aw Koprowski, Victor Magron, and Tristan Vaccon.
\newblock ``Pourchet's theorem in action: decomposing univariate nonnegative polynomials as sums of five squares''.
\newblock In Proceedings of the 2023 International Symposium on Symbolic and Algebraic Computation.
\newblock \href{https://dx.doi.org/10.1145/3597066.3597072}{Pages 425--433}.
\newblock ~(2023).

\bibitem{peyrl2008computing}
Helfried Peyrl and Pablo~A Parrilo.
\newblock ``Computing sum of squares decompositions with rational coefficients''.
\newblock \href{https://dx.doi.org/10.1016/j.tcs.2008.09.025}{Theoretical Computer Science {\bf 409}, 269--281}~(2008).

\bibitem{kaltofen_exact_2012}
Erich~L. Kaltofen, Bin Li, Zhengfeng Yang, and Lihong Zhi.
\newblock ``Exact certification in global polynomial optimization via sums-of-squares of rational functions with rational coefficients''.
\newblock \href{https://dx.doi.org/10.1016/j.jsc.2011.08.002}{Journal of Symbolic Computation {\bf 47}, 1--15}~(2012).

\bibitem{magron2021exact}
Victor Magron and Mohab Safey El~Din.
\newblock ``On exact {Reznick}, {Hilbert}-{Artin} and {Putinar}'s representations''.
\newblock \href{https://dx.doi.org/10.1016/j.jsc.2021.03.005}{Journal of Symbolic Computation {\bf 107}, 221--250}~(2021).

\bibitem{magron2018realcertify}
Victor Magron and Mohab Safey El~Din.
\newblock ``{RealCertify}: a {Maple} package for certifying non-negativity''.
\newblock \href{https://dx.doi.org/10.1145/3282678.3282681}{ACM Communications in Computer Algebra {\bf 52}, 34--37}~(2018).

\bibitem{jfr14}
V.~Magron, X.~Allamigeon, S.~Gaubert, and B.~Werner.
\newblock ``{Formal proofs for Nonlinear Optimization}''.
\newblock \href{https://dx.doi.org/10.6092/issn.1972-5787/4319}{Journal of Formalized Reasoning {\bf 8}, 1--24}~(2015).
\newblock  \href{http://arxiv.org/abs/1404.7282}{arXiv:1404.7282}.

\bibitem{certified_bounds_acopf}
Antoine Oustry, Claudia d'Ambrosio, Leo Liberti, and Manuel Ruiz.
\newblock ``Certified and accurate {SDP} bounds for the {ACOPF} problem''.
\newblock \href{https://dx.doi.org/10.1016/j.epsr.2022.108278}{Electric Power Systems Research {\bf 212}, 108278}~(2022).

\bibitem{PhysRevLett.127.200504}
Jessica Bavaresco, Mio Murao, and Marco~T\'ulio Quintino.
\newblock ``Strict hierarchy between parallel, sequential, and indefinite-causal-order strategies for channel discrimination''.
\newblock \href{https://dx.doi.org/10.1103/PhysRevLett.127.200504}{Phys. Rev. Lett. {\bf 127}, 200504}~(2021).

\bibitem{magron2023sum}
Victor Magron, Mohab Safey El~Din, and Trung-Hieu Vu.
\newblock ``Sum of squares decompositions of polynomials over their gradient ideals with rational coefficients''.
\newblock \href{https://dx.doi.org/10.1137/21M1436245}{SIAM Journal on Optimization {\bf 33}, 63--88}~(2023).

\bibitem{baldi2024effective}
Lorenzo Baldi, Teresa Krick, and Bernard Mourrain.
\newblock ``An effective positivstellensatz over the rational numbers for finite semialgebraic sets''.
\newblock \href{https://dx.doi.org/10.1090/mcom/4153}{Mathematics of Computation}~(2025).

\bibitem{leijenhorst2024solving}
Nando Leijenhorst and David de~Laat.
\newblock ``Solving clustered low-rank semidefinite programs arising from polynomial optimization''.
\newblock \href{https://dx.doi.org/10.1007/s12532-024-00264-w}{Mathematical Programming Computation {\bf 16}, 503--534}~(2024).

\bibitem{sdpa-gmp}
Maho Nakata.
\newblock ``A numerical evaluation of highly accurate multiple-precision arithmetic version of semidefinite programming solver: {SDPA-GMP, -QD and -DD}''.
\newblock In Proceedings of the IEEE International Symposium on Computer-Aided Control System Design.
\newblock \href{https://dx.doi.org/10.1109/CACSD.2010.5612693}{Pages 29--34}.
\newblock ~(2010).

\bibitem{rational_sohs}
Kristijan Cafuta, Igor Klep, and Janez Povh.
\newblock ``Rational sums of hermitian squares of free noncommutative polynomials''.
\newblock \href{https://dx.doi.org/10.26493/1855-3974.518.768}{Ars Mathematica Contemporanea {\bf 9}, 243--259}~(2015).

\bibitem{arblib}
Fredrik Johansson.
\newblock ``{Arb}: efficient arbitrary-precision midpoint-radius interval arithmetic''.
\newblock \href{https://dx.doi.org/10.1109/TC.2017.2690633}{IEEE Transactions on Computers {\bf 66}, 1281--1292}~(2017).

\bibitem{mosek}
{MOSEK ApS}.
\newblock ``{MOSEK} optimizer api for {Julia} 11.0.30''.
\newblock ~(2025).
\newblock  url:~\url{https://docs.mosek.com/11.0/juliaapi/index.html}.

\bibitem{CHSH}
John~F. Clauser, Michael~A. Horne, Abner Shimony, and Richard~A. Holt.
\newblock ``Proposed experiment to test local hidden-variable theories''.
\newblock \href{https://dx.doi.org/10.1103/PhysRevLett.23.880}{Phys. Rev. Lett. {\bf 23}, 880--884}~(1969).

\bibitem{a89}
Tsuyoshi Ito, Hiroshi Imai, and David Avis.
\newblock ``{Bell} inequalities stronger than the clauser-horne-shimony-holt inequality for three-level isotropic states''.
\newblock \href{https://dx.doi.org/10.1103/PhysRevA.73.042109}{Phys. Rev. A {\bf 73}, 042109}~(2006).

\bibitem{a89_optimality}
Timotej Hrga, Igor Klep, and Janez Povh.
\newblock ``Certifying {Optimality} of {Bell} {Inequality} {Violations}: {Noncommutative} {Polynomial} {Optimization} through {Semidefinite} {Programming} and {Local} {Optimization}''.
\newblock \href{https://dx.doi.org/10.1137/22M1473340}{SIAM Journal on Optimization {\bf 34}, 1341--1373}~(2024).

\bibitem{a88_exact}
K{\'a}roly~F P{\'a}l and Tam{\'a}s V{\'e}rtesi.
\newblock ``Quantum {Bounds} on {Bell} inequalities''.
\newblock \href{https://dx.doi.org/10.1103/PhysRevA.79.022120}{Physical Review A {\bf 79}, 022120}~(2009).

\bibitem{a89_exact}
K\'aroly~F. P\'al and Tam\'as V\'ertesi.
\newblock ``Maximal violation of a bipartite three-setting, two-outcome {Bell} inequality using infinite-dimensional quantum systems''.
\newblock \href{https://dx.doi.org/10.1103/PhysRevA.82.022116}{Phys. Rev. A {\bf 82}, 022116}~(2010).

\bibitem{tilted_bell}
Antonio Ac\'{\i}n, Serge Massar, and Stefano Pironio.
\newblock ``Randomness versus nonlocality and entanglement''.
\newblock \href{https://dx.doi.org/10.1103/PhysRevLett.108.100402}{Phys. Rev. Lett. {\bf 108}, 100402}~(2012).

\bibitem{tilted_selftesting}
Nicolas Gigena, Ekta Panwar, Giovanni Scala, Mateus Ara{\'u}jo, M{\'a}t{\'e} Farkas, and Anubhav Chaturvedi.
\newblock ``Self-testing tilted strategies for maximal loophole-free nonlocality''.
\newblock \href{https://dx.doi.org/10.1038/s41534-025-01029-6}{npj Quantum Information {\bf 11}, 82}~(2025).

\bibitem{exact_diagonalization}
Alexander Wei{\ss}e and Holger Fehske.
\newblock ``Exact {Diagonalization} {Techniques}''.
\newblock In Computational {Many}-{Particle} {Physics}.
\newblock \href{https://dx.doi.org/10.1007/978-3-540-74686-7_18}{Pages 529--544}.
\newblock Springer Berlin Heidelberg~(2008).

\bibitem{variatonal_benchmarks}
Dian Wu, Riccardo Rossi, Filippo Vicentini, Nikita Astrakhantsev, Federico Becca, Xiaodong Cao, Juan Carrasquilla, Francesco Ferrari, Antoine Georges, Mohamed Hibat-Allah, Masatoshi Imada, Andreas~M. L{\"{a}}uchli, Guglielmo Mazzola, Antonio Mezzacapo, Andrew Millis, Javier Robledo~Moreno, Titus Neupert, Yusuke Nomura, Jannes Nys, Olivier Parcollet, Rico Pohle, Imelda Romero, Michael Schmid, J.~Maxwell Silvester, Sandro Sorella, Luca~F. Tocchio, Lei Wang, Steven~R. White, Alexander Wietek, Qi~Yang, Yiqi Yang, Shiwei Zhang, and Giuseppe Carleo.
\newblock ``Variational benchmarks for quantum many-body problems''.
\newblock \href{https://dx.doi.org/10.1126/science.adg9774}{Science {\bf 386}, 296--301}~(2024).

\bibitem{tensor_rmp}
J.~Ignacio Cirac, David P\'erez-Garc\'{\i}a, Norbert Schuch, and Frank Verstraete.
\newblock ``Matrix product states and projected entangled pair states: Concepts, symmetries, theorems''.
\newblock \href{https://dx.doi.org/10.1103/RevModPhys.93.045003}{Rev. Mod. Phys. {\bf 93}, 045003}~(2021).

\bibitem{wang2026scalable}
Jie Wang, David Jansen, Ir{\'e}n{\'e}e Fr{\'e}rot, Marc-Olivier Renou, Victor Magron, and Antonio Ac{\'\i}n.
\newblock ``Scalable ground-state certification of quantum spin systems via structured noncommutative polynomial optimization''.
\newblock \href{https://dx.doi.org/10.21468/SciPostPhys.21.3.074}{SciPost Physics {\bf 21}, 074}~(2026).
\newblock  \href{http://arxiv.org/abs/2604.01555}{arXiv:2604.01555}.

\bibitem{degenerate_groundstate}
Chanchal~K. Majumdar and Dipan~K. Ghosh.
\newblock ``{On Next-Nearest-Neighbor Interaction in Linear Chain. I}''.
\newblock \href{https://dx.doi.org/10.1063/1.1664978}{J. Math. Phys. {\bf 10}, 1388--1398}~(1969).

\bibitem{sdp_complexity}
Aharon Ben-Tal and Arkadi Nemirovski.
\newblock ``Lectures on modern convex optimization: Analysis, algorithms, and engineering applications''.
\newblock \href{https://dx.doi.org/10.1137/1.9780898718829}{Society for Industrial and Applied Mathematics}. ~(2001).

\end{thebibliography}
\end{document}